\documentclass{article}
\usepackage[utf8]{inputenc}

    \usepackage[dvipsnames]{xcolor}
    \usepackage[T1]{fontenc}
	\usepackage[utf8]{inputenc}
	\usepackage[english]{babel}
	\usepackage{amsfonts}
	\usepackage{amssymb}
	    \usepackage{amsmath}
	
	\usepackage{enumerate}
	\usepackage{bbm}
	\usepackage{mathtools}
	\usepackage{subfig}
	\usepackage{booktabs}	
	\usepackage{setspace}
	\usepackage{graphicx}
	\usepackage{fullpage}
	
	\usepackage{soul}
	\usepackage{xspace}
	\usepackage{csquotes}
	\usepackage[margin = 1in]{geometry}
	\usepackage{placeins}
	\usepackage[footfont=normalsize,font=normalsize]{floatrow}
	\usepackage{lscape}
	\usepackage{longtable}
	\usepackage{mathabx}
    \usepackage{accents}
    \usepackage{float}

    \usepackage{tikz}
    \usetikzlibrary{calc,intersections}

	\usepackage[longnamesfirst]{natbib}
	\usepackage{bibunits}
	\defaultbibliographystyle{aer}

	\usepackage{hyperref}
	\hypersetup{
    colorlinks=true,
    linkcolor=Red,
    filecolor=magenta,      
    urlcolor=blue,
    citecolor = Blue
}
	\usepackage{cleveref}

    \usepackage{rotating}
	\usepackage[framemethod=tikz]{mdframed}
	\usepackage{todonotes}
	\presetkeys{todonotes}{color=black!5,inline}{} %

    \usepackage{tcolorbox}
    \newtcbox{\feedback}{nobeforeafter,colframe=black,colback=white,boxrule=0.5pt,arc=2pt,
      boxsep=0pt,left=2pt,right=2pt,top=2pt,bottom=2pt,tcbox raise base}
      
    \usepackage{amsthm}

    \newtheorem{asm}{Assumption}
    
    \newtheorem{prop}{Proposition}[section]

    \newtheorem{lem}{Lemma}[section]
    \newtheorem{cor}{Corollary}[section]
    \theoremstyle{definition}
    \newtheorem{rem}{Remark}
    
    \newtheorem{example}{Example}

    \crefname{prop}{Proposition}{Propositions}
\crefname{assumption}{assumption}{assumptions}
\crefname{asm}{assumption}{assumptions}
\crefname{prop}{proposition}{propositions}
\crefname{cor}{Corollary}{Corollaries}
\crefname{lem}{Lemma}{Lemmas}
\crefname{rem}{Remark}{Remarks}

\usepackage{array}
\newcolumntype{L}[1]{>{\raggedright\let\newline\\\arraybackslash}m{#1}}
\newcolumntype{C}[1]{>{\centering\let\newline\\\arraybackslash\hspace{0pt}}m{#1}}
\newcolumntype{R}[1]{>{\raggedleft\let\newline\\\arraybackslash\hspace{0pt}}m{#1}}

\usepackage{ragged2e}
\newlength\ubwidth

\newcommand\numberthis{\addtocounter{equation}{1}\tag{\theequation}}

\let\emptyset\varnothing

\newcommand{\indep}{\perp \!\!\! \perp}

	\renewcommand{\P}{\mathop{}\!\textnormal{P}}
	\newcommand{\E}{\operatorname{E}}
	
	\newcommand{\Cov}{\operatorname{Cov}}
	
	\newcommand{\Var}{\operatorname{Var}}

	\newcommand{\parens}[1]{\left(#1\right)}

	\DeclareMathOperator{\tr}{tr}
	\newcommand{\reals}{\mathbb{R}}

    \newcommand{\VhatM}{\widehat{V}_M}
    \newcommand{\VhatQ}{\widehat{V}_Q}
    \newcommand{\VhatCC}{\widehat{V}_{\mathrm{CC}}}
    \newcommand{\VhatKwithin}{\widehat{V}^{\tilde{\mathcal{K}}}_{\text{within}}}
    \newcommand{\VhatCCbold}{\widehat{\mathbf{V}}_{\mathrm{CC}}}

    \newcommand{\VhatAS}{\widehat{V}_{\mathrm{AS}}}
    \newcommand{\VhatAScollapsed}{\widehat{V}_{\mathrm{AS}}^{\mathrm{collapsed}}}
    \newcommand{\VhatNeyman}{\widehat{V}_{\mathrm{Neyman}}}
    \newcommand{\VhatNeymanCl}{\widehat{V}_{\mathrm{Neyman,cl}}}
    \newcommand{\VhatWithin}{\widehat{V}_{\mathrm{within}}}
    \newcommand{\DeltaCC}{\Delta_{\mathrm{CC}}}
    \newcommand{\DeltaAS}{\Delta_{\mathrm{AS}}}

\newcommand{\bernoulli}[1]{\mathrm{Bernoulli}\parens{#1}}

\newcommand{\twovec}[2]{\parens{\begin{array}{c} #1 \\ #2 \end{array}}}

\newcommand{\twobytwomat}[4]{\parens{\begin{array}{cc} #1 & #2 \\ #3 & #4 \end{array}}}

\newcommand{\VhatHAC}{\widehat{V}_{\mathrm{HAC}}}
\newcommand{\KHAC}[1][]{K_{\mathrm{HAC}#1}}
\newcommand{\AHAC}{A_{\mathrm{HAC}}}

\newcommand{\VhatHACPD}{\widehat{V}_{\mathrm{HACPD}}}
\newcommand{\KHACPD}[1][]{K_{\mathrm{HACPD}#1}}

\title{Simple Standard Errors for \\ Complicated (Quasi-)Experimental Designs\thanks{Klosin: UC Davis (\href{mailto:sylviaklosin@gmail.com}{sklosin@ucdavis.edu}). Roth: Brown
		University (\href{mailto:jonathanroth@brown.edu}{jonathanroth@brown.edu}). We are grateful to Haoge Chang, Kevin Chen, Chris Harshaw, Peter Hull,  Lihua Lei, Michael Pollmann, Stefan Wager, and seminar participants at BI Norwegian Business School, Columbia, and NHH for helpful conversations and suggestions. Roth gratefully acknowledges funding from the NIH under grant NIGMS 1R35GM155224 and from the Alfred P. Sloan Foundation. Various versions of Codex provided superb research assistance with coding and proofs. All remaining errors are our own.}}

\author{Sylvia Klosin \and Jonathan Roth}
\begin{document}
	\maketitle 
	
	\begin{abstract}
		We develop standard errors for estimates of causal effects in settings where the distribution of treatment exposures is known but potentially exhibits complicated dependence patterns across units. Such dependence arises, for example, in settings with spillovers/interference across units or when the treatment assignment is spatially correlated. Our variance estimator generalizes several properties of the Neyman variance estimator for completely randomized experiments to general designs: it is unbiased for an upper bound on the variance from the design-based perspective; the conservativeness of this upper bound depends on the heterogeneity of treatment effects; and it is also valid under certain forms of super-population sampling of the potential outcomes. Our estimator reduces to standard approaches in simple designs, such as the Neyman variance estimator in completely randomized experiments and the cluster-robust variance estimator in cluster-randomized experiments.  We show that the proposed estimator is admissible among the class of quadratic-form variance estimators, and derive asymptotic conditions under which it is consistent for an upper bound on the asymptotic variance of the treatment effect estimator. In simulations and applications based on two empirical studies, one with peer effects and another with a spatially correlated treatment, we find that the proposed variance estimator performs well, and in some cases yields standard errors that are 2-3 times shorter than those of existing alternatives.

	\end{abstract}
	
	\section{Introduction}
	Social scientists often study settings where the distribution of treatment assignments is assumed to be known, either because it was explicitly randomized or because it was determined by some quasi-random process resembling randomization. In such cases it is natural to leverage knowledge of the assignment scheme to obtain valid inference on causal effects without imposing strong assumptions about how the potential outcomes were generated. This motivates so-called \emph{design-based} inference, which conditions on the potential outcomes and leverages only the randomness of treatment assignment, and therefore is valid regardless of how the potential outcomes were generated. One approach to design-based inference is to conduct Fisher randomization tests (a.k.a. permutation tests), which are valid under arbitrary treatment assignment schemes for testing sharp null hypotheses (e.g. no effects of treatment for any unit). Yet, researchers often wish to obtain standard errors that allow them to test hypotheses and form confidence intervals for (possibly-weighted) average treatment effects. A challenge in the design-based setting is that the variance of the treatment effect estimator typically depends on products of potential outcomes that are not simultaneously observable, and therefore it is not possible to unbiasedly (or consistently) estimate its variance. It is often possible, however, to obtain valid standard errors that estimate an upper bound on the variance of the estimator. 
	
	For simple treatment assignment schemes, there are well-accepted approaches for obtaining valid design-based standard errors. In a completely randomized experiment, for example, the most common approach to obtain standard errors for the difference-in-means estimator is to use the \citet{neyman_application_1923} variance estimator, $\frac{1}{N_1} \hat{s}_1^2 + \frac{1}{N_0} \hat{s}_0^2$, where $N_d$ denotes the number of observations in arm $d$ and $\hat{s}_d^2$ is the sample variance of the outcome in arm $d$. This is equivalent to the heteroskedasticity-robust (HC2) variance estimator applied to a regression of the outcome on a treatment dummy. The Neyman variance estimator has three well-known properties (see, e.g., \citet{imbens_causal_2015}, for a textbook treatment). First, from the design-based perspective, it is unbiased for an upper bound on the estimator’s variance. Second, the conservativeness of the upper bound is proportional to the variance of treatment effects, and hence the Neyman variance is exact under treatment effect homogeneity. And third, the Neyman variance corresponds to the standard variance estimator if, in contrast to the design-based paradigm, one imagines the units as having been sampled iid from an infinite-superpopulation. The Neyman variance is thus valid from both the design-based and super-population sampling perspectives. This approach to inference extends naturally to a cluster-randomized experiment, where it is common to use the Neyman variance applied to cluster-level sums (or equivalently, cluster-robust standard errors). The resulting variance estimator is valid from the design-based perspective, with its conservativeness depending on the heterogeneity of treatment effects across clusters; it is also valid from the sampling perspective if we imagine the clusters as drawn independently from a super-population of clusters. 
	
	However, in many social-science applications, treatment assignments are correlated across units in more complex ways than in these simple designs. One example where treatment assignment exhibits complex dependence patterns is when treatments occur at a point in space. For example, the researcher may be interested in the effect of having a nearby grocery store on foot traffic to restaurants \citep{pollmann_causal_2026}, or the impact of having a nearby factory on economic outcomes \citep{dell_development_2020}. In such settings, nearby locations will tend to have correlated treatment assignments. A second example is settings in which there are spillovers across units. For example, assigning one student to an anti-bullying intervention as in \citet{paluck_changing_2016} may have impacts on that student's friends. In such settings, it is common to specify an exposure mapping that formalizes how one unit's outcome depends on other units' treatments \citep{aronow_estimating_2017}; e.g. we may specify that $i$'s outcome depends both on whether $i$ is treated and on whether any of $i$'s friends are treated. Even if the original treatment is assigned independently across units, the distribution of treatment \emph{exposures} will then exhibit dependence among units who are connected in the network. Correlated treatment assignments arise in a variety of other settings, such as when the treatment is determined by a centralized assignment mechanism (e.g. assignment of students to schools), when experiments are conducted in two-sided markets, or when an experiment uses a complex assignment mechanism involving re-randomization, stratification, or sequential treatment assignment. More generally, researchers frequently conduct permutation tests that simulate counterfactual treatment assignments by permuting observed treatments along a particular dimension, inducing a known but dependent assignment scheme. 
    
    This raises the question: can we construct a variance estimator for settings with complicated assignment schemes that has similar properties to the Neyman variance estimator in completely randomized experiments? In this paper, we propose a generalization of the Neyman variance estimator for completely randomized experiments that is applicable under complex assignment schemes. The estimator reduces to the Neyman variance estimator when applied to a completely randomized experiment and to a cluster-robust variance estimator when applied to a cluster-randomized experiment. Moreover, the estimator possesses similar properties to the Neyman variance estimator under more complex designs: it is valid from the design-based perspective, its conservativeness depends on treatment-effect heterogeneity, and it is also valid under certain superpopulation sampling schemes. 
    
	We now describe the setting and proposed variance estimator in more detail. At a high level, we suppose that there is a finite population of $K$ units each with an attribute $Y_k$, and we are interested in estimating the average $\tau = \frac{1}{K} \sum_k Y_k$ from a sample wherein only the attributes for individuals with $W_k = 1$ are observed. As described in \Cref{sec:setup}, this deceptively simple setup nests a large number of causal settings by setting $k$ to index the relevant potential outcomes, setting $Y_k$ to be an appropriate transformation of the $k$th potential outcome, and setting $W_k$ to be an indicator for whether the relevant potential outcome is observed. For example, if the treatment exposure $D_i$ can take values $\{0,\ldots, Q\}$, then we can let $k=(i,d)$ index the $K=N(Q+1)$ combinations of $(i,d)$. Setting $Y_k = \omega_{i,d} Y_i(d)$ for some chosen weights $\omega_{i,d}$ then yields the causal contrast $\tau = \frac{1}{K} \sum_i \sum_d \omega_{i,d} Y_i(d)$. We then observe $Y_k$ if and only if $W_k = 1\{D_i =d\}$ is equal to 1. The framework thus nests estimation of average causal contrasts in settings where the treatment exposure takes on a finite number of values. Beyond causal settings, this set-up also nests the problem of learning about the mean of an attribute in a finite population under a complicated survey sampling scheme. 
	
	Within this framework, we focus on the problem of estimating the variance of the Horvitz-Thompson (inverse-probability weighting) estimator, $\hat\tau = \frac{1}{K} \sum_k \tfrac{W_k}{\pi_k} Y_k$, where $\pi_k = P(W_k =1)$. (As we describe below, a similar approach extends to other estimators that have an asymptotic linear representation, such as those obtained by linear regressions.) We adopt a design-based approach wherein the $Y_k$ are viewed as fixed and the variation in the data comes from the stochastic nature of the assignment $W_k$. We refer to the distribution of $W_k$ as the \emph{assignment} or \emph{design} distribution, and we assume that it is known, so that $\pi_k$ and $\pi_{kk'} = P(W_k=1, W_{k'}=1)$ are both known. As first shown in \citet{horvitz_generalization_1952}, the variance of $\hat\tau$ over the design distribution is given by $\Var(\hat\tau) = K^{-2} \sum_{k,k'} v_{kk'} Y_kY_{k'}$, where $v_{kk'} = \Cov(\tfrac{W_k}{\pi_k}, \tfrac{W_{k'}}{\pi_{k'}} )$ is a function of the distribution of $W_k$ and hence is known. If $Y_k$ and $Y_{k'}$ are both observed with positive probability, so that $\pi_{kk'}= P(W_k=1,W_{k'}=1)>0$, then it is straightforward to unbiasedly estimate $Y_kY_{k'}$ by inverse probability weighting, $\frac{W_k W_{k'}}{\pi_{kk'}} Y_k Y_{k'}$. 	The challenge for variance estimation is that certain pairs of potential outcomes are not jointly observable, so that $\pi_{kk'}=0$ for some pairs $k,k'$. For example, if $k$ and $k'$ correspond to different potential outcomes for the same unit, say $Y_i(0)$ and $Y_i(1)$, then it is impossible to observe both $W_k=1$ and $W_{k'}=1$ simultaneously. Under complex assignment schemes, other pairs of potential outcomes may also be jointly unobservable. For example, if the treatment has spillovers and units $i$ and $j$ are neighbors, then we cannot observe that unit $i$ is treated while unit $j$ has no treated neighbors. To derive valid standard errors, we therefore need to obtain some upper bound on the terms involving products $Y_k Y_{k'}$ for pairs of potential outcomes that are jointly unobservable.
    
    Our approach to constructing an upper bound on the variance is based on what we call the \emph{joint unobservability graph}. This is defined to be the graph among the nodes $k=1,\ldots,K$, where $k$ and $k'$ are adjacent in the graph if $Y_{k}$ and $Y_{k'}$ cannot be jointly observed, i.e. $\pi_{kk'}=0$. We can then partition the joint unobservability graph into its connected components $\mathcal{K}_1,\ldots,\mathcal{K}_C$. This construction implies that if $\pi_{kk'}=0$, then $k,k'$ are in the same $\mathcal{K}_c$ for some $c$. For our main approach, we assume that there are at least two connected components, $C \geq 2$, although we consider an extension of the main idea to settings with $C=1$ in \Cref{sec:arbitrary-clusters}. In a completely randomized experiment, each connected component contains the two potential outcomes for a given unit $i$. Likewise, in a cluster-randomized experiment, each connected component collects all the potential outcomes for units in a particular cluster. We thus see that under these two simple assignment schemes, the standard approaches to variance estimation have conservativeness that is proportional to the variance of treatment effect sums across the connected components.	
	
	Our proposed estimator generalizes this property to more general designs: its conservativeness is proportional to the variance of treatment effect sums across connected components. Specifically, we consider the variance estimator 
	$$\VhatCC =  \frac{1}{K^2} \sum_{ \substack{ k,k' \\ \pi_{kk'} >0 }} (v_{kk'} + s_{kk'}) \frac{W_k W_{k'}}{ \pi_{kk'} } Y_k Y_{k'}, \text{ where } s_{kk'} = \begin{cases}1 & k,k' \in \mathcal{K}_c \text{ for some } c \\ 
		-1/(C-1) & \text{o.w.}
	\end{cases} $$
	which we refer to as the connected components (CC) estimator. The estimator $\VhatCC$ does inverse-probability weighting only over pairs $k,k'$ that are jointly observable. However, rather than using the weights $v_{kk'}$ found in the expression for $\Var(\hat\tau)$, it uses $v_{kk'}+s_{kk'}$, where $s_{kk'}$ is 1 if individuals are in the same connected component and $-1/(C-1)$ otherwise. The CC estimator thus compensates for the jointly unobservable pairs by adding the slack term $s_{kk'}$, which depends on whether $k$ and $k'$ are in the same component, to the observable pairs. We show that $\VhatCC$ matches the Neyman variance in a completely randomized experiment, and a cluster-robust variance estimator in a cluster-randomized completely randomized experiment. More generally, our first main result (\Cref{prop:cc-valid}) shows that $\VhatCC$ is weakly upwards-biased for $\Var(\hat\tau)$, with conservativeness proportional to the variance of treatment effects across clusters, $\Var_c(\tau_c)$, where $\tau_c = \sum_{k \in \mathcal{K}_c} Y_k$. It follows that $\VhatCC$ is unbiased if treatment effect sums are homogeneous across connected components (which necessarily occurs under the sharp null when all effects are zero). In addition to its design-based interpretation, we further show that $\VhatCC$ can be viewed as an unbiased estimator of $\Var(\hat\tau)$ under a sampling scheme in which the component sums $\tau_c$ are sampled iid across components. The CC estimator thus possesses the three main properties of the Neyman variance estimator for completely randomized experiments in more general designs.

	We further establish that $\VhatCC$ is admissible in the class of quadratic-form variance estimators (defined formally below), i.e. there is no quadratic-form variance estimator that is upwards-biased for $\Var(\hat\tau)$ and that dominates $\VhatCC$ in the sense that it is always weakly less conservative and strictly so for some choice of potential outcomes. Our proof of admissibility builds on the results of \citet{harshaw_optimized_2026}, who show that any admissible quadratic-form estimator corresponds to the solution of a class of optimization problems with an objective over matrices that is increasing in the Loewner order. We show that $\VhatCC$ is a solution to this optimization for a particular choice of objective, and hence is admissible. The objective that $\VhatCC$ optimizes also has a Bayesian interpretation. Specifically, consider a Bayesian who has a prior over the potential outcomes and wishes to minimize the expected conservativeness of the variance estimator, subject to the constraint that it is weakly upwards-biased for all possible potential outcomes. We show that $\VhatCC$ corresponds to the optimal solution (among quadratic-form estimators) for a Bayesian whose prior suggests that cross-component heterogeneity in treatment effects is small and ``equally spread out'' across components in a sense that we make precise below. Since $\VhatCC$ reduces to the Neyman variance estimator in a completely randomized experiment and the cluster-robust variance estimator in cluster-randomized experiments, this result also gives a novel Bayesian justification for these commonly-used variance estimators.
	
	In \Cref{sec:comparison-in-experiments}, we compare the CC estimator to other design-based estimators that have been proposed for general assignment schemes, focusing primarily on the estimators proposed by \citet{aronow_estimating_2017} and \citet{harshaw_optimized_2026}. When specialized to a completely randomized experiment, the CC estimator reduces to the Neyman variance, with conservativeness proportional to the variance of treatment effects, $ \frac{1}{N}\Var_i(\tau_i)$. By contrast, we show that the \citet{aronow_estimating_2017} estimator has conservativeness proportional to the uncentered second moment of treatment effects, $\frac{1}{N} E_i(\tau_i^2)$. Hence, we expect the CC estimator to be much less conservative when treatment effects are large but relatively homogeneous, so that $\Var_i(\tau_i) \ll E_i(\tau_i^2)$.\footnote{We note that $\Var_i(\tau_i)$ is defined to be the finite-population variance, with denominator $N-1$ rather than $N$, so it is not the case that $\Var_i(\tau_i) \leq E_i(\tau_i^2)$ for all choices of potential outcomes. Hence, the conservativeness of the \citet{aronow_estimating_2017} estimator can be slightly smaller than that of the CC estimator, but only by the factor $(N-1)/N$.} \citet{harshaw_optimized_2026} propose a class of estimators based on minimizing different measures of the slackness of the upper bound, and in their simulations consider implementations using simple norms such as the trace and Frobenius norms. In a completely randomized experiment with equal-sized treatment arms, we show that the \citet{harshaw_optimized_2026} estimator using the Frobenius norm objective corresponds exactly to the AS estimator, and hence the same comparison applies to this estimator. Turning to more general designs, any variance estimator corresponding to a solution to the class of optimizations considered in \citet{harshaw_optimized_2026} is admissible in the class of quadratic-form variance estimators, as is the CC estimator, and thus neither dominates the other for all choices of potential outcomes. The CC estimator is exact when treatment effects are homogeneous across components, while this is not guaranteed for the \citet{harshaw_optimized_2026} estimator based on the trace or Frobenius norms, and thus we expect the CC estimator to be less conservative when effects are relatively homogeneous across components, while admissibility implies that the trace and Frobenius-norm based estimators must be less conservative under at least some forms of treatment effect heterogeneity. We further show in \Cref{subsec:complex-designs} that all valid quadratic-form variance estimators that are exact under the sharp null of no treatment effects have conservativeness that depends on the heterogeneity of treatment effects across connected components. Specifically, the conservativeness of any such estimator is proportional to a quadratic form involving the component-level sums of treatment effects and a correlation matrix. The CC estimator corresponds to a particular choice of the correlation matrix, which we show has the minimum operator norm among correlation matrices corresponding to variance estimators that are exact under component homogeneity.\footnote{Estimators based on other correlation matrices, which optimize different objectives, are also possible; see \Cref{rem:alt-estimators}.}

	In \Cref{sec:consistency}, we consider finite-population asymptotics and provide conditions under which the CC estimator is consistent for an upper bound on the asymptotic variance of $\hat\tau$, employing similar arguments for the consistency of Horvitz-Thompson estimators as in \citet{delevoye_consistency_2020}.  In \Cref{sec:non-ht-estimators}, we show that our approach to variance estimation for Horvitz-Thompson estimators is also asymptotically valid for other estimators that have an asymptotic linear representation. We provide high-level conditions for consistency for generic asymptotically linear estimators, and lower-level conditions for estimators based on linear regressions. 
	
	In \Cref{sec:arbitrary-clusters}, we consider the setting where the joint unobservability graph is connected ($C=1$), in which case the CC estimator is not well-defined. However, we show that if the units in the population are partitioned into clusters, then using similar ideas it is possible to construct a variance estimator whose conservativeness intuitively depends on (a) the heterogeneity of treatment effects across the clusters, and (b) the number of cross-cluster pairs $k,k'$ that have $\pi_{kk'}=0$. The resulting estimator will thus not be too conservative if the clusters are relatively homogeneous and have relatively few joint unobservability links between them. Owing to term (b), the resulting estimator will not be exact even under the sharp null of no treatment effects. We show, however, that some conservativeness under the sharp null appears to be necessary in the challenging case when $C=1$: the only quadratic-form variance estimator that is exact under the sharp null is the trivial estimator $\hat\tau^2$, which always produces a $t$-statistic with absolute value of 1. 
	
	We find that the CC estimator performs well in calibrated simulations and applications to two applied settings, as described in \Cref{sec:sims-and-apps}. We first revisit and extend \citet{harshaw_optimized_2026}'s simulation study calibrated to \citet{paluck_changing_2016}, who studied an anti-bullying intervention in schools which is expected to have social spillovers. We find that confidence intervals based on the CC estimator have good coverage and are shorter on average than the \citet{aronow_estimating_2017} estimator and two implementations of the \citet{harshaw_optimized_2026} approach based on simple norms in all of the simulation specifications taken from \citet{harshaw_optimized_2026}. In line with the intuition from a completely randomized experiment, the differences tend to be larger when the treatment effect is large but relatively homogeneous. Applied to the actual \citet{paluck_changing_2016} data, the CC estimator yields the shortest standard error in 7 of 8 specifications; in the most extreme case, the CC SE is less than half that of the other estimators. The CC estimator can also be computed in seconds in all specifications, compared with over 30 hours in some cases for the \citet{harshaw_optimized_2026} approach, which relies on nonlinear optimization. Our second simulation design is based on the spatial treatment setting studied in \citet{pollmann_causal_2026}, who investigated the impact of having a nearby grocery store on foot traffic to restaurants. We calibrate our simulations based on a public-use synthetic dataset provided by \citet{pollmann_causal_2026}, which was simulated to match features of the proprietary data used in his actual application. We find in our simulations that the CC variance estimator yields confidence intervals with close to nominal coverage, and that are approximately three times shorter on average than those obtained using the variance estimator proposed by \citet{pollmann_causal_2026}. We likewise obtain standard errors that are approximately three times smaller when applied to Pollmann's dataset. We also find in our simulations that the CC estimator is competitive with spatial HAC \citep{conley_gmm_1999} using a design-motivated choice of kernel and bandwidth \citep{wang_design-based_2025} under homogeneous effects and simple forms of treatment effect heterogeneity, while remaining valid under non-homophilic forms of heterogeneity that cause HAC to be downward-biased. We emphasize that all these comparisons across estimators may be specific to the simulation designs considered, as theory generally suggests that there is no uniform dominance relationship across the estimators considered.

\paragraph{Related Literature.} The most closely related work is \citet{aronow_estimating_2017} and \citet{harshaw_optimized_2026}, discussed above, who also propose variance estimation methods for causal settings with known but potentially arbitrary assignment probabilities and a finite number of known treatment exposure levels. \citet{chang_design-based_2025} also studies complex experimental designs; his focus is mainly on estimation, although he also extends the \citet{aronow_estimating_2017} approach to variance estimation to a broader class of estimators. Our contribution relative to this work is to propose a variance estimator that generalizes the well-known properties of the Neyman variance estimator to more complicated designs. Another closely related paper is \citet{pollmann_causal_2026}, who considers a design-based setting with spatial treatments that, as we describe in \Cref{sec:sims-and-apps} below, can be cast as a special case of the setting we consider. 

Our work relates to the broader literatures on inference in settings with spatial treatments or network spillovers. One common approach is to use HAC-style variance estimators, such as \citet{conley_gmm_1999}. HAC-style estimators were originally developed in a model-based framework in which one restricts the dependence of errors over space; a challenge in implementing these estimators in practice is choosing the distance metric, bandwidth, and kernel, which together govern the allowed spatial correlation. \citet{xu_design-based_2026} study HAC-style estimators from the design-based perspective and provide general mixing conditions under which spatial HAC is asymptotically valid, but choosing a bandwidth and kernel for good finite-sample performance remains a practical challenge. Other recent work has provided a design-based justification for HAC-style variance estimators with particular choices of distance metric, bandwidth, and kernel in specific network \citep{leung_causal_2022} and spatial-treatment \citep{wang_design-based_2025} settings.\footnote{The appendices to \citet{harshaw_optimized_2026} and \citet{pollmann_causal_2026} respectively discuss how the frameworks of \citet{leung_causal_2022} and \citet{wang_design-based_2025} relate to our design-based framework. \citet{harshaw_optimized_2026}, for example, find that the \citet{leung_causal_2022} variance estimator can be anti-conservative in their simulation design, which we adopt in \Cref{sec:sims-and-apps} below. They also discuss how the \citet{leung_causal_2022} estimator may be preferable in other designs where the exposure mapping is mis-specified.} One advantage of these approaches relative to ours is that they typically can allow for some mis-specification of the exposure mapping (e.g. approximate network interference), whereas our framework relies on a correctly-specified exposure mapping; see \citet{savje_average_2021} for discussion of how inference can be distorted under mis-specified exposure mappings. On the other hand, these results typically rely on specific treatment assignment mechanisms (e.g. Bernoulli assignment of nodes or candidate sites) and require some restrictions on treatment effects, either in terms of how they decay with distance in the network, or that treatment effects satisfy a spatial homophily assumption. The CC estimator, by contrast, can accommodate complex treatment assignment mechanisms and does not require restrictions on treatment effect heterogeneity or the choice of a distance metric, bandwidth or kernel, although it does require that $C>1$ and may be conservative when treatment effects are very heterogeneous. We thus view the two approaches as complementary in spatial settings since they have non-nested use-cases; see \Cref{subsec:complex-designs} for further discussion. An alternative design-based approach in settings with spillovers is to use randomization tests, as in \citet{athey_exact_2018}, \citet{basse_randomization_2019}, \citet{zhong_unconditional_2024}, which yield finite-sample valid inference for (partially) sharp nulls in settings with interference; by contrast, we focus on developing standard errors for weak nulls on average effects. See \citet{bhadra_causal_2025} for a review of methods for inference in the presence of spillovers, including model-based approaches that make assumptions about the formation of the network and potential outcomes.

\section{Set-up\label{sec:setup}}
	
	For our formal results, we will consider the problem of estimating the mean of an attribute in a finite population under a general sampling scheme; as we will show shortly, this framework nests a variety of causal settings if each item in the finite population corresponds to a potential outcome of interest. Formally, we suppose there exists a finite population of $K$ items, where the $k$th item has attribute $Y_k$. These attributes are viewed as non-stochastic (or conditioned on). We are interested in the population average $\tau := \tfrac{1}{K} \sum_k Y_k$. We denote by $W_k$ the random variable indicating whether the $k$th item is observed. We denote the first-order and second-order sampling probabilities by $\pi_k := P(W_k =1)$ and $\pi_{kk'} := P(W_k = 1, W_{k'}=1)$, respectively, and assume that these probabilities are known (or simulatable to arbitrary accuracy). We assume that $\pi_k>0$ for all $k$, so that all attributes are sampled with positive probability.\footnote{When $k$ indexes potential outcomes, the requirement that $\pi_k>0$ for all $k$ imposes the overlap restriction that each unit receives each treatment level with positive probability. Since the design is assumed to be known, the researcher can restrict attention to the sub-population for which this overlap condition holds. For ease of exposition, we assume that the researcher has already restricted attention to the population satisfying overlap.} We consider the Horvitz-Thompson estimator $\hat\tau = \frac{1}{K} \sum_k \frac{W_k}{\pi_k} Y_k$; see \Cref{sec:non-ht-estimators} below for extensions to other estimators. 
	
	We can nest a variety of causal settings in this framework by having $k$ index the relevant potential outcomes and setting $Y_k$ to be an appropriate transformation of the potential outcomes. Consider a setting where unit $i$ receives treatment exposure level $D_i \in \{0,\ldots, Q\}$, with realized outcome $Y_i = Y_i(D_i)$, where $Y_i(\cdot)$ is $i$'s potential outcome, which following the design-based paradigm we view as fixed (or conditioned on). If we let $k=(i,d)$ index the $K=N (Q+1)$ combinations of $i$ and $d$, and set $Y_k = \omega_{i,d}Y_i(d)$ for some weights $\omega_{i,d}$, then $\tau$ corresponds to the causal contrast $\frac{1}{K} \sum_{i} \sum_d \omega_{i,d} Y_i(d)$, and $\hat{\tau}$ corresponds to the Horvitz-Thompson estimator $\frac{1}{K} \sum_i \sum_d \frac{1\{ D_i =d \} }{  P(D_i =d)} \omega_{i,d} Y_i$. To be more concrete, we now outline several examples fitting this framework. 
	
	\begin{example}[Completely randomized experiment] \label{example:cre}
	Consider a completely randomized experiment (CRE) in which $N_1$ out of $N$ units are randomly selected to receive a binary treatment ($D_i \in \{0,1\}$) with equal probability. (We will assume throughout that $2 \leq N_1 \leq N-2$, so there are at least two units in each treatment arm.) Let $k$ index the possible values of $(i,d)$ for $i=1,\ldots,N$ and $d = 0,1$, so that $K= 2N$. For $k=(i,d)$, define $Y_k = 2Y_i(1)$ if $d=1$ and $-2 Y_i(0)$ if $d=0$, and let $W_k = 1\{D_i =d \}$. Then $$\tau = \frac{1}{K} \sum_k Y_k = \frac{1}{N} \sum_i Y_i(1) - Y_i(0), $$
	the finite-population average treatment effect. For $k=(i,d)$ and $k'=(i',d')$, we have $\pi_k = P(D_i =d)$ and $\pi_{kk'} = P(D_i=d, D_{i'} = d')$. For example, if $d=d' = 1$ and $i \neq i'$, then $\pi_k  = P(D_i=1) = \frac{N_1}{N}$ and $\pi_{kk'} = P(D_i = 1, D_{i'}= 1) = \frac{N_1}{N} \frac{N_1 - 1}{N-1}$. %
	\end{example}

	\begin{example}[Cluster randomized experiment] \label{example:cluster-experiment}
	As in the previous example, suppose there is a binary treatment. However, we now suppose that the $N$ units are divided into $G$ different clusters (or groups), and denote by $g(i) \in \{1,\ldots,G\}$ the cluster of unit $i$. For simplicity, we assume each cluster $g$ has exactly $m=N/G$ units. We suppose that we have a cluster randomized experiment in which $G_1$ out of the $G$ clusters are selected for treatment, and unit $i$ is treated if and only if $D_{g(i)}=1$, where $D_g$ is an indicator for whether cluster $g$ is selected for treatment. (We will assume throughout that $2 \leq G_1 \leq G-2$, so there are at least 2 clusters in each treatment arm.) We can define $k$, $Y_k$, and $W_k$ as in the previous example, so that the target parameter is still the ATE. The only change is that the $\pi_{k}, \pi_{kk'}$ now correspond to the treatment probabilities under \emph{clustered} assignment. For example, if $k = (i,1)$ and $k' = (i',1)$ for $i$ and $i'$ in the same cluster ($g(i)= g(i')$), then $\pi_{kk'}= P(D_{g(i)} =1, D_{g(i')}=1) = \frac{G_1}{G}$. Likewise, if $k=(i,1)$ and $k'=(i',0)$ for $i$ and $i'$ again in the same cluster, then $\pi_{kk'}=0$, since units in the same cluster cannot have different treatment statuses.
	
	\end{example} 	

	\begin{example}[Spatial treatments] \label{example:spatial}
	Similar to \citet{pollmann_causal_2026} and \citet{wang_design-based_2025}, suppose we have locations $i=1,\ldots,N$ that are located at various points on a map. We are interested in the effect of having a local amenity (or disamenity) nearby; for example, we may wish to know how the presence of a grocery store affects foot traffic, or how a polluting plant affects property values. We suppose that the sites of the amenities are randomly selected from one of $S$ possible sites, and let $T_s \in \{0,1\}$ be an indicator for whether site $s$ is selected. We may then say that location $i$ is treated if there is a closeby selected site, so that $D_i = 1\{ \exists s \text{ s.t. } T_s =1, d(s,i) \leq d_0 \}$, where $d(s,i)$ is the distance between $s$ and $i$ and $d_0$ is some threshold. (We could likewise define treatment to depend on the \emph{number} of closeby sites, or even the set of sites within distance $d_0$ of $i$.) We can again define $k$ and $Y_k$ as in the previous examples so that the target parameter is the ATE. The key change is that the $\pi_{k}$ and $\pi_{kk'}$ are now determined by the spatial nature of the treatment. For $k=(i,1)$ and $k'=(i',1)$, for example, $\pi_{kk'}$ corresponds to the probability that there is both a site close to $i$ that is treated and a site close to $i'$ that is treated. 
	
	\end{example}

	\begin{example}[Interference] \label{example:interference}
		In many settings, we may expect the treatment assignment of one unit to affect the outcome of another unit. For example, if $T_i$ denotes whether individual $i$ is randomized to receive a vaccine, and unit $j$ is a family member of unit $i$, then we may expect $T_i$ to affect the outcome for unit $j$. A popular framework for studying spillovers is the exposure mapping framework of \citet{aronow_estimating_2017}, in which we specify that unit $i$'s treatment depends on some function $h_i(T_i,T_{-i})$. In the vaccine context, we expect that unit $i$'s outcome depends on both whether they were vaccinated and whether any of their neighbors were vaccinated, in which case $h$ takes on four possible values (for each combination of own-treatment and neighbors' treatment). In this case, we can let $k$ index the relevant values for $(i,h)$ and let $Y_k = \omega_{i,h} Y_i(h)$, in which case $\tau$ estimates the contrast $\frac{1}{K} \sum_i \sum_h \omega_{i,h} Y_i(h)$. For example, if we're interested in the $(1,1)$ vs. $(0,0)$ contrast, then setting $\omega_{i,h}=2$ if $h=(1,1)$ and $-2$ if $h=(0,0)$ yields the average total effect, $\tau = \frac{1}{N} \sum_i Y_i(1,1) - Y_i(0,0)$. In this case, the $\pi_k$ and $\pi_{kk'}$ are determined both by the treatment assignment $T_i$ and by the network. 
	\end{example} 

    \paragraph{Other examples.} Although we focus on the four examples given above, we briefly sketch several other classes of applications in which the assignment mechanism is known but may exhibit dependence across units. First, a variety of quasi-random processes will induce treatments with spatial correlation, beyond the example with site-level treatments in \Cref{example:spatial}. For example, previous papers have argued that the timing of hurricanes is quasi-random \citep[e.g.][]{Deryugina2017}; the treatment assignments induced by randomly assigning dates to hurricane paths will exhibit complicated dependence owing to multiple places being hit by the same hurricane. Second, consider the literature that has exploited randomness in centralized assignment mechanisms that, say, assign students to schools or doctors to residencies \citep{abdulkadiroglu_research_2017,chen_nonparametric_2026}. It is straightforward to simulate from the assignment mechanism many times, and therefore to learn assignment probabilities to arbitrary accuracy. However, owing to capacity constraints, assignments are likely to exhibit complicated dependence across individuals. Third, consider experiments in two-sided markets---e.g. a randomized intervention to drivers on Uber or sellers on eBay \citep{johari_experimental_2021}. If we are interested in the impact of this intervention on the other side of the market---e.g. the impact on Uber riders or eBay buyers---then the exposures to treatment likely exhibit complicated dependence patterns since, for example, two riders who had the same driver will have correlated exposures. Fourth, in a variety of settings, researchers employ potentially complex treatment assignment protocols involving re-randomization, stratification, or sequential treatment assignment rules \citep[e.g.][]{li_asymptotic_2018,bai_optimality_2022}. Finally, we note that in a variety of settings, researchers run permutation tests that simulate counterfactual treatment exposures by permuting the observed treatment exposures along various dimensions: the assignment of treatments over this distribution is known, but may again exhibit complicated dependence across units. 
    
	With these examples in mind, we now describe the challenge of estimating the variance of $\hat\tau$. Since the $Y_k$ are fixed, as shown in \citet{horvitz_generalization_1952}, 
	\begin{equation} \Var(\hat\tau) = \frac{1}{K^2} \sum_{k,k'} v_{kk'} Y_k Y_{k'}, \qquad \text{ where } \qquad v_{kk'} := \underbracket{\frac{\pi_{kk'} - \pi_k \pi_{k'}}{\pi_k \pi_{k'}}}_{= \Cov\big(\frac{W_k}{\pi_k}, \frac{W_{k'}}{\pi_{k'}} \big)}. \label{eqn:var-tauhat}
	\end{equation}
	
	We can decompose this expression into sums involving pairs $(k,k')$ with $\pi_{kk'}>0$ and those with $\pi_{kk'} = 0$,
	
	$$ \Var(\hat\tau) = \frac{1}{K^2} \sum_{ \substack{ k,k' \\ \pi_{kk'} >0 }} v_{kk'} Y_k Y_{k'}  + \frac{1}{K^2} \sum_{ \substack{ k,k' \\ \pi_{kk'} =0 } } -Y_k Y_{k'}.$$
	
	The first sum is straightforward to estimate since it is a weighted sum of terms of the form $Y_k Y_{k'}$, where $Y_k$ and $Y_{k'}$ are simultaneously observable with positive probability ($\pi_{kk'}>0$). These terms can be unbiasedly estimated by Horvitz-Thompson weighting, $\E[ v_{kk'} \frac{W_k W_{k'} }{  \pi_{kk'} } Y_k Y_{k'} ]= v_{kk'} Y_k Y_{k'}$. The challenge arises from the second term, which involves products $-Y_k Y_{k'}$ which are never directly observed since $\pi_{kk'}=0$. The key challenge is to find an estimable upper bound on the sum of these unobservable terms.

	\section{Proposed estimator and its properties\label{sec:estimator}}

    To derive an upper bound on the jointly unobservable product terms that enter $\Var(\hat\tau)$, we first formalize the structure of which potential outcomes cannot be observed together using what we call the \emph{joint unobservability graph}. We first define this graph, then show how it can be used to obtain a variance estimator. 
	
	\paragraph{Joint Unobservability Graph.} We say that $k$ and $k'$ are jointly unobservable if $\pi_{kk'}=0$. We then consider the \emph{joint unobservability graph} in which the nodes $k$ and $k'$ are connected if and only if $k,k'$ are jointly unobservable. This graph has adjacency matrix $A$ such that $A_{kk'} = 1$ if and only if $\pi_{kk'} =0$. We can then partition the indices $\{1,...,K\}$ into the connected components of the graph: i.e. the finest partition $\mathcal{K}_1,...,\mathcal{K}_C$ of $\{1,...,K\}$ such that if $\pi_{kk'} =0$, then $k,k' \in \mathcal{K}_c$ for some $c$. We will assume for now that $C \geq 2$, so the joint unobservability graph is not connected. See \Cref{sec:arbitrary-clusters} below for an extension to the setting where the joint unobservability graph may be connected but has ``weakly connected'' clusters. 
	
	\paragraph{Examples (revisited).} Revisiting our examples from earlier, we note that in a completely randomized experiment, $k$ and $k'$ are jointly unobservable if and only if they correspond to the potential outcomes for the same unit ($Y_i(0),Y_i(1)$). There are thus $N$ connected components of the graph, where each connected component contains the two potential outcomes corresponding to a single unit. In a cluster-randomized experiment, $k$ and $k'$ are jointly unobservable if either (i) they correspond to the two potential outcomes for a single unit $i$, or (ii) they correspond to potential outcomes $Y_i(1)$ and $Y_j(0)$ for $i$ and $j$ in the same cluster. It follows that the number of connected components equals the number of clusters ($G$), with each component collecting all the potential outcomes associated with a particular cluster $g$. Next consider our spatial treatment example (\Cref{example:spatial}). As before, the two potential outcomes for the same unit are jointly unobservable. Moreover, if every possible site $s$ that is within distance $d_0$ of unit $i$ is also within distance $d_0$ of unit $j$, then it is impossible to jointly observe $Y_i(1)$ and $Y_j(0)$. The joint unobservability graph thus contains edges generated from these two types of pairs of potential outcomes. The number of connected components then depends on the spatial orientation of the sites and units. Finally, consider the setting with interference (\Cref{example:interference}). Suppose the exposure mapping $h$ captures own-treatment and whether any neighbors were treated, every unit has at least one neighbor, and we are interested in the $(1,1)$ vs. $(0,0)$ contrast. Suppose the individual-level treatments $T_i$ are assigned by Bernoulli randomization. In this case, it turns out that each connected component of the joint unobservability graph collects all the potential outcomes for units in a given connected component of the original network. To see why, observe that as in our other examples, the two potential outcomes for the same unit $i$ are jointly unobservable. Moreover, $Y_i(1,1)$ and $Y_j(0,0)$ are not jointly observable if $i$ and $j$ are neighbors, since it is not possible for $i$ to be treated while $j$ has no treated neighbors. This implies that all the potential outcomes for units who are connected in the original network must be connected in the joint unobservability graph. On the other hand, it is straightforward to verify that all potential outcomes for units that are not connected in the original graph are jointly observable under Bernoulli assignment. Thus, the connected components of the joint unobservability graph are one-to-one with those of the original network. We note that for other contrasts, the connected components of the joint unobservability graph need not align with the connected components of the original graph; see our application to \citet{paluck_changing_2016} in \Cref{sec:sims-and-apps} for an example.

	\begin{rem}[Relationship to dependency graphs]
		We note that the joint unobservability graph is distinct from the dependency graph considered in e.g. \citet{baldi_asymptotic_1989} and \citet{leung_treatment_2020}. In a dependency graph, $k$ and $k'$ are adjacent if $W_k$ and $W_{k'}$ are dependent.\footnote{The definition of a dependency graph typically requires that for disjoint sets $S_1,S_2$ with no links between them, $(W_k)_{k \in S_1} \indep (W_k)_{k\in S_2}$, which implies that $k$ and $k'$ are adjacent if they are dependent but is actually a stronger condition.} By contrast, the joint unobservability graph defines $k$ and $k'$ to be adjacent only if $\pi_{kk'} = P(W_k=1,W_{k'}=1)=0$, which is a stricter criterion. To appreciate the difference, note that in a completely randomized experiment, the treatment assignment indicators for all units are correlated owing to the constraint that $\sum_i D_i = N_1$. It follows that the dependency graph is connected. By contrast, as described above, the joint unobservability graph has $N$ connected components, with each component containing only the two potential outcomes for a given unit. This distinction arises because while treatment indicators are correlated across units, any pair of treatment assignments for units $i \neq j$ occurs with positive probability (i.e. $P(D_i = d, D_j = d') > 0$ for any $i \neq j$ and $d,d'\in \{0,1\}$). Thus pairs of potential outcomes involving distinct units are adjacent in the dependency graph but not in the joint unobservability graph.
	\end{rem}	
	 \paragraph{The CC Estimator.}
	 
	 We now derive a variance estimator that is unbiased for an upper bound on $\Var(\hat\tau)$. The conservativeness of the variance estimator is proportional to the heterogeneity across the connected components of the joint unobservability graph.
	 
 	Let $s_{kk'} = 1$ if $k$ and $k'$ are in the same component ($k,k' \in \mathcal{K}_c$ for some $c$), and let $s_{kk'} = - 1/(C-1)$ otherwise. (We will assume $C \geq 2$ so that division by $C-1$ is well-defined.) Let $b_{kk'} = v_{kk'}+ s_{kk'}$. Define the feasible variance estimator
 	$$ \VhatCC = \frac{1}{K^2} \sum_{ \substack{ k,k' \\ \pi_{kk'} >0 }} b_{kk'} \frac{W_k W_{k'}}{ \pi_{kk'} } Y_k Y_{k'}, $$
 	which we refer to as the connected components (or CC) estimator.
	
	\begin{prop} \label{prop:cc-valid}
		Suppose $C > 1$. Then
		
	\begin{equation}
	\E[\VhatCC] = \Var(\hat\tau) + \frac{C}{K^2} \Var_c(\tau_c) \geq \Var(\hat\tau) \label{eqn:conservativeness-vhatcc}	
	\end{equation}
	where $$\Var_c(\tau_c) = \frac{1}{C-1} \sum_c (\tau_c - \bar{\tau})^2 $$
	is the finite-population variance of the component sums $\tau_c:=\sum_{k\in\mathcal{K}_c}Y_k$ across components (with $\bar{\tau}:= \frac{1}{C} \sum_c \tau_c$).  	
	\end{prop}

As we now show, $\VhatCC$ recovers the standard variance estimators for completely randomized and cluster-randomized experiments, where the conservativeness of the usual estimator is proportional to the treatment effect heterogeneity across either individuals or clusters.

\paragraph{\Cref{example:cre} (revisited).} In a completely randomized experiment, it is straightforward to verify that $\VhatCC$ corresponds to the usual Neyman variance estimator, i.e. $\VhatCC = \tfrac{1}{N_1} \hat{s}_1^2 	+ \tfrac{1}{N_0} \hat{s}_0^2$ where $\hat{s}_d^2 = \tfrac{1}{N_d - 1} \sum_{i:D_i=d} (Y_i - \bar{Y}_d )^2$ is the sample variance for treatment arm $d$. In this case, the slackness term $\frac{C}{K^2} \Var_c(\tau_c)$ is equal to $\frac{1}{N} \Var_i(\tau_i)$, the variance of treatment effects $\tau_i = Y_i(1)- Y_i(0)$.\footnote{For any index $\ell=1,\ldots,L$ and fixed attributes $X_\ell$, we use $\Var_{\ell}(X_\ell)$ to denote the finite-population variance, $\frac{1}{L-1} \sum_\ell (X_\ell - \E_\ell[X_\ell])^2$, where $\E_\ell[X_\ell]:=\frac{1}{L}\sum_\ell X_\ell$ is the average of the $X_\ell$.} To see why this is the case, recall that the connected components $c$ correspond to the set of potential outcomes for a particular unit $i$. Thus, $\tau_c = \sum_{k \in \mathcal{K}_c} Y_k = 2 \tau_i$, so the variance of $\tau_c$ across components is proportional to the variance of $\tau_i$ across units. Equation \eqref{eqn:conservativeness-vhatcc} thus reduces to the classic result that the Neyman variance estimator's conservativeness is proportional to the variance of treatment effects. 

\paragraph{\Cref{example:cluster-experiment} (revisited).} In a cluster-randomized experiment with equal cluster sizes, $\VhatCC$ has two familiar representations. First, it is equal to the Neyman variance estimator applied to the cluster-averages $\bar{Y}_g = \frac{1}{m} \sum_{i: g(i) =g} Y_i$. Second, it is equal to the CR2 cluster-robust variance estimator applied to the representation of $\hat{\tau}$ as the OLS coefficient from the regression $Y_{i} = \alpha + \tau D_{g(i)} + \epsilon_i$.\footnote{\citet{samii_equivalencies_2012} show the equivalence between the Neyman variance estimator and the HC2 variance estimator in a completely randomized experiment. It is straightforward to verify a similar equivalence for CR2 and the Neyman variance estimator applied to cluster averages.} The conservativeness is proportional to the variance of cluster-averages of treatment effects, $\Var_g( \bar{\tau}_g )$. (Equivalently, it is proportional to the variance of cluster-level sums of treatment effects, $\Var_c(\tau_c) = 4 m^2 \Var_g(\bar\tau_g)$, since clusters are the same size $m$.) This is because each connected component $c$ corresponds to the set of potential outcomes for a particular cluster $g$, and thus $\tau_c$ is proportional to $\bar{\tau}_g$.

\paragraph{Conservativeness in general causal settings.} More generally, \Cref{prop:cc-valid} shows that the conservativeness of $\VhatCC$ will generically be proportional to the heterogeneity of causal effects across components in causal settings. Recall that in causal settings, all the potential outcomes associated with a particular unit $i$ must be in the same connected component, since any two potential outcomes for the same unit are jointly unobservable. If the $Y_k$ are equal to $ \omega_{i,d} Y_i(d)$ such that $\tau_i = \sum_d \omega_{i,d} Y_i(d)$ is a causal contrast, then $\tau_c = \sum_{i \in \mathcal{I}_c} \tau_i$ is the sum of the individual contrasts over the set of individuals $i$ whose potential outcomes are in component $c$ (denoted $\mathcal{I}_c$). Thus, the conservativeness of $\VhatCC$ in causal settings will generically be proportional to the heterogeneity of treatment effect sums across connected components.\footnote{Note that if the connected components are the same size, then the variance of treatment effect sums is proportional to the variance of treatment effect averages. This is why in \Cref{example:cluster-experiment}, which assumes equal cluster sizes for simplicity, the conservativeness is proportional to $\Var_g(\bar\tau_g)$ as described above.}  

\subsection{Validity under sampling perspective}

Although we primarily focus on the design-based setting where the potential outcomes are viewed as fixed, we now show that the CC estimator, like the Neyman variance in a completely randomized experiment, is also valid under certain data-generating processes in which the potential outcomes are viewed as stochastic. In this subsection only, we consider the setting in which potential outcomes $Y=(Y_1,\ldots,Y_K)'$ and treatments $W=(W_1,\ldots,W_K)'$ are both viewed as stochastic. We assume that the treatment assignment mechanism does not depend on the sampled values of the potential outcomes $Y$, so that $Y \indep W$. We will write $\E_{Y,W}[\cdot]$ and $\Var_{Y,W}(\cdot)$ to denote the expectation and variance over the joint distribution of $(Y,W)$. We note that the notation in previous sections implicitly conditions on $Y$ (since it is viewed as fixed), so expressions like $\Var(\hat\tau)$ in earlier sections can be viewed as $\Var_W(\hat\tau \mid Y)$ from the perspective in this subsection. 

We now describe when the CC estimator is valid from this alternative sampling perspective. By the law of total variance, we have that
\begin{align}
	\Var_{Y,W}(\hat\tau) &= \E_{Y}[ \overbracket{\Var_W(\hat\tau \mid Y)}^{\text{Design-based variance}} ] + \Var_{Y}\big( \E_W[\hat\tau \mid Y] \big) . \label{eqn:var-tauhat-sampling}
\end{align}
On the other hand, using \Cref{prop:cc-valid}, we have that
\begin{align}
	\E_{Y,W}[\VhatCC] = \E_{Y}[\Var_W(\hat\tau \mid Y)] + \frac{C}{K^2} \E_{Y}\big( \Var_c(\tau_c) \big) \label{eqn:e-vcc-sampling}.
\end{align}
It is immediate that the first terms on the right-hand side of \eqref{eqn:var-tauhat-sampling} and \eqref{eqn:e-vcc-sampling} are the same, and so $\VhatCC$ is unbiased for $\Var_{Y,W}(\hat\tau)$ if the second terms are the same as well. Note, however, that since $\hat\tau$ is unbiased for $\tau$, we have that 
\begin{equation} 
	\Var_{Y}\big( \E_W[\hat\tau \mid Y] \big) = \Var_{Y}(\tau) = \frac{C^2}{K^2} \Var_{Y}\bigg( \frac{1}{C} \sum_c \tau_c \bigg) ,
	\label{eqn:e-hat-tau-given-y}
\end{equation}
where the second equality uses the fact that $\tau = \frac{1}{K} \sum_k Y_k = \frac{C}{K} \frac{1}{C} \sum_c \tau_c$. We thus see that $\VhatCC$ is unbiased if 
\[ \E_{Y}\big( \Var_c(\tau_c) / C \big)  = \Var_{Y}\bigg(\frac{1}{C} \sum_c \tau_c \bigg) .\]
The term on the right-hand side is the variance of the sample mean of the $\tau_c$, while the term on the left-hand side is the sample variance of the $\tau_c$ (with denominator $C-1$) divided by $C$, i.e. the usual squared standard-error for the sample mean of $\tau_c$ under iid sampling. It follows that if the distribution of $Y$ is such that the component-level sums $\tau_c$ are iid, then $\VhatCC$ is unbiased for $\Var_{Y,W}(\hat\tau)$.

The following proposition gives a somewhat more general statement: if the sampling scheme for $Y$ is such that $\tau_c$ and $\tau_{c'}$ are uncorrelated for $c \neq c'$, then $\VhatCC$ estimates an upper bound on $\Var_{Y,W}(\hat\tau)$, with the conservativeness of the upper bound corresponding to the heterogeneity in $\mu_c = \E_{Y}[\tau_c]$. Hence, the CC estimator is valid from the sampling perspective if the potential outcome generation process is uncorrelated across connected components, and unbiased if there is no heterogeneity in the $\mu_{c}$. We note that if the potential outcomes are generated independently at the individual level, then this necessarily implies that they are uncorrelated across connected components (since all potential outcomes for the same unit are in the same connected component).

\begin{prop} \label{prop:super-pop-interp}
Suppose that the distribution of $Y$ is such that $\Cov_Y(\tau_c,\tau_{c'}) = 0$ for $c \neq c'$. Assume that $C>1$. Let $\mu_c = \E_{Y}[\tau_c]$. Then 
\[	\E_{Y,W}[\VhatCC] = 	\Var_{Y,W}(\hat\tau) + \frac{C}{K^2} \Var_c(\mu_c) \geq 	\Var_{Y,W}(\hat\tau). \]
\end{prop}

\subsection{Admissibility and Bayesian Interpretation}
\paragraph{Admissibility.} We now show that the proposed estimator $\VhatCC$ is admissible among the class of quadratic-form estimators by the definition proposed in \citet{harshaw_optimized_2026}. As in \citet{harshaw_optimized_2026}, we will consider the class of variance estimators that use Horvitz-Thompson weighting over observable pairs to estimate a quadratic-form in the potential outcomes: specifically, letting $\mathcal{M}$ denote the set of symmetric $K \times K$ matrices with elements $m_{kk'}=0$ for all $k,k'$ with $\pi_{kk'}=0$, we will consider the class of estimators for  $\Var(\hat\tau)$ of the form
\begin{equation}
\VhatM = \frac{1}{K^2} \sum_{ \substack{ k,k' \\ \pi_{kk'} >0 }} m_{kk'} \frac{W_k W_{k'} }{ \pi_{kk'} } Y_k Y_{k'} , \qquad M=(m_{kk'})_{k,k'} \in \mathcal{M}, \label{eqn:hat-v-m}
\end{equation} 
which we call the class of quadratic-form estimators for short. It is clear that $\VhatCC$ above is a special case of this class in which $M=B$, for $B$ the matrix with elements $b_{kk'}= v_{kk'} + s_{kk'}$. Note that 
$$\E[\VhatM] = \frac{1}{K^2} \sum_{ \substack{ k,k' \\ \pi_{kk'} >0 }} m_{kk'} Y_k Y_{k'} = \frac{1}{K^2} Y' M Y, $$
where we use the fact that $m_{kk'}=0$ if $\pi_{kk'}=0$. Recall that
$$\Var(\hat\tau) = \frac{1}{K^2} \sum_{k,k'} v_{kk'} Y_k Y_{k'} = \frac{1}{K^2} Y' V Y,$$
where $V$ is the matrix with $kk'$th element $v_{kk'}$. It follows that  $\E[\VhatM] \geq \Var(\hat\tau)$ for all possible values of $Y$ if and only if $M - V$ is positive semi-definite. (Henceforth, we will write that $M-V$ is PSD or $M-V \succeq 0$ to denote that the matrix $M-V$ is positive semi-definite.) \citet{harshaw_optimized_2026} define an estimator $\VhatM$ to be inadmissible if there exists another upwards-biased estimator $\VhatQ$ of the form \eqref{eqn:hat-v-m} such that $\E[\VhatQ] \leq \E[\VhatM]$ for all values of $Y$ with strict inequality for some value of $Y$. More precisely, $\VhatM$ is admissible (among the class of estimators of the form \eqref{eqn:hat-v-m}) if $M-V$ is PSD and there exists no matrix $Q \neq M$ such that, $Q-V$ is PSD, $Q_{kk'}=0$ if $\pi_{kk'}=0$, and $M - Q$ is PSD. 

\citet{harshaw_optimized_2026} shows that $\VhatM$ is admissible if and only if it is the solution to an optimization of the form 
\begin{equation}
\min_{M} g(M-V) \text{ s.t. } M_{kk'} = 0 \text{ if } \pi_{kk'}=0 \text{ and } M - V \text{ is PSD} , \label{eqn:g-optimization}
\end{equation}
for some function $g$ that is strictly monotone in the Loewner order in the sense that $g(M_1) \geq g(M_2)$ if $M_1 - M_2$ is PSD, with strict inequality if $M_1 \neq M_2$. \citet{harshaw_optimized_2026} propose to find admissible estimators by solving the optimization \eqref{eqn:g-optimization} for particular choices of $g$, such as the trace or Frobenius norm. This leads to a convex optimization over PSD matrices, that \citet{harshaw_optimized_2026} implement using a splitting conic solver, which may be computationally demanding. 

We now show that the estimator $\VhatCC$ corresponds to a closed-form solution to the optimization \eqref{eqn:g-optimization} for a particular choice of $g$, although not one of the objectives that \citet{harshaw_optimized_2026} consider in their applications. It follows immediately that $\VhatCC$ is admissible.

\begin{prop} \label{prop:vhatcc-admissible}
	Let $\langle M_1,M_2 \rangle = \tr(M_1 M_2)$ be the trace inner product. Suppose $C>1$ and each connected component $\mathcal{K}_c$ contains at least two items.\footnote{This is trivially satisfied in causal settings since all potential outcomes for the same unit must be in the same connected component.}  Then there exists a positive definite matrix $\Sigma$ such that $B$ is the solution to \eqref{eqn:g-optimization} for $g(M-V) := \langle \Sigma, M-V \rangle$, which is strictly monotone in the Loewner order. Hence, $\VhatCC$ is admissible (in the class of quadratic-form estimators).	 
\end{prop}

\paragraph{Bayesian interpretation.} Building on the discussion in \citet{harshaw_optimized_2026}, we can attribute a Bayesian interpretation to the optimization that yields $\VhatCC$ as a solution. Consider a Bayesian who has a prior over $Y$ and wishes to choose a quadratic-form estimator $\VhatM$ that minimizes the expected conservativeness of $\VhatM$ over the prior, subject to the constraint that $\E[\VhatM] \geq \Var(\hat\tau)$ for all values of $Y$. Note that for any fixed value of $Y$, the conservativeness of $\VhatM$ is $$\E[\VhatM] - \Var(\hat\tau) = \frac{1}{K^2} Y'(M-V)Y = \frac{1}{K^2} \tr(YY' (M-V)) = \frac{1}{K^2} \langle YY', M-V  \rangle.$$
It follows that given a prior $Y \sim \Pi$, the expected conservativeness of $\VhatM$ over $\Pi$ is $K^{-2} \langle  \E_{\Pi}[YY'], M-V \rangle$. Hence, the Bayesian wishes to solve the problem \eqref{eqn:g-optimization} with $g(M-V)= \langle  \E_{\Pi}[YY'], M-V \rangle$. \Cref{prop:vhatcc-admissible} shows that $\VhatCC$ optimizes \eqref{eqn:g-optimization} for the objective $g(M-V) = \langle \Sigma, M-V \rangle$ for a particular positive-definite choice of $\Sigma$. Hence, $\VhatCC$ minimizes expected conservativeness for a Bayesian with a prior $\Pi$ such that $\E_\Pi[Y]=0$ and $\Var_{\Pi}(Y) = \Sigma$. It is therefore instructive to consider the moments of this prior. The proof to \Cref{prop:vhatcc-admissible} considers the class of variance matrices of the form 
$$\Sigma_{\alpha} = P_{S^\perp} + \alpha A, $$ 
where $\alpha$ is a positive constant and $P_{S^\perp}$ is the orthogonal projection into the null space of $S$ (for $S$ the matrix with elements $s_{kk'}$). The proof to \Cref{prop:vhatcc-admissible} then shows that $\VhatCC$ corresponds to the solution to \eqref{eqn:g-optimization} with $g(M-V) =  \langle \Sigma_{\alpha}, M-V \rangle$ for $\alpha$ sufficiently small. Intuitively, this corresponds to a class of priors that expect the heterogeneity in component sums $\tau_c$ to be relatively small and to be ``equally distributed'' across components in a particular sense. Specifically, for distinct components $c,c'$, if $\Var_\Pi(Y) = \Sigma_\alpha$, then 
$$\Var_\Pi(\tau_c - \tau_{c'}) = \alpha (m_c + m_{c'}),$$
where $m_c = \sum_{k,k' \in \mathcal{K}_c} A_{kk'}$ is the number of pairs in component $c$ with $\pi_{kk'}=0$. Note that for $\alpha =0$, this imposes that $\Var_\Pi(\tau_c - \tau_{c'}) = 0$, in which case there is no cross-component heterogeneity in the component-sums of treatment effects $\tau_c$.  Moreover, if $\alpha >0$ and the $m_c$ are constant across components (as occurs in a cluster-randomized experiment with equal cluster sizes, for example), then $\Pi$ expects the heterogeneity across components to be equally-distributed across pairs in the sense that $\Var_\Pi(\tau_c - \tau_{c'})$ is the same for all pairs. In sum, we see that $\VhatCC$ corresponds to the optimal choice for a class of Bayesians who expect heterogeneity across clusters to be relatively small, and who distribute that heterogeneity equally across pairs of clusters that have the same number of zero-probability links. 

\begin{rem}[Bayesian justification for the Neyman variance in experiments]
Recall that in a completely randomized experiment, the estimator $\VhatCC$ corresponds to the usual Neyman variance estimator. In this setting, all the connected components are the same size, consisting of the two potential outcomes for each unit, and $\tau_c$ is proportional to $\tau_i$, the treatment effect for unit $i$. The discussion above thus implies that the Neyman variance is optimal (among the class of quadratic-form estimators) for a Bayesian who expects treatment effect heterogeneity to be sufficiently small, and whose prior is symmetric in the sense that $\Var_\Pi(\tau_i - \tau_j)$ is constant across distinct pairs $i \neq j$.
\end{rem}

\section{Comparison to other estimators\label{sec:comparison-in-experiments}}

We now provide a comparison of the CC estimator to other variance estimators that are applicable for complex designs. We begin by comparing the estimators in individual- and cluster-randomized experiments (\Cref{example:cre,example:cluster-experiment} above), which are simple enough that the estimators can be derived in closed form. We then provide a more general characterization of variance estimators that are exact under the sharp null of no treatment effects, and show how the CC estimator compares to other estimators in this class. 

We focus our comparisons in simple experimental settings on three variance estimators: two implementations of the \citet{harshaw_optimized_2026} approach described above, using the trace and Frobenius norms as the objective, and the \citet{aronow_estimating_2017} estimator. Before proceeding, we briefly review the construction of the \citet{aronow_estimating_2017} estimator. \citet[][AS]{aronow_estimating_2017} make the insightful observation that by Young's inequality, $-Y_k Y_{k'} \leq \frac{1}{2}(Y_k^2 +Y_{k'}^2)$. We can therefore obtain an upper bound on $\Var(\hat\tau)$ by replacing the unobservable terms $-Y_kY_{k'}$ in the expression for $\Var(\hat\tau)$ with the upper bound $\frac{1}{2}(Y_k^2 +Y_{k'}^2)$. This leads to an estimator of the form $\VhatM$ given in \eqref{eqn:hat-v-m}, where $M= V + D + A$, where $D$ is the degree matrix of the joint unobservability graph (i.e. the diagonal matrix with $D_{kk} = \sum_{k'} A_{kk'}$). We thus see that the slack matrix of the AS estimator, $D+A \succeq 0$, is precisely the signless Laplacian of the joint unobservability graph. We refer to the resulting estimator as $\VhatAS$.

We now turn to a comparison of the CC estimator and the AS and \citet{harshaw_optimized_2026} estimators in individual- and cluster-randomized experiments. We provide a summary of the comparisons of the estimators, with detailed derivations in \Cref{sec:cre-derivations}. 

\subsection{Completely-randomized experiment.}
First, consider a completely randomized experiment, with $N_0,N_1 \geq 2$. The estimator $\hat\tau$ is the usual difference-in-means estimator.

\paragraph{The CC estimator.} As noted above, the CC estimator reduces to the usual Neyman variance estimator, $\frac{1}{N_1} \hat{s}_1^2 + \frac{1}{N_0} \hat{s}_0^2$. Its conservativeness is $\DeltaCC = \frac{1}{N} \Var_i(\tau_i)$. 

\paragraph{The AS estimator.} The AS estimator can be written as $\VhatAS = \frac{N-1}{N} \VhatNeyman  + \frac{1}{N} \hat\tau^2$. Its conservativeness is $\DeltaAS = \frac{1}{N} \E_i[\tau_i^2]$. We thus see that while the CC estimator's conservativeness is proportional to the finite-population variance of treatment effects, the AS estimator's conservativeness is proportional to the second moment of treatment effects, $\E_i[\tau_i^2]$. We thus expect the AS estimator to be substantially more conservative than the CC estimator when the magnitude of the average treatment effect ($\E_i[\tau_i]^2$) is large relative to the heterogeneity $\Var_i(\tau_i)$.  We note that because $\Var_i(\tau_i)$ is defined as the finite-population variance with denominator $(N-1)$, it is not the case that the conservativeness of the CC estimator, $ \frac{1}{N} \Var_i(\tau_i) = \frac{1}{N} \frac{N}{N-1} (\E_i[\tau_i^2] - \E_i[\tau_i]^2) $ is necessarily weakly smaller than that of the AS estimator, $\frac{1}{N} \E_i[\tau_i^2]$. However, the conservativeness of the AS estimator can never be much smaller than that of the CC estimator: $\DeltaAS \geq (1-\tfrac{1}{N}) \DeltaCC$, so the proportional reduction in conservativeness from using AS over CC is never more than $1/N$. By contrast, if $\tau = \E_i[\tau_i]$ is non-zero but treatment effects are homogeneous, then the CC estimator is unbiased for the variance of $\hat\tau$ while the AS estimator is conservative by $\frac{1}{N} \tau^2$, which can be large if the treatment effect is large in magnitude. Thus, it is possible to have $\DeltaAS \gg \DeltaCC$ if treatment effects can be large but relatively homogeneous. 

\paragraph{Frobenius Norm Minimization.} Consider the \citet{harshaw_optimized_2026} variance estimator obtained by minimizing the Frobenius norm of $M$.\footnote{As in \citet{harshaw_optimized_2026}, we consider minimizing the Frobenius norm of $||M||_F$, rather than of the slack $||M-V||_F$. This corresponds to setting $g(S)= ||S+V||_F$ in \eqref{eqn:g-optimization}.} \Cref{prop:collapsed-as-frobenius-optimal} in the appendix shows that the AS estimator is the unique solution to the Frobenius norm optimization when the experiment is balanced $(N_1=N_0 \geq 2)$. Thus, the comparison between CC and AS also applies to a comparison of CC and Frobenius-norm minimization in a balanced experiment.

\paragraph{Trace Norm Minimization.} Next consider the \citet{harshaw_optimized_2026} estimator using the trace as the objective. Applied to the CRE, it turns out that the trace optimization does not have a unique solution. Both the CC estimator and AS estimator can be shown to be trace-optimal. If one numerically solves the optimization \eqref{eqn:g-optimization} with the trace as the objective, which estimator is obtained will depend on how the solver chooses between multiple optima.

\subsection{Cluster-randomized experiment}

Next consider a cluster-randomized experiment with \(G\) clusters each of size \(m = \tfrac{N}{G}>1\), of which \(G_1\) are assigned to treatment and \(G_0\) are assigned to control.\footnote{The derivations in \Cref{sec:cre-derivations} allow for $m=1$ to nest the completely randomized experiment as a special case. Since we have already discussed the CRE, we focus on $m>1$ in this subsection for ease of exposition.} Let
\[
\bar Y_g(d) := \frac{1}{m} \sum_{i:g(i)=g}Y_i(d), \qquad \bar Y_g := \bar Y_g(D_g), \qquad \bar\tau_g = \frac{1}{m} \sum_{i:g(i)=g} (Y_i(1)-Y_i(0)) =  \bar Y_g(1)- \bar Y_g(0).
\]
The estimator \(\hat\tau\) is the difference in means across treated and control units, $\hat\tau = \frac{1}{N_1} \sum_{i : D_{g(i)} =1} Y_i - \frac{1}{N_0} \sum_{i : D_{g(i)} =0} Y_i$, where $N_1=mG_1,N_0=mG_0$ denote the number of treated and untreated units. Equivalently, it is the difference-in-means across clusters of cluster-level averages, $\hat\tau =\frac{1}{G_1} \sum_{g: D_g =1} \bar Y_g - \frac{1}{G_0} \sum_{g: D_g =0} \bar Y_g$.

\paragraph{The CC estimator.} The CC estimator reduces to the usual Neyman variance estimator applied to cluster-level averages,
\[
\VhatNeymanCl
=
\frac{1}{G_1}\hat s_{1,\mathrm{cl}}^2
+
\frac{1}{G_0}\hat s_{0,\mathrm{cl}}^2,
\]
where \(\hat s_{d,\mathrm{cl}}^2\) is the sample variance of the observed cluster level averages \(\bar{Y}_g\) among clusters assigned to treatment arm \(d\). This is equivalent to applying the CR2 cluster-robust standard errors to the regression of $Y_i$ on $D_{g(i)}$. The conservativeness of the CC estimator is proportional to the variance of treatment effects across clusters, $\DeltaCC=\frac{1}{G}\Var_g(\bar\tau_g)$.

\paragraph{The AS estimator.} There are two possible implementations of the AS estimator in the cluster-randomized setting. First, one can directly apply the AS estimator to the individual-level outcomes, in which case a pair $Y_i(1),Y_j(0)$ is jointly unobservable if and only if $i$ and $j$ are in the same cluster $g$. Alternatively, one can view the cluster-randomized trial as a completely-randomized experiment, viewing the cluster-level averages as the potential outcomes. Then the only jointly unobservable pairs are the two potential outcomes for the same cluster, $\bar{Y}_g(1),\bar{Y}_g(0)$. Denoting the direct and collapsed version of AS as $\VhatAS$ and $\VhatAScollapsed$, we have that
\[ \VhatAScollapsed = \frac{G-1}{G} \VhatNeymanCl + \frac{1}{G} \hat\tau^2, \qquad   \VhatAS = \VhatAScollapsed + \VhatWithin,  \]
\noindent where
\[
\VhatWithin
=
\frac{m-1}{mG^2}
\left[
\frac{G}{G_1}\sum_{g:D_g=1}s_g^2
+
\frac{G}{G_0}\sum_{g:D_g=0}s_g^2
\right],
\qquad
s_g^2
=
\frac{1}{m-1}
\sum_{i:g(i)=g}
\left(Y_i-\bar Y_g\right)^2.
\]
The collapsed AS estimator is analogous to the AS estimator in a randomized experiment, replacing individual outcomes with cluster-level averages. The direct AS estimator contains the additional term, $\VhatWithin$, which is a weighted sum of the sample variances of the outcome \emph{within} clusters, $s_g^2$. Since $\VhatWithin \geq 0$, with strict inequality if there is within-cluster heterogeneity, $\VhatAScollapsed \leq \VhatAS$. Thus, the direct AS estimator is dominated by the collapsed one. 

The comparison between the collapsed version of AS and the CC estimator is analogous to that in the individual-level CRE. The CC estimator's conservativeness is proportional to the cross-cluster variance of treatment effects, $\frac{1}{G} \Var_g(\bar\tau_g)$, while the collapsed AS estimator has conservativeness proportional to the uncentered second moment of cluster-level treatment effect averages, $\frac{1}{G} \E_g[\bar\tau_g^2]$. The conservativeness of the collapsed AS estimator can be smaller than that of the CC estimator only by the finite-sample factor \((G-1)/G\), attained when \(\tau=0\). By contrast, if cluster-level average treatment effects are homogeneous and nonzero, then the CC estimator is unbiased for the variance while collapsed AS is conservative by \(\tau^2/G\). The direct AS estimator is more conservative than the collapsed AS estimator, adding the within-cluster heterogeneity term \(\E[\VhatWithin]\), which is strictly positive if there is within-cluster heterogeneity in potential outcomes. We thus see directly that the direct AS estimator is inadmissible, since it is dominated by the collapsed AS estimator. 

\paragraph{Frobenius norm minimization.} Similar to the individual-level CRE, for balanced cluster assignment ($G_0=G_1$), the Frobenius norm optimization reduces to the collapsed AS estimator. The comparison of the CC estimator and collapsed AS estimator above thus also applies to a comparison of the CC estimator and Frobenius-norm optimizing estimator. 

\paragraph{Trace norm minimization.} As in the individual-level CRE, trace minimization is not unique. Both the CC estimator and the collapsed AS estimator are among the optimal solutions. The non-collapsed AS estimator is not trace-optimal when \(m>1\), as it is dominated by the collapsed AS estimator. Thus, if one solves the trace optimization numerically, the selected estimator depends on how the solver chooses among multiple solutions.

\subsection{Comparisons in more complex designs\label{subsec:complex-designs}}
We next consider a comparison of the CC estimator versus other estimators in more complex designs. 

We first note that, by definition, two admissible estimators cannot uniformly dominate each other for all sets of potential outcomes. As discussed above, the CC estimator is admissible (among quadratic-form estimators), as is any estimator that solves \eqref{eqn:g-optimization} for an objective $g$ that is strictly monotone in the Loewner order. Hence, neither the CC estimator nor the \citet{harshaw_optimized_2026} estimators based on the trace or Frobenius norms can uniformly dominate each other for all choices of potential outcomes. As shown above, the CC estimator is exact when the component-level treatment effect sums $\tau_c$ are homogeneous, and thus must perform at least as well as other quadratic-form estimators under homogeneity. Admissibility implies, however, that there must be orientations of the potential outcomes with sufficient heterogeneity in the $\tau_c$ such that the trace or Frobenius optimizers outperform the CC estimator.\footnote{This assumes that the Frobenius- and trace-optimal estimators are not equal to the CC estimator. We showed above that CC is one of the trace-optimal solutions in a CRE, for example.} 

By contrast, the AS estimator need not be admissible. \citet{harshaw_optimized_2026} show, for example, that the AS estimator is inadmissible in an individual-level CRE with $N_1=N_0=1$. The following result provides a complete characterization of when the AS estimator is admissible: AS is admissible if and only if all connected components $\mathcal{K}_c$ have at most two elements. 

\begin{prop} \label{prop:as-admissibility}
	The Aronow-Samii estimator $\VhatAS$ is admissible if and only if every connected component of the joint unobservability graph has 2 or fewer elements, $|\mathcal{K}_c| \leq 2$ for all $c$. 
\end{prop}

Note that in causal settings, all of the potential outcomes for the same unit are jointly unobservable and hence in the same connected component. Thus, in a causal setting where we are interested in a contrast involving multiple potential outcomes for each unit, AS is admissible if and only if (a) the contrast involves two potential outcomes per unit, and (b) any pair of potential outcomes involving distinct units is jointly observable. This implies that AS is admissible in a completely randomized experiment with at least two units per treatment arm.\footnote{\citet{harshaw_optimized_2026} show that AS is inadmissible in a CRE involving only one unit per treatment arm. In this case, $Y_1(1),Y_2(1)$ are jointly unobservable, since units $i=1,2$ cannot simultaneously be treated.} It follows that the AS estimator will be inadmissible in complex experimental designs where there are jointly unobservable pairs involving multiple units, or the contrast involves more than two potential outcomes per unit, and hence there must be a quadratic-form estimator that dominates AS, although it need not necessarily be the CC estimator, or the Frobenius or trace-optimal solutions. 

\paragraph{Comparison to other estimators that are exact under homogeneity.} A notable feature of the CC estimator is that it is exact under treatment effect homogeneity across connected components, which nests the special case of the sharp null of no treatment effects, since then $\tau_c =0$ for all $c$, and so $\Var_c(\tau_c)=0$. Our next result characterizes the class of quadratic-form variance estimators that are exact under the sharp null, as well as the subclass that are also exact under treatment effect homogeneity across components. This characterization implies a sense in which the CC estimator is optimal in the subclass of estimators that are exact under component homogeneity. For this comparison, we focus on the setting where the $Y_k$ correspond to a weighted average treatment effect. To be precise, we will say that the $Y_k$ correspond to a weighted average treatment effect if there exist potential outcomes $Y_i(q)$ for $i=1,\ldots,N$, $q=0,\ldots,Q$ such that $Y_k = \omega_{i,q} Y_i(q)$, where the weights $\omega_{i,q}$ are non-zero and sum to 0 within unit, $\sum_q \omega_{i,q} = 0$, and such that $W_k = 1\{D_i = q\}$.\footnote{One could allow for the number of potential outcomes to depend on the unit, so that unit $i$ has $Q_i$ potential outcomes, at the cost of additional notation. This would accommodate network settings where it may not be possible for certain units to realize certain exposure mappings.} We will then say that a variance estimator $\VhatM$ of the form in \eqref{eqn:hat-v-m} is \emph{sharp-null-exact} if for any $Y$ corresponding to a weighted average treatment effect for potential outcomes such that $Y_i(q)=Y_{i}(q')$ for all $i$ and all $q,q'$, it is the case that $\frac{1}{K^2}Y' M Y= \frac{1}{K^2}Y' V Y = \Var(\hat\tau)$. We will further say that $\VhatM$ is exact under component homogeneity if $\frac{1}{K^2}Y' M Y= \frac{1}{K^2}Y' V Y = \Var(\hat\tau)$ whenever $\tau_c = \sum_{k:k \in \mathcal{K}_c} Y_k$ is constant in $c$.  We then have the following result, which characterizes the set of valid estimators $\VhatM$ that are sharp null exact and exact under component homogeneity. 

\begin{prop} \label{prop:sharp-null-exact}
Suppose that the $Y_k$ correspond to a weighted average treatment effect, and consider an estimator $\VhatM$ that takes the form in \eqref{eqn:hat-v-m}.

\begin{enumerate}
	\item \emph{(Characterization of sharp-null exact $\VhatM$).} $\VhatM$ is weakly upwards-biased, $\E[\VhatM] \geq \Var(\hat\tau)$, and sharp-null-exact if and only if there exists a $C \times C$ correlation matrix $\Lambda$ such that
\begin{equation} T(M) := M-V = H \Lambda H', \quad \Lambda \succeq 0, \quad \Lambda_{cc} =1 \text{ for all } c, \label{eqn:h-lambda-h} 
\end{equation}
where $H$ is the $K \times C$ matrix indicating component-membership, i.e. such that $H_{kc} = 1\{k \in \mathcal{K}_c\}$. The conservativeness of any such $\VhatM$ depends only on the component-level sums, $\tau_c = \sum_{k: k \in \mathcal{K}_c} Y_k$. In particular,
\[ \E[\VhatM] - \Var(\hat\tau) = \frac{1}{K^2}\boldsymbol{\tau}' \Lambda \boldsymbol{\tau}, \]
where $\boldsymbol{\tau}$ is the vector with elements $\tau_c$.  When $C>1$, the CC estimator corresponds to the choice $\Lambda_{CC}$ which has elements 1 on the diagonal and $-1/(C-1)$ on the off-diagonal.

\item
\emph{(CC estimator minimizes the operator norm among estimators exact under component homogeneity.)} Suppose that $\VhatM$ is weakly upwards-biased, $\E[\VhatM] \geq \Var(\hat\tau)$, and exact under component homogeneity. Then $T(M)$ takes the form in \eqref{eqn:h-lambda-h} for $\Lambda$ such that $\Lambda \mathbf{1}_C = 0$, for $\mathbf{1}_C$ the $C \times 1$ vector of 1s. Among such $\Lambda$, the choice $\Lambda_{CC}$ minimizes the operator norm $||\Lambda||_{op}$ (provided $C > 1$ so $\Lambda_{CC}$ is well-defined).
\end{enumerate}
\end{prop}

The first part of the result shows that all quadratic form variance estimators that are weakly upwards-biased and sharp-null exact have conservativeness equal to a quadratic-form of the component-level sums $\tau_c$ and a correlation matrix $\Lambda$. The CC estimator corresponds to the maximal negative equicorrelation choice with $-1/(C-1)$ on the off-diagonal. The second part of \Cref{prop:sharp-null-exact} shows that of all estimators $\VhatM$ that are unbiased when treatment effects are homogeneous across components, the choice $\Lambda_{CC}$ minimizes the operator norm. 

\begin{rem}[Alternative variance estimators suggested by \Cref{prop:sharp-null-exact}] 
The formulation in \eqref{eqn:h-lambda-h} also makes it straightforward  to construct variance estimators that are exact under the sharp null and satisfy other criteria (e.g. minimize norms other than the operator norm). This formulation may help reduce the computational burden of implementing the numerical optimization suggested by \citet{harshaw_optimized_2026}; if we are content with a variance estimator that is sharp-null-exact, we can optimize only over the $C (C-1) / 2$ unique terms in $\Lambda$, rather than the $K(K+1)/2$ terms in $T$. \Cref{prop:sharp-null-exact} also suggests a way of constructing a variance estimator that is exact under \emph{individual}-level homogeneous effects rather than component-level homogeneous effects. Recall that the CC estimator is exact when the component-sums $\tau_c$ are homogenous. If components are the same size, then homogeneity of the $\tau_c$ is implied by homogeneous treatment effects at the individual level. However, if the components are different sizes, then homogeneity in treatment effects at the individual level will imply that $\tau_c \propto n_c$, for $n_c = |\mathcal{K}_c|$ the size of component $c$.\footnote{This assumes that $\omega_{i,d}\equiv \omega_d$, so all units get the same weight.} \Cref{prop:sharp-null-exact} implies that one could obtain a variance estimator that is exact under individual-level homogeneous effects by choosing a $\Lambda$ satisfying \eqref{eqn:h-lambda-h} and the constraint $\Lambda \mathbf{n} =0$, where $\mathbf{n}=(n_1,\ldots,n_C)'$, rather than the constraint $\Lambda \mathbf{1}_C = 0$ satisfied by the CC estimator.\footnote{It is possible to construct such a $\Lambda$ if and only if $\max_c n_c \leq \sum_c n_c /2$, so no component accounts for more than half of the total observations. This follows from the fact that $\Lambda$ is a correlation matrix and therefore $\Lambda_{cc'} \in [-1,1]$. The $c$th element of the constraint $\Lambda \mathbf{n}=0$ then yields $n_c = -\sum_{c': c' \neq c} \Lambda_{cc'} n_{c'}$, which together with the bound on $\Lambda$ implies that $n_c \leq \sum_{c': c' \neq c} n_{c'} = \sum_{c'} n_{c'} - n_c$, which yields the desired inequality. It is straightforward to verify that this condition is also sufficient by constructing the correlation matrix appropriately.} 
\label{rem:alt-estimators}
\end{rem}

\paragraph{Comparison to spatial HAC estimators.} In settings with spatial treatments, similar to \Cref{example:spatial}, a common approach is to use spatial-HAC variance estimators, such as \citet{conley_gmm_1999}. Spatial-HAC estimators were originally developed from the model-based perspective, but recent work has also considered design-based justifications for spatial-HAC in specific settings. \citet{wang_design-based_2025} consider a setting similar to \Cref{example:spatial} and show that spatial HAC with a particular choice of bandwidth is valid for the variance of a H\'ajek-style estimator when (a) the site-level assignments $T_s$ are assigned independently across sites, and (b) treatment effects satisfy a spatial ``homophily'' assumption. In \Cref{sec:hac-comparison-details}, we explore the properties of spatial HAC variance estimators and compare them to the CC estimator. We show that when treatment assignments satisfy a local dependence assumption (which is implied by the Bernoulli assignment scheme in \citet{wang_design-based_2025}), spatial HAC is guaranteed to be upwards-biased for $\Var(\hat\tau)$ without restrictions on the potential outcomes if and only if the HAC-kernel $\KHAC$ is PSD. When this is the case, we show that the conservativeness of the HAC estimator depends on the heterogeneity of treatment effects across a coarsening of the connected components used by the CC estimator, and thus the two estimators have similar properties. However, it need not necessarily be the case that $\KHAC$ is PSD, in which case the HAC variance estimator requires restrictions on the potential outcomes, such as the homophily assumption in \citet{wang_design-based_2025}, to ensure its validity. We also show that under local dependence, the positive semidefinite projection of HAC is upwards-biased from the design-based perspective without restrictions on the potential outcomes, but may be conservative even under the sharp null. By contrast, the CC estimator does not require a local dependence assumption, is valid without restrictions on the potential outcomes, and is unbiased under the sharp null. On the other hand, the CC estimator requires multiple connected components ($C>1$), which is not required by spatial-HAC estimators. The two approaches therefore appear to be complementary in spatial settings. 

\section{Consistency\label{sec:consistency}}

Our results so far give conditions under which $\VhatCC$ is an unbiased estimator of an upper bound on $\Var(\hat\tau)$. We now formalize conditions under which it is asymptotically consistent for an upper bound on the asymptotic variance of $\hat\tau$. Recall that $\hat\tau$ is a Horvitz-Thompson estimator over items indexed by $k$, while $\VhatCC$ is a Horvitz-Thompson estimator over items indexed by the pairs $(k,k')$ with $\pi_{kk'}>0$. Thus, general sufficient conditions for convergence of Horvitz-Thompson estimators, such as those in \citet{delevoye_consistency_2020}, apply immediately to both estimators. In what follows, we provide a simple set of sufficient conditions that yield consistency, rather than aiming for the most general set of conditions. We consider a triangular array of finite populations of growing size indexed by $K$, with outcomes $Y_{1,K},\ldots,Y_{K,K}$ and assignment probabilities $\pi_{kk',K}$ depending on the population size $K$. For ease of exposition, we leave the subscript $K$ implicit in the notation where not needed for clarity. 

We begin with simple conditions under which $\hat\tau$ is consistent.
\begin{asm}[Assumptions for consistency of $\hat\tau$]  \label{asm:first-order-ht-convergence}
	\leavevmode
	\begin{enumerate}[(i)]
		\item 
		(Bounded operator norm of $V$) $s_K := ||V||_{op}/K \to 0$.
		
		\item 
		(Bounded second moment of $Y$): $\frac{1}{K} \sum_k Y_k^2 = O(1)$.
		\end{enumerate}
\end{asm}

\begin{lem} \label{lem:tauhat-consistency}
	Under \Cref{asm:first-order-ht-convergence}, $\Var(\hat\tau) = O(s_K)$, and hence by Chebyshev's inequality, $\hat\tau - \tau = O_p(\sqrt{s_K})$.
\end{lem}

\noindent We note that $||V||_{op} \leq ||V||_{F}$, and so for (i) to hold it suffices that 
$$\frac{1}{K} ||V||_F = \sqrt{\frac{1}{K^2} \sum_{k,k'} \Cov\bigg(\frac{W_k}{\pi_k}, \frac{W_{k'}}{\pi_{k'}} \bigg)^2 } \to 0 ,$$
so that the average squared covariance in normalized sampling indicators is small.

Having established conditions under which $\hat\tau$ is consistent, we now consider conditions under which the CC estimator is consistent for an upper bound on $\Var(\hat\tau)$ in the sense that $\VhatCC/\Var(\hat\tau) \xrightarrow{p} r \geq 1$.

\begin{asm} \label{asm:second-order-ht-convergence}
		\leavevmode
	\begin{enumerate}[(i)]
		\item (Limiting distribution of $\hat\tau$)
		Suppose that $s_K^{-1/2} (\hat\tau - \tau) \xrightarrow{d} L$ for a non-degenerate distribution $L$ with $\Var(L) = \lim_{K \to \infty} s_K^{-1} \Var(\hat\tau) =: V^* \in (0,\infty)$.
		
		\item (Limit of CC slack)
		$s_K^{-1} (\E[\VhatCC]- \Var(\hat\tau)) = s_K^{-1}  \frac{C}{K^2} \Var_c(\tau_c) \to \delta \in [0,\infty)$.
		
		\item (Bounded fourth moment of $Y$) $\frac{1}{K} \sum_k Y_k^4 = O(1)$.
		
		\item (Limited pair-dependence in treatment indicators) 
		\begin{equation} \bigg(\frac{1}{K^4} \sum_{ \substack{(k,k'), (\ell,\ell') \\ \pi_{kk'}>0, \pi_{\ell \ell'}>0}}   \left| \Cov\bigg(\frac{W_k W_{k'}}{\pi_{kk'}} b_{kk'}, \frac{W_{\ell}W_{\ell'} }{\pi_{\ell \ell'}} b_{\ell\ell'} \bigg)\right|^{4/3}\bigg)^{3/4} = o(s_{K}^2) .\label{eqn:pairwise-cov-bound}
	\end{equation}
	\end{enumerate}
\end{asm}

\begin{prop} \label{prop:vhat-consistency}
	Suppose $C>1$ and \Cref{asm:first-order-ht-convergence,asm:second-order-ht-convergence} hold. Then $\VhatCC = \E[\VhatCC] + o_p(s_K)$ and hence $\frac{\VhatCC}{\Var(\hat\tau)} \xrightarrow{p} 1+  \frac{\delta}{V^*} \geq 1$.
\end{prop}

Part (i) imposes regularity conditions on the asymptotic limit distribution of $\hat\tau$. Specifically, it imposes that after centering and normalizing $\hat\tau$ by the rate $s_K^{-1/2}$, we obtain a non-degenerate limit distribution $L$. If a central limit theorem were applicable to $s_K^{-1/2} (\hat\tau - \tau)$, then $L$ would correspond to a normal distribution with variance $V^*$. Establishing a central limit theorem typically requires stronger assumptions than establishing variance consistency, and so we simply require that $s_K^{-1/2} (\hat\tau - \tau)$ converges to \emph{some} limit distribution; we refer the reader to e.g. \citet{chandrasekhar_general_2024}, \citet{wang_design-based_2025}, \citet{xu_design-based_2026}, among others, for some sufficient conditions for central limit theorems under complicated assignment schemes.  Part (i) also imposes that the variance of the limiting distribution is equal to the limit of the rescaled variance, $s_K^{-1} \Var(\hat\tau)$, as occurs if the sequence $s_K^{-1}(\hat\tau-\tau)^2$ is uniformly integrable. Part (ii) requires that the conservativeness of the CC estimator, which is proportional to the cross-component heterogeneity, is $O(s_K)$ and hence is the same (or smaller) order as $\Var(\hat\tau)$.\footnote{If we instead assume that $s_K^{-1} (\E[\VhatCC]- \Var(\hat\tau)) \to \infty$, then under the remaining conditions of \Cref{prop:vhat-consistency}, $\VhatCC / \Var(\hat\tau) \xrightarrow{p} \infty$, so the CC estimator is asymptotically valid but infinitely conservative.} Part (iii) bounds the fourth moment of $Y_k$. 
	
Part (iv) is perhaps the most substantive assumption, which intuitively, imposes that knowing that one pair $(k,k')$ is observed cannot, on average, be very informative about whether another pair $(\ell,\ell')$ is observed. Recall that to obtain consistency of $\hat\tau$, which is a Horvitz-Thompson estimator over items indexed by $k$, we had to restrict the dependence between the events that $k$ is observed and $k'$ is observed. However, $\VhatCC$ is a Horvitz-Thompson estimator over items indexed by observable pairs $(k,k')$. Thus, to obtain consistency, we have to restrict the dependence between pairs of pairs. More precisely, \eqref{eqn:pairwise-cov-bound} restricts the dependence of pairwise indicators multiplied by the CC-variance weights $b_{kk'}$. Note that when $v_{kk'} \approx 0$, $b_{kk'} \approx -1/(C-1)$ if $k,k'$ are in distinct components, and thus weighting by $b_{kk'}$ reduces the weight placed on covariances involving cross-component pairs whose treatment indicators are weakly correlated.

\section{Extensions}

\subsection{Extension to other estimators\label{sec:non-ht-estimators}}

For simplicity, we have focused our attention on estimating the variance of a Horvitz-Thompson estimator, and provided a variance estimator that is upwards-biased in finite samples. We now show that our approach can be applied to obtain asymptotically valid variance estimators for estimators that have an influence function representation under an appropriate asymptotic sequence of finite populations. We first provide a general sketch for how the approach can be applied to any estimator with an influence function representation under high-level conditions, and then provide lower-level conditions for the case of ordinary least squares estimators. 

\paragraph{Application to estimators with influence function representation.} First, observe that a sum of the form $\frac{1}{K} \sum_{k} W_k Y_k$ can be written as a Horvitz-Thompson estimator $\frac{1}{K} \sum_k \frac{W_k}{\pi_k} \tilde{Y}_k$ for the transformed outcome $\tilde{Y}_k := \pi_k Y_k$. However, a variety of estimators can be well-approximated by a sum of influence functions,
$$ \hat\theta - \theta = \underbracket{ \frac{1}{K} \sum_k W_k \phi(Z_k; \eta)}_{=O_p(r_K)} + o_p(r_K)  $$ 
where $\hat\theta$ is the estimator and $\theta$ its estimand, $r_K$ is the rate of convergence of $\hat\theta$, and $\phi(Z_k;\eta)$ is the influence function, which takes as input observable features $Z_k$ and a nuisance parameter $\eta$, and is such that $\E[\sum_k W_k \phi(Z_k;\eta)] = \sum_k \pi_k \phi(Z_k;\eta) =0$. The leading term in this expansion can be written as
$$ \hat\theta^{Lin}_K = \frac{1}{K} \sum_k W_k \phi(Z_k; \eta) =  \frac{1}{K} \sum_k \frac{W_k}{\pi_k} \psi_k, \qquad \psi_k := \pi_k \phi(Z_k;\eta),$$
a Horvitz-Thompson estimator for the average of the transformed influence functions $\psi_k := \pi_k \phi(Z_k;\eta)$. Under some regularity conditions, it will typically be the case that the asymptotic variance of $r_K^{-1}(\hat\theta - \theta)$ is equal to the asymptotic variance of $r_K^{-1} \hat\theta^{Lin}_K$, so that $r_K^{-1} (\hat\theta - \theta) \xrightarrow{d} L$ where $\Var(L) = \lim_{K \to \infty} r_K^{-2} \Var(\hat\theta^{Lin}_K)$. If the $\psi_k$ were directly observed, then it would be straightforward to apply the CC estimator (or other estimators described above) to estimate $\Var(\hat\theta^{Lin}_K)$ by setting $Y_k= \psi_k$, i.e. using $\widehat{V}_{\psi} :=  \frac{1}{K^2} \sum_{k,k': \pi_{kk'}>0} b_{kk'} \frac{W_{k} W_{k'}}{ \pi_{kk'}} \psi_k \psi_{k'}$. This would yield a consistent estimator for an upper bound on the asymptotic variance if the design satisfies the conditions in \Cref{prop:vhat-consistency} above with $Y_k = \psi_k$. In practice $\psi_k$ is typically not observed since it depends on the nuisance parameter $\eta$, and so we replace $\psi$ with $\hat\psi_k := \pi_k \phi(Z_k;\hat\eta) = \psi_k + \epsilon_k$, for $\epsilon_k:= \hat\psi_k - \psi_k$ the estimation error. This leads to the feasible variance estimator 
$$ \widehat{V}_{\hat\psi} :=  \frac{1}{K^2} \sum_{k,k': \pi_{kk'}>0} b_{kk'} \frac{W_{k} W_{k'}}{ \pi_{kk'}} \hat\psi_k \hat\psi_{k'} = \widehat{V}_{\psi} +  \frac{1}{K^2} \sum_{k,k': \pi_{kk'}>0} b_{kk'} \frac{W_{k} W_{k'}}{ \pi_{kk'}} (2 \psi_k \epsilon_{k'} + \epsilon_k \epsilon_{k'}) . $$ 
Under regularity conditions that control the estimation error $\epsilon_k$, the equation in the previous display will imply $\widehat{V}_{\hat\psi} = \widehat{V}_{\psi} + o_p(r_K^2)$, so that the error from plugging-in $\hat\psi$ is asymptotically negligible and hence $\widehat{V}_{\hat\psi}$ is also consistent for an upper bound. For example, suppose that $\epsilon_k = -x_k (\hat\eta - \eta)$ and $\hat\eta - \eta = O_p(r_K)$ (e.g. because $\hat\eta$ is also a Horvitz-Thompson estimator). In that case, we have that
$$\widehat{V}_{\hat\psi} = \widehat{V}_{\psi} - \underbracket{(\hat\eta - \eta)}_{=O_p(r_K)} \cdot  \frac{1}{K^2} \sum_{k,k': \pi_{kk'}>0}  \frac{W_{k} W_{k'}}{ \pi_{kk'}} 2 b_{kk'} \psi_k x_{k'}  + \underbracket{(\hat\eta-\eta)^2}_{=O_p(r_K^2)}  \frac{1}{K^2} \sum_{k,k': \pi_{kk'}>0} \frac{W_{k} W_{k'}}{ \pi_{kk'}} b_{kk'} x_k x_{k'},$$
and so it suffices to establish that the pair-level Horvitz-Thompson estimators for $b_{kk'} \psi_k x_{k'} $ and $b_{kk'} x_k x_{k'}$ are $o_p(r_K)$ and $o_p(1)$, respectively. 

\paragraph{OLS Estimators.} We now focus on the case of ordinary least squares (OLS) estimators, and give explicit conditions that ensure consistency. Suppose each unit $k$ is associated with a scalar outcome $y_k$ and a vector of covariates $X_k$, and consider the OLS estimator $\hat\theta = \big( \frac{1}{K} \sum_k W_k X_k X_k' \big)^{-1} \frac{1}{K} \sum_k W_k X_k y_k$. For example, if $d$ is a binary treatment and $k=(i,d)$ indexes the potential outcomes $Y_i(d)$, then setting $y_k = Y_i(d)$ and $X_k = (1,d,x_i')'$, for unit-level covariates $x_i$, yields $\hat\theta$ equal to the OLS estimator from the regression of $Y_i= Y_i(D_i)$ on $(1,D_i,x_i)$ in the observed sample. Let $$\theta = \E\bigg[ \frac{1}{K} \sum_k W_k X_k X_k' \bigg]^{-1} \E \bigg[ \frac{1}{K} \sum_k W_k X_k y_k\bigg] = \underbracket{ \bigg( \frac{1}{K} \sum_k \pi_k X_k X_k' \bigg)^{-1} }_{=:\Gamma_K^{-1}} \frac{1}{K} \sum_k \pi_k X_k y_k$$
be the OLS estimand, and define the residual $\epsilon_k = y_k - X_k'\theta$, which from the least-squares first order conditions satisfies the orthogonality condition $\sum_k \pi_k X_k \epsilon_k = 0$. We can then write
\[\hat\theta - \theta =\underbracket{ \bigg( \frac{1}{K} \sum_k \frac{W_k}{\pi_k} \pi_k X_k X_{k}'\bigg)^{-1} }_{=: \hat\Gamma_K^{-1} } \underbracket{ \frac{1}{K }\sum_k \frac{W_k}{\pi_k} \pi_k X_k \epsilon_k}_{=:\hat{U}_K} . \]
Observe that $\hat\Gamma_K$ and $\hat{U}_K$ are Horvitz-Thompson estimators for averages of the outcomes $\pi_k X_k X_{k}'$ and $\pi_k X_k \epsilon_k$, and thus if the regularity conditions in  \Cref{lem:tauhat-consistency} apply to these outcomes, we will have that $\hat\Gamma_K = \Gamma_K + O_p(\sqrt{s_K})$ and $\hat{U}_K = O_p(\sqrt{s_K})$. Suppose for simplicity that we are interested in a scalar contrast of the form $a'\theta$. If $\Gamma_K$ converges to a positive-definite limit $\Gamma$, then by Slutsky's lemma we expect that $\Var(a'\hat\theta - a'\theta) \approx a'\hat\Gamma_K^{-1} \Var(\hat{U}_K) \hat\Gamma_K^{-1} a$. This suggests a ``sandwich form'' estimator
$$\widehat{V}_{\epsilon} := a' \hat\Gamma_{K}^{-1} \hat\Omega_{\epsilon} \hat\Gamma_{K}^{-1} a, \quad \text{ where } \quad \hat\Omega_{\epsilon} = \frac{1}{K^2} \sum_{k,k': \pi_{kk'}>0} b_{kk'} \frac{W_k W_{k'}}{\pi_{kk'}} \pi_k X_k \epsilon_k \pi_{k'} X_{k'}' \epsilon_{k'}.$$ 
$\hat\Omega_{\epsilon}$ is the CC estimator for the variance of $\hat{U}_K$, i.e. a Horvitz-Thompson estimator with outcome $\pi_k X_k \epsilon_k$. The estimator $\widehat{V}_{\epsilon}$ is not feasible since it requires knowledge of the true residuals $\epsilon_k$. We therefore will consider the plug-in replacement $\widehat{V}_{\hat\epsilon}$ that replaces $\epsilon_k = y_k - X_k'\theta$ with $\hat\epsilon_k = y_k - X_k'\hat\theta$. Since $\hat\epsilon_k - \epsilon_k = -X_k'(\hat\theta - \theta)$, which is pointwise $o_p(1)$, we expect the error from plugging in $\hat\epsilon$ to be small asymptotically under appropriate regularity conditions. The following lemma formalizes sufficient conditions under which the resulting estimator $\widehat{V}_{\hat\epsilon}$ is consistent for an upper bound on the asymptotic variance.

\begin{lem}[Consistency for OLS] \label{lem:ols-consistency}
Consider an asymptotic sequence where $\dim(X)$ is fixed for all $K$. Suppose that $C>1$ for all $K$ and that
\begin{enumerate}[(i)]
	\item
	$s_K \to 0$.
	
	\item 
	(Bounded moments for $\pi_k X_k X_k'$, $\pi_k X_k\epsilon_k$) $\frac{1}{K} \sum_k || \pi_k X_k X_k'||^4_F = O(1)$ and $\frac{1}{K} \sum_k || \pi_k X_k \epsilon_k||^4_F = O(1)$.
	
	\item
	(Limited pair dependence in treatment indicators)
	Assignment probabilities satisfy \eqref{eqn:pairwise-cov-bound}.
	
	\item
	(Limiting distribution for $\hat{U}_K$)
	$\Gamma_K \to \Gamma \succ 0$ and $s_K^{-1/2} \hat{U}_K \xrightarrow{d} L$ for some distribution $L$ with $\Var(L) = \lim_{K \to \infty} s_K^{-1} \Var(\hat{U}_K) =: \Omega$ for finite $\Omega \succ 0$.
	
	\item
   (Convergence of CC slack at rate $s_K$.) For  $u_k =  \pi_k X_k \epsilon_k$ and $z_k = \mathrm{vec}(\pi_k X_k X_k')$, define $\tau_c^{(u,z)} := \sum_{k \in \mathcal{K}_c} (u_k' , z_k')'$. We have that $s_K^{-1} \frac{C}{K^2} \Var_c( \tau_c^{(u,z)} ) \to S^*_{(u,z)}$ for a finite matrix $S^*_{(u,z)}$, so that the CC conservativeness for the HT estimator for $\frac{1}{K} \sum_k  (u_k',z_k')'$ is $O(s_K)$. 
	
\end{enumerate}
Then for any vector $a \neq 0$, we have that 
$$a' s_K^{-1/2} (\hat\theta - \theta) \xrightarrow{d} L^* := a' \Gamma^{-1} L, \text{ where } \Var(L^*) = a' \Gamma^{-1} \Omega \Gamma^{-1} a$$
and $ s_K^{-1} \widehat{V}_{\hat\epsilon}  \xrightarrow{p} V^{ub} := \Var(L^*) + a' \Gamma^{-1} S_u^* \Gamma^{-1} a \geq \Var(L^*)$, where $S^*_u$ is the upper left block of $S^*_{(u,z)}$. 
\end{lem}

\subsection{Using arbitrary clustering instead of connected components\label{sec:arbitrary-clusters}}
The construction of the CC estimator requires that the joint unobservability graph has at least two connected components. We now show how a similar idea can be extended to the case where the joint unobservability graph is connected, but one has a partition of the units into clusters $\tilde{\mathcal{K}}_1,\ldots, \tilde{\mathcal{K}}_{\tilde{C}}$ for some $\tilde{C}>1$. We then show how one can construct a variance estimator whose conservativeness depends on (a) the variance of sums $\tau_{\tilde{c}}$ across the $\tilde{\mathcal{K}}_{\tilde{c}}$, and (b) the conservativeness of a separate upper bound (e.g. Young's inequality) applied to pairs of $(k,k')$ that cross clusters and have $\pi_{kk'}=0$. The estimator will therefore tend to be close to exact if cross-cluster heterogeneity is small \emph{and} there are relatively few links across clusters.

To construct the estimator, observe that for any partition $\tilde{\mathcal{K}}_1,\ldots, \tilde{\mathcal{K}}_{\tilde{C}}$ with $\tilde{C}>1$, we can write the variance of $\hat\tau$ as
\begin{align*}
	\Var(\hat\tau) = \bigg(  \overbracket{\frac{1}{K^2} \sum_{\substack{ k,k' \\ \pi_{kk'} >0 }} v_{kk'} Y_k Y_{k'}}^{\text{Jointly observable pairs}} + \overbracket{ \frac{1}{K^2} \sum_{ \substack{ k,k' \\ \pi_{kk'} =0  \\ k,k' \in \tilde{\mathcal{K}}_{\tilde{c}}} } (-1) Y_k Y_{k'}  }^{\text{Same-cluster unobservable pairs}}\bigg) +  \overbracket{\frac{1}{K^2} \sum_{ \substack{ k,k' \\ \pi_{kk'} =0  \\ k \in \tilde{\mathcal{K}}_{\tilde{c}}, k' \in \tilde{\mathcal{K}}_{\tilde{c}'} \neq \tilde{\mathcal{K}}_{\tilde{c}} } } (-1) Y_k Y_{k'} }^{\text{Cross-cluster unobservable pairs}} . \label{eqn:vhat-within-across-cluster} 
\end{align*}
We will denote the term in parentheses by $V_{\text{within}}$ and the term outside the parentheses $V_{\text{across}}$. Observe that if the $\tilde{\mathcal{K}}_{\tilde{c}}$ are the connected components, then the $V_{\text{across}}$ term vanishes since there are no cross-component pairs with $\pi_{kk'}=0$. We now show that an analog to the CC estimator that uses the $\tilde{\mathcal{K}}_{\tilde{c}}$ instead of the $\mathcal{K}_{c}$ provides an upwards-biased estimator for $V_{\text{within}} - \frac{1}{\tilde{C}-1} V_{\text{across}}$ whose conservativeness depends on the heterogeneity across the $\tilde{\mathcal{K}}_{\tilde{c}}$. Specifically, define $s^{\tilde{\mathcal{K}}}_{kk'}$ equal to 1 if $k,k' \in \tilde{\mathcal{K}}_{\tilde{c}}$ for some $\tilde{c}$ and equal to $-1/(\tilde{C}-1)$ otherwise, and define 
$$\VhatKwithin =\frac{1}{K^2} \sum_{ \substack{ k,k' \\ \pi_{kk'} >0 }} (v_{kk'}+s^{\tilde{\mathcal{K}}}_{kk'}) \frac{W_k W_{k'}}{ \pi_{kk'} } Y_k Y_{k'}. $$
Note that the estimator $\VhatKwithin$ corresponds exactly to $\VhatCC$ when the $\tilde{\mathcal{K}}_{\tilde{c}}$ are the connected components. The following lemma characterizes the expectation of $\VhatKwithin$.

\begin{lem}[Expectation of $\VhatKwithin$]
	\label{lem:expectation-vwithin}
	$$\E[\VhatKwithin] = V_{\text{within}} - \frac{1}{ \tilde{C} -1 } V_{\text{across}} + \frac{\tilde{C}}{K^2} \Var_{\tilde{c}}(\tau_{\tilde{c}}) $$
	where $\tau_{\tilde{c}}  = \sum_{k \in \tilde{\mathcal{K}}_{\tilde{c}}} Y_k$ is the within-cluster sum and $\Var_{\tilde{c}}(\cdot)$ denotes the finite-population variance across the $\tilde{\mathcal{K}}_{\tilde{c}}$ (with denominator $\tilde{C}-1$). 
\end{lem}
It follows that we can obtain an upwards-biased estimator for $\Var(\hat\tau)$ if we can obtain an upwards-biased estimator for $V_{\text{across}}$. One simple approach would be to use Young's inequality, as in AS, but applied only to the cross-cluster components. That is, we could use 
$$\widehat{V}_{\text{across}} = \frac{1}{K^2} \sum_{ \substack{ k,k' \\ \pi_{kk'} =0  \\ k \in \tilde{\mathcal{K}}_{\tilde{c}}, k' \in \tilde{\mathcal{K}}_{\tilde{c}'} \neq \tilde{\mathcal{K}}_{\tilde{c}} } } \frac{1}{2}\big( \frac{W_k}{\pi_k} Y_k^2 + \frac{W_{k'}}{\pi_{k'}} Y_{k'}^2 \big) $$
which has expectation
$$\E[\widehat{V}_{\text{across}} ]= V_{\text{across}} + \frac{1}{K^2} \sum_{ \substack{ k,k' \\ \pi_{kk'} =0  \\ k \in \tilde{\mathcal{K}}_{\tilde{c}}, k' \in \tilde{\mathcal{K}}_{\tilde{c}'} \neq \tilde{\mathcal{K}}_{\tilde{c}} } } \frac{1}{2}\big( Y_k +Y_{k'})^2 =: V_{\text{across}} + \Delta_{\text{Young},\text{across}}. $$
If a small share of the pairs $(k,k')$ are cross-cluster with $\pi_{kk'}=0$, then we will generally expect $ \Delta_{\text{Young},\text{across}}$ to be small, since it is normalized by $1/K^2$. The variance estimator $\widehat{V}^{\tilde{\mathcal{K}}} := \VhatKwithin + \frac{\tilde{C}}{\tilde{C}-1} \widehat{V}_{\text{across}}$ is then upwards-biased for $\Var(\hat\tau)$ with conservativeness $\frac{\tilde{C}}{K^2} \Var_{\tilde{c}}(\tau_{\tilde{c}}) + \frac{\tilde{C}}{\tilde{C}-1} \Delta_{\text{Young},\text{across}} $. We therefore expect the conservativeness of this estimator to be relatively small if it is true that both (a) the $\tau_{\tilde{c}}$ are close to homogeneous across clusters, so that $\Var_{\tilde{c}}(\tau_{\tilde{c}})$ is small, and (b) there are relatively few cross-cluster pairs $(k,k')$ with $\pi_{kk'}=0$, so that $ \Delta_{\text{Young},\text{across}} $ is small. Other approaches to obtain an upper bound on $V_{\text{across}}$ are possible, too: for example, one could solve an optimization similar to that in \citet{harshaw_optimized_2026} to find a feasible quadratic-form estimator for $V_{\text{across}}$ that minimizes a particular matrix norm.

While the CC estimator is exact when components are homogeneous, the estimator $\widehat{V}^{\tilde{\mathcal{K}}}$ just described will generally be conservative owing to the $\Delta_{\text{Young},\text{across}}$ term. It appears, however, that some conservativeness even under homogeneous effects is necessary when there is only one connected component, at least if we focus on quadratic-form estimators. Specifically, \Cref{prop:sharp-null-exact} above characterizes the set of quadratic-form variance estimators for causal contrasts that are exact under the sharp null of no treatment effects. A corollary of this result is that the only quadratic-form variance estimator that is sharp-null-exact when $C=1$ is the variance estimator $\widehat{V} = \hat\tau^2$.
\begin{cor} \label{cor:cequals1-trivial-sharp-null-exact}
Consider the setting in \Cref{prop:sharp-null-exact} and suppose $C=1$. Then $\VhatM$ is weakly upwards-biased, $\E[\VhatM] \geq \Var(\hat\tau)$, and sharp-null-exact if and only if $\VhatM = \hat\tau^2$. 
\end{cor}

We note that since $\Var(\hat\tau) = \E[\hat\tau^2] - \tau^2$, the variance estimator $\hat\tau^2$ is trivially upwards-biased and exact when $\tau=0$. It is not a very useful variance estimator, however. For one, if $\tau \neq 0$, then it has conservativeness $\tau^2$, regardless of sample size, and therefore does not converge to zero even as $K \to \infty$. Moreover, the implied $t$-statistic always has absolute value equal to 1, and therefore it is not useful for $t$-based inference. \Cref{cor:cequals1-trivial-sharp-null-exact} shows that any quadratic-form estimator that does not have these undesirable properties will necessarily not be sharp-null exact when $C=1$.  

\begin{rem}[Choice of clusters]
The above discussion allows an arbitrary partition of units into clusters $\tilde{\mathcal{K}}_1,\ldots, \tilde{\mathcal{K}}_{\tilde{C}}$. This raises the question of how to choose the clusters in practice. To minimize conservativeness, the researcher would ideally choose clusters that (a) have relatively homogeneous effects and (b) for which there are few cross-cluster pairs $k,k'$ with $\pi_{kk'}=0$. In some settings, there may be natural choices that approximately satisfy these criteria. For example, in spatial settings there may be geographic units (say counties) that we expect to have relatively homogeneous effects and for which there are relatively few cross-border pairs of jointly unobservable potential outcomes. In principle, one could also choose clusters algorithmically based on the joint unobservability graph, e.g. using balanced graph-partitioning methods \citep{kernighan_efficient_1970,buluc_recent_2016} to find clusters of approximately equal size with few cross-cluster edges. We leave a formal analysis of the optimal choice of algorithm to future work. 
\end{rem}

\section{Simulations and Applications\label{sec:sims-and-apps}}

We now consider simulations and applications based on two applied settings, one involving interference and the other a spatial treatment, similar to \Cref{example:spatial,example:interference}.
\subsection{A setting with interference}
We first revisit the simulation designs in \citet{harshaw_optimized_2026}, which are based on the social network experiment in \citet{paluck_changing_2016}. 

\paragraph{Setting.} As a brief review, \citet{paluck_changing_2016} conduct a randomized trial in schools intended to reduce bullying and conflict. They employ a two-stage stratified randomization procedure, in which 28 of 56 schools are randomly selected for the intervention; among the treated schools, a random subset of students are selected to participate in an anticonflict intervention. Prior to the intervention, the researchers collected information on the social network of students (based on questions about which students they spent time with). In this setting it is natural to expect that the intervention may have spillovers across students. \citet{harshaw_optimized_2026} specify an exposure mapping where student $i$'s outcome is assumed to depend on (i) whether they are in a treated school, (ii) whether they were directly selected for the intervention, and (iii) whether any of their friends were selected for the intervention.\footnote{This is a somewhat simplified version of the exposure mapping in \citet{paluck_changing_2016}, who further subdivide (iii) into levels for whether any peer was treated and for whether any popular peer was treated.} 

\paragraph{Simulation designs.} \citet{harshaw_optimized_2026} focus on two treatment contrasts.  First, they study the $(1,1,0)$ vs. $(0,0,0)$ contrast, which corresponds to the effect of being directly treated without any treated peers versus being in a control school (i.e. the ``direct effect'' of treatment). Second, they study the $(1,1,0)$ vs. $(1,0,1)$ contrast, which compares the effect of being directly treated (without any treated peers) to not being directly treated but having a peer treated (i.e. a comparison of the ``direct'' vs. ``indirect'' effects). \citet{harshaw_optimized_2026} approximate the $\pi_{kk'}$ by simulating 5 million times from the assignment mechanism. For each treatment contrast, they restrict attention to students who have at least a 0.005 probability of receiving both treatment levels in the contrast. They construct point estimates of the relevant contrasts using the Horvitz-Thompson estimator.  They then implement five variance estimators, the \citet{aronow_estimating_2017} estimator and the variance estimator defined by the optimization in \eqref{eqn:g-optimization} for four different objectives. For simplicity of presentation, we focus on the AS estimator and those based on the trace and ``composite'' objectives, which of the four objectives considered in \citet{harshaw_optimized_2026} yield the most and least conservative variance estimators in their main results reported in Table 1. (The composite objective is a weighted combination of the trace and Frobenius norms.) \citet{harshaw_optimized_2026} consider four designs for the potential outcomes. In two of the designs, they impose the sharp null of no treatment effects by setting all potential outcomes equal to an observed outcome in \citet{paluck_changing_2016}. They do this for two observed outcomes: whether a student wore an anticonflict wristband, and whether there were any disciplinary actions against the student. In the third simulation design, they used ``mixed'' potential outcomes where one of the potential outcomes is set to the observed Wristband outcome and the other is set to the observed Disciplinary outcome. The final simulation design is based on drawing the potential outcomes as $(X' \beta_0 +\epsilon_0, X' \beta_1 + \epsilon_1)$ where $X$ are observed covariates and the $\epsilon$'s are independent Gaussian noise (see \citet{harshaw_optimized_2026} for details on the choice of $X$ and $\beta$). The first two designs have homogeneous (zero) treatment effects, whereas the third and fourth designs have heterogeneous effects. To explore how the estimators behave under homogeneous but non-zero effects, we also consider a modification to the sharp-null simulations where we add a constant effect $\tau$ to the $Y(1,1,0)$ potential outcome for all units. Below we focus on $\tau=0.5$, which is close to the estimated average direct effect for the wristband outcome in the \citet{paluck_changing_2016} data.

\paragraph{Simulation implementation.} We implement the AS and trace and composite objective estimators identically to the simulations in \citet{harshaw_optimized_2026} using their replication code.\footnote{\citet{harshaw_optimized_2026} impose a constraint such that the variance estimator does not put any weight on pairs with $\pi_{kk'}<0.002$, and so we impose this constraint on these estimators as well.} For comparison, we implement the CC variance estimator using the same design. For the $(1,1,0)$ vs. $(0,0,0)$ contrast, there are 56 connected components corresponding exactly to the 56 schools in the data. This is because it is not possible for one student in a school to be directly treated while another student in the same school is in a control school; hence, all the potential outcomes for students in the same school are connected in the joint unobservability graph. For the $(1,1,0)$ vs. $(1,0,1)$ contrast, there are 468 connected components. The connected components do not reduce to the school for this outcome because the two relevant potential outcomes are not necessarily jointly unobservable for students in the same school; it is possible that unit $i$ in one school is directly treated with no treated neighbors and unit $j$ in the same school is not directly treated but has treated neighbors.\footnote{This is possible unless all of $j$'s friends are also $i$'s friends. In that case, $j$ can only have a treated neighbor if $i$ also has a treated neighbor. Hence, there are some joint unobservability links across units in the same school.} We conduct 100,000 simulations using the simulation designs in \citet{harshaw_optimized_2026} described above. For each variance estimator, we report the same statistics as in \citet{harshaw_optimized_2026}: (i) the bias of the variance estimator as a fraction of the true variance, namely $(\hat{\E}[\widehat{V}] - \widehat{\Var}(\hat\tau))/\widehat{\Var}(\hat\tau)$ (where $\hat{\E}$, $\widehat{V}$ denote mean and variance over the simulations), (ii) the precision of the variance estimator $\widehat{\mathrm{SD}}(\widehat{V})/ \widehat{\Var}(\hat\tau)$, (iii) coverage of a nominal 95\% interval of the form $\hat\tau \pm 1.96 \sqrt{\widehat{V}}$, and (iv) the average width of the nominal 95\% CI relative to that for the AS estimator.\footnote{Horvitz-Thompson-based variance estimators are generally not guaranteed to be non-negative with probability 1. If $\hat{V}<0$, we set $\sqrt{\hat{V}}=0$ when calculating the confidence intervals and average length. This occurs for the CC estimator in only 2 simulation draws across all specifications (which each use 100,000 simulations), and does not occur for the other variance estimators.}

\paragraph{Simulation results.} The results for the two contrasts are shown in \Cref{tab:simulation-results-30,tab:simulation-results-32}, respectively. Recall from \Cref{prop:cc-valid} that the CC estimator should be unbiased under the sharp null. In line with this, the estimated proportional bias in Panels A and B, which impose the sharp null, is very small in magnitude, less than 1\% of $\Var(\hat\tau)$ (with differences from zero likely due to simulation error). As it turns out, the CC estimator is also close to unbiased under the synthetic outcomes shown in Panel D, with upwards bias of 1-2\%. Intuitively, this occurs because the noise generated by the $\epsilon$ draws is independent across units, and thus with the specified error process there does not wind up being a large variance in the total treatment effect across components. In the Mixed PO design shown in Panel C, there is relatively more heterogeneity, and so the CC estimator has larger upward bias (64 and 16 percent, respectively). Coverage for the CC estimator is at least 93.5\% for nominal 95\% intervals in all specifications. Comparing to the other estimators, we find that the CC estimator has the lowest bias and shortest resulting confidence interval widths in all simulation designs. The gains relative to the AS estimator can be large, with the width of the CC estimator CI ranging from 44 to 84\% of that for the AS estimator across designs. The gains relative to the \citet{harshaw_optimized_2026} estimators using the trace and composite objectives are much more modest in Panels A-D (which match the original \citet{harshaw_optimized_2026} design), with the CI width relative to AS typically differing only by single-digit percentage points. We find larger reductions in average length from using the CC estimator in Panels E and F, which consider DGPs with a constant non-zero treatment effect, with differences of at least 10\% of the AS estimator length in all cases and over 30\% in one case. To better understand these results, \Cref{fig:bias-as-function-of-tau} plots the normalized bias of each estimator as a function of the constant effect size $\tau$, which shows that the bias of the CC estimator increases much more slowly in $\tau$ as for the other estimators. To gain intuition, recall that the bias of the CC estimator is proportional to $\Var(\tau_c)$, the variance of the sum of treatment effects across connected components. Hence, if all components were the same size, the CC estimator would be unbiased for all values of $\tau$; in practice, there is modest variation in component size in this application, and so the bias of the CC estimator is not completely insensitive to the magnitude of $\tau$, but much less sensitive than the other estimators. In addition, we find that the precision of the CC estimator is generally similar to that of the two \citet{harshaw_optimized_2026} estimators, and generally smaller than that of AS; the reductions in length do not appear to come at a large cost in terms of estimator precision.  Finally, we comment briefly on computation time. Solving the optimization \eqref{eqn:g-optimization} to compute the trace-optimal and composite-objective-optimal variance estimators for the direct effect took 31 and 2 hours, respectively, on a MacStudio; for the direct vs. indirect contrast, the analogous runtimes were 9 and 2 hours. In comparison, computation of the analogous matrices for the CC and AS approach took under 1 second for both contrasts.\footnote{These runtimes compare how long it takes to compute the matrix $M$ given the second-order treatment probabilities ($\pi_{kk'}$) as input. In practice we calculate these via 5 million simulations, which takes about 5 minutes when run in parallel using 20 cores.} 

\paragraph{Application results.} \Cref{tab:paluck-application} applies the four estimators studied in the simulations to the actual data from \citet{paluck_changing_2016}. We report results for four outcomes, the two studied in the simulation (Wristband and Discipline) and the two additional outcomes reported in Figure 3 of \citet{paluck_changing_2016}. These are an indicator for whether a student's friends talked to them about conflict, and an index capturing whether students perceived social norms against conflict. The first column shows the Horvitz-Thompson point estimates while the remaining columns show the standard errors (SEs) from the four estimators. Overall, we find that the CC estimator SEs are shorter than the AS and trace-optimal SEs for all eight specifications, and smaller than those for the Composite objective in 7 of the 8. The differences relative to AS are large in many cases, with the CC SEs being less than half those of the AS ones in several specifications. The CC SEs are also substantively shorter than the trace-optimal and composite-objective optimal SEs for the two contrasts involving the Wristband outcome. The CC SE is less than half the magnitude of the other estimators for the direct effect, and about one-third shorter than the composite-objective SE for the comparison of the direct and indirect effects. For the latter comparison, the implied $t$-statistic is 2.4 using the CC SE, compared to 1.97 and 1.6 for the trace and composite-objective estimators. The differences for the other outcomes are much smaller in percentage terms. The direct effect for the Wristband outcome is the largest in magnitude (and has the largest $t$-statistic), so these results are consistent with the findings in \Cref{fig:bias-as-function-of-tau} in which the gap between the CC and other estimators grew as the treatment effect became larger. 

\subsection{A setting with spatial treatments\label{subsec:sims-pollmann}}
We next consider simulations based on \citet{pollmann_causal_2026}, who studies settings in which treatments occur at specific points in space and affect nearby locations, similar to \Cref{example:spatial}.

\paragraph{Setting.} Pollmann supposes there is a finite set $\mathcal{S}$ of possible treatment sites, with $T_s = 1$ denoting that site $s$ is treated. Outcomes are measured for a separate set of units indexed by $i$, and each unit's outcome $Y_i$ may be affected by treatment at nearby sites. This is a bipartite experiment in the terminology of \citet{zigler_bipartite_2021}. In his application, Pollmann is interested in the causal effect of having a nearby grocery store on the foot traffic to restaurants during the COVID-19 pandemic; hence, $\mathcal{S}$ indicates possible locations for grocery stores, while $Y_i$ denotes foot traffic to restaurant $i$. Pollmann considers the setting where $\pi_s = P(T_s =1)$ is treated as known, and assumes that the $T_{s}$ are independent across $s$.  For inference, Pollmann imposes that there is some distance $d_0$ such that the outcome for unit $i$ is only affected by sites within distance $d_0$ of $i$ (see his Assumption 3(a)) . We can thus write the potential outcomes as $Y_i(\mathbf{s}_i)$ where $\mathbf{s}_i$ denotes the set of sites within distance $d_0$ of $i$ that are treated---that is $\mathbf{s}_i \subseteq \mathcal{N}_i$, where $\mathcal{N}_i$ is the set of sites within distance $d_0$ of $i$. 

Pollmann considers the H\'ajek-style estimator
\begin{equation}
\hat\tau(d)
=
\frac{
	\sum_{s\in\mathcal S}T_s\sum_i w_i(s,d)Y_i
}{
	\sum_{s\in\mathcal S}T_s\sum_i w_i(s,d)
}
-
\frac{
	\sum_{s\in\mathcal S}(1-T_s)\frac{\pi_s}{1-\pi_s}
	\sum_i w_i(s,d)Y_i
}{
	\sum_{s\in\mathcal S}(1-T_s)\frac{\pi_s}{1-\pi_s}
	\sum_i w_i(s,d)
}, \label{eqn:tauhat-d-pollmann}
\end{equation}
where for a reference distance $d$, $w_i(s,d) \geq 0$ denotes the weight placed on site $s$ for unit $i$. For example, we may set $w_i(s,d)$ to be positive if and only if site $s$ is between 0.075 and 0.1 miles of $i$.  For simplicity of exposition, we will assume that the weights $w_i(s,d)$ are non-zero only if $s \in \mathcal{N}_i$, so we place no weight on sites that are assumed not to affect the outcome of unit $i$. Under appropriate regularity conditions, Pollmann shows that $\hat\tau(d)$ is consistent for the estimand
\begin{equation}
\tau_{\mathrm{marginal}}(d) =  \frac{
	\sum_{\mathbf{s} \subseteq \mathcal{S} } P(\mathbf{S} = \mathbf{s} ) \sum_{s\in \mathbf{s} } \sum_i w_i(s,d) (Y_i(\mathbf{s}_i) -Y_i( \mathbf{s}_i \setminus \{s\} ))
}{
	\sum_{\mathbf{s} \subseteq \mathcal{S} } P(\mathbf{S} = \mathbf{s} ) \sum_{s\in \mathbf{s} } \sum_i w_i(s,d)
} \label{eqn:tau-marginal-def}
\end{equation}
where $\mathbf{S}$ is the random set denoting which sites are treated, and $\mathbf{s}_i = \mathbf{s} \cap \mathcal{N}_i$. In words, $\tau_{\mathrm{marginal}}(d)$ is a weighted average of the marginal effect of including site $s$ in the treated set of sites for nearby restaurants. 

As we describe in detail in \Cref{sec:pollmann-details}, Pollmann also shows that under appropriate regularity conditions, the H\'ajek-style estimator $\hat\tau$ is asymptotically equivalent to a Horvitz-Thompson estimator $\tilde\tau$, where 
$$\tilde\tau - \tau_{\mathrm{marginal}}(d) = \frac{1}{K} \sum_{k} \frac{W_k}{\pi_k} \psi_k, $$
where $k$ indexes pairs of the form $(i,s,\mathbf{s}_i)$, $K$ is the total number of such configurations, $W_k = 1\{\mathbf{S}_i = \mathbf{s}_i\}$, $\pi_k = \P(W_k=1)$, and the $\psi_k$ are fixed scores defined in \eqref{eqn:defn-psi-k} in the appendix. Since $\tilde\tau$ is a Horvitz-Thompson estimator, its variance is given by the expression in \eqref{eqn:var-tauhat}, with $Y_k = \psi_k$. As we describe in \Cref{sec:pollmann-details}, Pollmann obtains an upper bound on $\Var(\tilde\tau)$ by using a weighted version of Young's inequality to upper bound the terms involving $\psi_k \psi_{k'}$ for pairs with $\pi_{kk'} =0$. If the $\psi_k$ were known, it would then be straightforward to estimate this upper bound using a Horvitz-Thompson variance estimator for the pairs with $\pi_{kk'}>0$; one could likewise estimate $\Var(\tilde\tau)$ using the CC estimator. In practice, the scores $\psi_k$ are not known exactly since they depend on $\mu_t(d), \mu_c(d)$, which are the values that are obtained if $Y_i(\mathbf{s}_i) -Y_i( \mathbf{s}_i \setminus \{s\} )$ is replaced with $Y_i(\mathbf{s}_i) $ and $Y_i( \mathbf{s}_i \setminus \{s\} )$, respectively, in the definition of $\tau_{\mathrm{marginal}}(d)$. Pollmann thus considers a feasible version of the variance estimator that replaces $\psi_k$ with $\hat\psi_k$, where $\hat\psi_k$ is the estimated score that replaces $\mu_t(d), \mu_c(d)$ with plug-in estimates $\hat\mu_t(d), \hat\mu_c(d)$ corresponding to the first and second terms in the expression for $\hat\tau(d)$. A plug-in version of the CC estimator can be constructed analogously. We also consider a leave-out version of the Pollmann and CC variance estimators where the score $\hat\psi_{k,-c}$ for potential outcome $k$ in component $c$ is constructed only using observations $k$ that are not in component $c$.

\paragraph{Data.} Pollmann uses proprietary data on foot traffic to businesses from Safegraph. However, he posted on Github a dataset that is ``simulated to resemble the spatial distributions of businesses'' in his data. We base our simulations and application on this publicly available dataset. The data contains information on the foot traffic $Y_i$ to 1,575 restaurants. It also contains the locations of 481 potential sites for grocery stores, of which 167 actually have a grocery store. For each site $s$, the data contain $\pi_s$, the probability that site $s$ would have a grocery store.\footnote{In practice, Pollmann estimates the $\pi_s$ using a convolutional neural network trained on characteristics of the sites. Following Pollmann's ``hypothetical experiment perspective'', however, we treat the $\pi_s$ as known. Pollmann also provides conditions under which one can ignore estimation error in the $\pi_s$ if one uses an estimator based on a doubly-robust score; in his Table I, he finds similar results using the doubly-robust score as he does treating the $\pi_s$ as known. We therefore focus on the hypothetical experiment perspective for simplicity, although one could likewise construct the CC estimator based on the doubly robust score.}

\paragraph{Simulation design.} We consider simulations under the sharp null of no treatment effects based on Pollmann's posted data. That is, the potential outcomes $Y_i(\mathbf{s}_i)$ are set equal to the observed value $Y_i$ in the data for all $\mathbf{s}_i$. We then simulate different data realizations by re-drawing the treatment statuses $T_s \sim \bernoulli{\pi_s}$, independently across $s$. We first calculate Pollmann's variance estimator and the CC estimator based on the true scores $\psi_k$ (which are known since we know $\mu_t(d),\mu_c(d)$) and compare these relative to the estimator $\tilde\tau$. We then calculate the feasible variance estimators for $\hat\tau$ that plug-in the estimated score and compare these to the estimator $\hat\tau$. Comparing these two approaches illustrates the quality of the first-order approximation to the distribution of $\hat\tau$ based on estimated scores. Following the example, on Pollmann's Github page, we impose that there is no treatment effect for sites more than 0.1 miles away ($d_0=.1$), and we consider the weights $w_i(s,d)$ corresponding to indicators for whether the distance between $i$ and $s$ falls in the following bins: $(0,0.025],(0.025,.05], (0.05,0.075], (0.075,.1]$.

We construct the connected components as follows. Recall that the potential outcome for restaurant $i$, $Y_i(\mathbf{s_i})$, depends on the set of grocery stores within distance $d_0=0.1$ miles. It follows that restaurants $i$ and $j$ have some pair of potential outcomes that are jointly unobservable if and only if they are both within distance $d_0$ of the same site, i.e. $\mathcal{N}_i \cap \mathcal{N}_j \neq \emptyset$. The potential outcomes for units $i$ and $j$ are thus in the same connected component if there is a path of restaurants from $i$ to $j$ that each share a neighboring site. When focusing on a particular distance bin, we can restrict attention to the subset of restaurants that fall within that distance of some grocery store, and hence get non-zero weight in the estimand. The number of connected components thus differs across distance bins, ranging from 89 to 121.

\paragraph{Simulation results.} \Cref{tab:pollmann-sims-1bin} shows the simulated results for the $(0.075,0.1]$ bin. Panel A shows the results for estimating the variance of the Horvitz-Thompson estimator $\tilde\tau$ using the true scores. Since the simulation design imposes the sharp null, we expect the CC estimator to be unbiased, and consistent with this, the simulated relative bias of the CC estimator is very close to zero. Coverage for CIs based on the CC estimator is also very close to the nominal level of 95\%. By contrast, the Pollmann variance estimator is upwards-biased by approximately 7 times the actual variance of $\tilde\tau$. This leads to coverage of 100\% in our simulations. The average length of the CI based on the CC estimator is approximately three times smaller than that based on Pollmann's variance estimator. Panel B shows the results using estimated scores as estimators for $\hat\tau$. Using the naive plug-in CC estimator, we find moderate downward-bias for the CC estimator, with coverage dropping to about 92 percent for nominal 95\% intervals. This is natural, since we expect terms of the form $(Y_k - \hat\mu)^2$ to be downward biased for $(Y_k - \mu)^2$ when the outcome $Y_k$ contributes to the estimate $\hat\mu$, owing to overfitting. Using the leave-out estimates for the scores, we find that the CC estimator is now moderately upwards-biased, and coverage is close to the nominal level (94\%). Both the plug-in and leave-out versions of Pollmann's estimator remain extremely conservative, with coverage of 100 percent and average length roughly 3 times that of the CC estimators. We find similar results for the other bins (\Cref{tab:pollmann-sims-all-bins-true,tab:pollmann-sims-all-bins-feasible}).

\paragraph{Additional simulations.} In \Cref{tab:pollmann-sims-hac} in the appendix, we present some extensions of the simulation design in which we consider designs with heterogeneous treatment effects and also report results for spatial HAC estimators (\citet{conley_gmm_1999}) using the choice of bandwidth and kernel recommended in \citet{wang_design-based_2025}. We find that spatial HAC and the CC estimator both have good coverage and yield CIs of a similar length under the sharp null and under a simple form of treatment effect heterogeneity, but that spatial HAC can be strongly downward biased under treatment effect heterogeneity that violates the homophily assumption in \citet{wang_design-based_2025}. The positive semi-definite projection of HAC has coverage above the nominal level in all our simulation designs, but tends to be somewhat more conservative than the CC estimator. In all our simulation designs, Pollmann's variance estimator has large upwards bias, by a factor of 5 or more relative to the true variance.

\paragraph{Application.} \Cref{tab:pollmann-application} shows the point estimates ($\hat\tau$) and SEs when applied to Pollmann's data (which recall is simulated, but calibrated based on the real Safegraph data). We focus on comparing the SEs using leave-out scores, although the results are similar with plug-in scores. For all of the distance bins, we find that the CC SEs are approximately 3 times smaller than those using Pollmann's approach. For the first three bins, the estimates are statistically significantly different from zero at the 5\% level when using the CC SEs, while none of the estimates are significant using Pollmann's SEs. 

\begin{table}[htbp]
	\centering
	\caption{Simulation results based on \citet{paluck_changing_2016} for the direct effect ($(1,1,0)$ vs. $(0,0,0)$ contrast)}
	\label{tab:simulation-results-30}
	\begin{tabular}{lrrrr}
\toprule
Estimator & Bias & Precision & Coverage & Width \\
\midrule
\multicolumn{5}{l}{\textit{Panel A: Outcome: Wristband (Sharp null)}} \\
Aronow--Samii        & 0.8411 & 0.285 & 0.992 & 1.000 \\
Harshaw--trace       & 0.0460 & 0.201 & 0.960 & 0.752 \\
Harshaw--composite   & 0.0182 & 0.181 & 0.951 & 0.743 \\
CC                   & -0.0090 & 0.206 & 0.940 & 0.732 \\
\addlinespace
\multicolumn{5}{l}{\textit{Panel B: Outcome: Disciplinary (Sharp null)}} \\
Aronow--Samii        & 2.0277 & 0.628 & 1.000 & 1.000 \\
Harshaw--trace       & 0.0879 & 0.334 & 0.970 & 0.596 \\
Harshaw--composite   & 0.0432 & 0.339 & 0.955 & 0.583 \\
CC                   & -0.0061 & 0.376 & 0.939 & 0.567 \\
\addlinespace
\multicolumn{5}{l}{\textit{Panel C: Outcome: Mixed POs}} \\
Aronow--Samii        & 2.7559 & 0.700 & 1.000 & 1.000 \\
Harshaw--trace       & 0.7573 & 0.416 & 0.992 & 0.682 \\
Harshaw--composite   & 0.7520 & 0.388 & 0.991 & 0.682 \\
CC                   & 0.6424 & 0.440 & 0.984 & 0.659 \\
\addlinespace
\multicolumn{5}{l}{\textit{Panel D: Outcome: Synthetic Linear}} \\
Aronow--Samii        & 4.0343 & 0.768 & 1.000 & 1.000 \\
Harshaw--trace       & 0.1114 & 0.416 & 0.968 & 0.465 \\
Harshaw--composite   & 0.0277 & 0.347 & 0.954 & 0.448 \\
CC                   & 0.0162 & 0.384 & 0.945 & 0.444 \\
\addlinespace
\multicolumn{5}{l}{\textit{Panel E: Outcome: Wristband (Constant effect = 0.5)}} \\
Aronow--Samii        & 1.4795 & 0.541 & 0.995 & 1.000 \\
Harshaw--trace       & 0.7884 & 0.497 & 0.979 & 0.846 \\
Harshaw--composite   & 0.4934 & 0.414 & 0.972 & 0.773 \\
CC                   & 0.0050 & 0.292 & 0.940 & 0.634 \\
\addlinespace
\multicolumn{5}{l}{\textit{Panel F: Outcome: Disciplinary (Constant effect = 0.5)}} \\
Aronow--Samii        & 2.5925 & 0.862 & 0.999 & 1.000 \\
Harshaw--trace       & 1.5259 & 0.678 & 0.993 & 0.837 \\
Harshaw--composite   & 0.9360 & 0.603 & 0.985 & 0.731 \\
CC                   & 0.0265 & 0.435 & 0.941 & 0.528 \\
\bottomrule
\end{tabular}
	
	\begin{minipage}{0.85\textwidth}
		\vspace{0.5em}
		\footnotesize
		\textit{Notes:} Results are based on 100,000 simulation draws. Bias is the bias of the variance estimator relative to the true variance. Precision is the standard deviation of the variance estimator relative to the true variance. Coverage is the coverage rate of nominal 95\% confidence intervals. Width is average confidence-interval width relative to Aronow--Samii. All estimators use assignment probabilities calculated from 5 million design draws.
	\end{minipage}
\end{table}

\begin{table}[htbp]
	\centering
	\caption{Simulation results based on \citet{paluck_changing_2016} for comparison of direct vs. indirect effect ($(1,1,0)$ vs. $(1,0,1)$ contrast)}
	\label{tab:simulation-results-32}
	\begin{tabular}{lrrrr}
\toprule
Estimator & Bias & Precision & Coverage & Width \\
\midrule
\multicolumn{5}{l}{\textit{Panel A: Outcome: Wristband (Sharp null)}} \\
Aronow--Samii        & 1.3665 & 1.298 & 0.992 & 1.000 \\
Harshaw--trace       & 0.0554 & 0.649 & 0.947 & 0.664 \\
Harshaw--composite   & 0.3428 & 0.672 & 0.970 & 0.756 \\
CC                   & -0.0019 & 0.661 & 0.941 & 0.644 \\
\addlinespace
\multicolumn{5}{l}{\textit{Panel B: Outcome: Disciplinary (Sharp null)}} \\
Aronow--Samii        & 1.3292 & 1.650 & 0.990 & 1.000 \\
Harshaw--trace       & 0.0596 & 0.810 & 0.944 & 0.671 \\
Harshaw--composite   & 0.2876 & 0.793 & 0.962 & 0.749 \\
CC                   & -0.0007 & 0.803 & 0.935 & 0.652 \\
\addlinespace
\multicolumn{5}{l}{\textit{Panel C: Outcome: Mixed POs}} \\
Aronow--Samii        & 0.6537 & 0.787 & 0.983 & 1.000 \\
Harshaw--trace       & 0.1996 & 0.441 & 0.970 & 0.858 \\
Harshaw--composite   & 0.1873 & 0.448 & 0.963 & 0.853 \\
CC                   & 0.1634 & 0.496 & 0.958 & 0.842 \\
\addlinespace
\multicolumn{5}{l}{\textit{Panel D: Outcome: Synthetic Linear}} \\
Aronow--Samii        & 1.3774 & 1.503 & 0.998 & 1.000 \\
Harshaw--trace       & 0.0894 & 0.733 & 0.963 & 0.674 \\
Harshaw--composite   & 0.0401 & 0.645 & 0.958 & 0.662 \\
CC                   & 0.0174 & 0.711 & 0.950 & 0.650 \\
\addlinespace
\multicolumn{5}{l}{\textit{Panel E: Outcome: Wristband (Constant effect = 0.5)}} \\
Aronow--Samii        & 1.6601 & 1.200 & 0.994 & 1.000 \\
Harshaw--trace       & 1.0938 & 0.715 & 0.988 & 0.895 \\
Harshaw--composite   & 0.9633 & 0.707 & 0.987 & 0.865 \\
CC                   & 0.1041 & 0.626 & 0.944 & 0.639 \\
\addlinespace
\multicolumn{5}{l}{\textit{Panel F: Outcome: Disciplinary (Constant effect = 0.5)}} \\
Aronow--Samii        & 1.6382 & 1.377 & 0.993 & 1.000 \\
Harshaw--trace       & 1.2596 & 0.789 & 0.991 & 0.937 \\
Harshaw--composite   & 0.9649 & 0.757 & 0.986 & 0.872 \\
CC                   & 0.1264 & 0.688 & 0.945 & 0.649 \\
\bottomrule
\end{tabular}
	
	\begin{minipage}{0.85\textwidth}
		\vspace{0.5em}
		\footnotesize
		\textit{Notes:} Results are based on 100,000 simulation draws. Bias is the bias of the variance estimator relative to the true variance. Precision is the standard deviation of the variance estimator relative to the true variance. Coverage is the coverage rate of nominal 95\% confidence intervals. Width is average confidence-interval width relative to Aronow--Samii. All estimators use assignment probabilities calculated from 5 million design draws.
	\end{minipage}
\end{table}

\begin{figure}
	\includegraphics[width = 0.9 \linewidth]{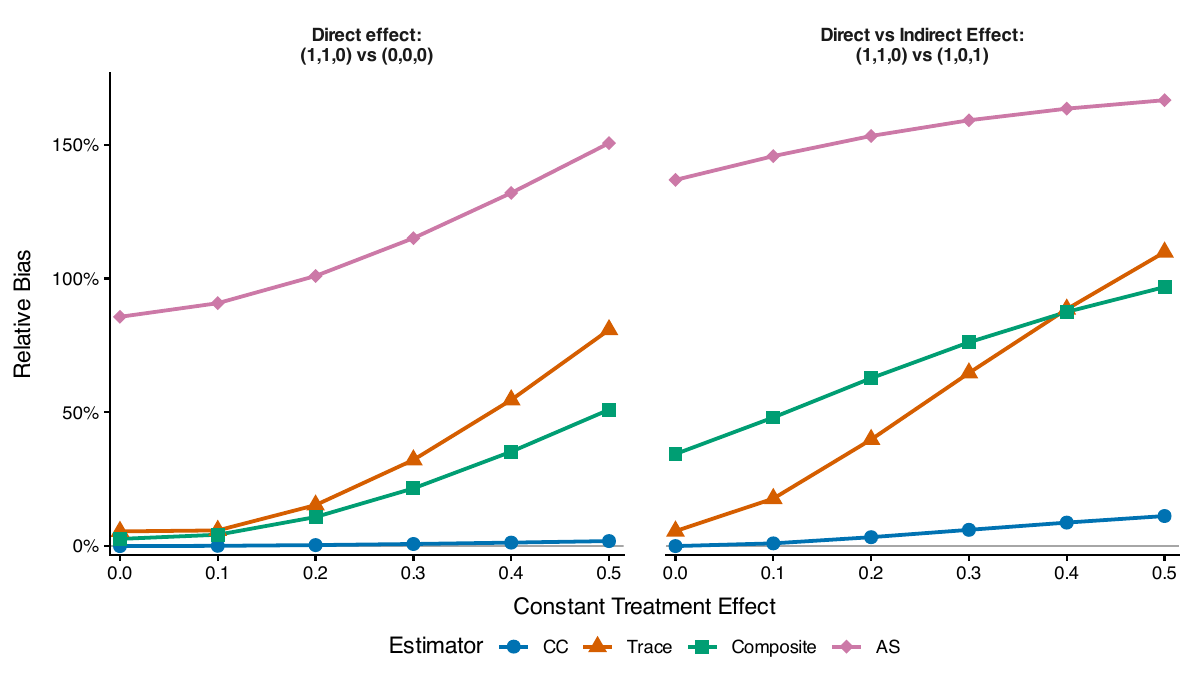}
	\caption{Bias of variance estimators as a function of constant effect $\tau$ using wristband outcome}
	\label{fig:bias-as-function-of-tau}
	\begin{minipage}{0.85\textwidth}
		\vspace{0.5em}
		\footnotesize
		\textit{Notes: } This figure reports the relative bias of the variance estimators as we vary the constant effect $\tau$ for the wristband outcome. $\tau=0$ corresponds to the sharp null (Panel A of \Cref{tab:simulation-results-30,tab:simulation-results-32}) while $\tau=0.5$ corresponds to the results in Panel E of \Cref{tab:simulation-results-30,tab:simulation-results-32}. We calculate the relative bias analytically (by dividing $1/K^2 Y' (M-V) Y$ by $1/K^2 Y' V Y$), although we have verified that the results are quite similar using simulation.
	\end{minipage}
\end{figure}

\begin{table}
	\centering
	\caption{Comparison of standard errors using actual \citet{paluck_changing_2016} data}
	\label{tab:paluck-application}
	\begin{tabular}{lrrrrr}
\toprule
 & & \multicolumn{4}{c}{SE} \\
\cmidrule(lr){3-6}
Outcome & Estimate & CC & Trace & Comp. & AS \\
\midrule
\multicolumn{6}{l}{\textit{Panel A: Direct Effect: $(1,1,0)$ vs. $(0,0,0)$}} \\
Wristband              & 0.546 & 0.043 & 0.111 & 0.091 & 0.128 \\
Discipline             & -0.047 & 0.042 & 0.046 & 0.045 & 0.082 \\
Friends talked         & 0.236 & 0.111 & 0.112 & 0.110 & 0.240 \\
Perceived social norm  & -0.143 & 0.317 & 0.349 & 0.332 & 0.720 \\
\addlinespace
\multicolumn{6}{l}{\textit{Panel B: Direct vs Indirect Effect: $(1,1,0)$ vs. $(1,0,1)$}} \\
Wristband              & 0.144 & 0.060 & 0.073 & 0.089 & 0.111 \\
Discipline             & -0.038 & 0.027 & 0.029 & 0.030 & 0.042 \\
Friends talked         & 0.083 & 0.127 & 0.131 & 0.145 & 0.180 \\
Perceived social norm  & 0.526 & 0.373 & 0.391 & 0.403 & 0.568 \\
\bottomrule
\end{tabular}
 
	\begin{minipage}{0.85\textwidth}
		\vspace{0.5em}
		\footnotesize
		\textit{Notes: } This table reports the Horvitz-Thompson point estimates and SEs for different methods using the actual data from \citet{paluck_changing_2016}. The implementation of the estimators is the same as in the simulations. The Wristband and Discipline outcomes are the same as in the simulations; the remaining two outcomes are taken from Figure 3 of \citet{paluck_changing_2016}: `Friends talked' is an indicator for whether a student talked with friends about conflict; `Perceived social norm' is an index with higher values indicating a student is more anti-conflict. 
	\end{minipage}
\end{table}

\begin{table}
	\caption{Simulation results based on \citet{pollmann_causal_2026}}
	\label{tab:pollmann-sims-1bin}
	\begin{tabular}{lrrrr}
\toprule
Estimator & Bias & Precision & Coverage & Width \\
\midrule
\multicolumn{5}{l}{\textit{Panel A: Variance estimates for $\tilde\tau$ using true scores}} \\
Pollmann & 7.0222 & 2.372 & 1.000 & 1.000 \\
CC & 0.0024 & 1.124 & 0.949 & 0.330 \\
\addlinespace
\multicolumn{5}{l}{\textit{Panel B: Variance estimates for $\hat\tau$ using plug-in/leave-out scores}} \\
Pollmann (plug-in) & 7.4809 & 2.127 & 1.000 & 0.988 \\
Pollmann (leave-out) & 8.1547 & 2.822 & 1.000 & 1.023 \\
CC (plug-in) & -0.0625 & 0.845 & 0.918 & 0.312 \\
CC (leave-out) & 0.2241 & 1.383 & 0.939 & 0.348 \\
\bottomrule
\end{tabular}

	\begin{minipage}{0.85\textwidth}
		\vspace{0.5em}
		\footnotesize
		\textit{Notes:} Results are based on 100,000 simulation draws. Panel A shows results where the variance of the Horvitz-Thompson $\tilde\tau$ is estimated using the true scores $\psi$. Panel B shows results where the variance of the H\'ajek-style estimator $\hat\tau$ is estimated using feasible versions of the scores that replace $\mu_t,\mu_c$ with feasible estimates. We consider two feasible variance estimators, a direct plug-in and leave-one-out version where the estimated value for units in component $c$ excludes units in component $c$. All CI widths are normalized with respect to the CI width for the Pollmann estimator for $\tilde\tau$. Results are shown for the $(0.075,.1]$ mile bin. 
	\end{minipage}
\end{table}

\begin{table}
	\begin{tabular}{lrrrrr}
\toprule
 & & \multicolumn{4}{c}{SE} \\
\cmidrule(lr){3-6}
Distance bin & Estimate & CC-PI & CC-LO & P-PI & P-LO \\
\midrule
$(0.000,0.025]$ & 6.014 & 2.413 & 2.645 & 7.188 & 7.421 \\
$(0.025,0.050]$ & 5.210 & 1.609 & 1.788 & 5.360 & 5.514 \\
$(0.050,0.075]$ & 4.566 & 1.383 & 1.549 & 4.795 & 4.950 \\
$(0.075,0.100]$ & 3.052 & 1.559 & 1.656 & 5.385 & 5.486 \\
\bottomrule
\end{tabular}

	\caption{Application to Pollmann's (simulated) data} \label{tab:pollmann-application}
	\begin{minipage}{0.85\textwidth}
		\vspace{0.5em}
		\footnotesize
		\textit{Notes:} This table shows the estimated causal effect of having a grocery store at the indicated distance bin on restaurant foot traffic. We use the public version of Pollmann's data, which is simulated to match characteristics of the proprietary data he uses in his application. CC-PI and CC-LO denote the SE based on the CC variance estimator using plug-in and leave-out estimates of the scores, respectively. P-PI and P-LO are defined analogously for Pollmann's variance estimator.
	\end{minipage}
\end{table}
\clearpage
\section{Conclusion}

This paper develops a variance estimator that is applicable in general settings where the treatment assignment scheme is known but may exhibit complicated dependence across units. The proposed estimator reduces to the common Neyman variance estimator in a completely randomized experiment and a cluster-robust variance estimator in a cluster-randomized experiment. It moreover generalizes three well-known properties of the Neyman variance estimator to more general designs: it is valid from the design-based perspective, its conservativeness depends on treatment effect heterogeneity, and it is also valid from a super-population sampling perspective. We find the proposed estimator performs well in simulations and two applications. 

Future work could extend the analysis in this paper in a variety of directions. First, we have focused on quadratic-form variance estimators. In the context of a completely randomized experiment, \citet{aronow_sharp_2014} and \citet{abadie_sampling-based_2020} respectively show that one can obtain tighter variance estimators by using distributional information or observable covariates to lower-bound treatment effect heterogeneity and hence the conservativeness of the Neyman variance estimator. Future work might explore whether analogous approaches could be used to lower-bound the conservativeness of the CC estimator in more general designs, which likewise depends on treatment effect heterogeneity. Second, we have focused on settings where the exposure to treatment is discrete. Future work might extend the framework to settings with continuous exposures (possibly under functional form assumptions), which would, for example, incorporate shift-share settings with known shock distributions \citep{borusyak_nonrandom_2023}. Third, while we have shown consistency of the CC estimator for an upper bound on the asymptotic variance of $\hat\tau$, formal justification for normal-based confidence intervals requires a central limit theorem (CLT); such CLTs are currently available for some complicated design-based settings, but not (to our knowledge) for all settings matching the general set-up here. Future work may develop more general CLTs for this setting. Finally, our framework assumes correctly-specified exposure mappings. Future work may explore the properties of the CC estimator when the exposure mapping may be mis-specified, as in \citet{savje_average_2021} or \citet{leung_causal_2022}.

\bibliography{main.bib}

@article{bai_optimality_2022,
	title = {Optimality of {Matched}-{Pair} {Designs} in {Randomized} {Controlled} {Trials}},
	volume = {112},
	issn = {0002-8282},
	url = {https://pubs.aeaweb.org/doi/10.1257/aer.20201856},
	doi = {10.1257/aer.20201856},
	language = {en},
	number = {12},
	urldate = {2026-09-02},
	journal = {American Economic Review},
	author = {Bai, Yuehao},
	month = dec,
	year = {2022},
	pages = {3911--3940},
}

@article{abadie_sampling-based_2020,
	title = {Sampling-{Based} versus {Design}-{Based} {Uncertainty} in {Regression} {Analysis}},
	volume = {88},
	copyright = {© 2020 The Econometric Society},
	issn = {1468-0262},
	url = {https://onlinelibrary.wiley.com/doi/abs/10.3982/ECTA12675},
	doi = {10.3982/ECTA12675},
	language = {en},
	number = {1},
	urldate = {2020-05-06},
	journal = {Econometrica},
	author = {Abadie, Alberto and Athey, Susan and Imbens, Guido W. and Wooldridge, Jeffrey M.},
	year = {2020},
	note = {\_eprint: https://onlinelibrary.wiley.com/doi/pdf/10.3982/ECTA12675},
	pages = {265--296},
}

@article{aronow_sharp_2014,
	title = {Sharp bounds on the variance in randomized experiments},
	volume = {42},
	issn = {0090-5364},
	url = {http://arxiv.org/abs/1405.6555},
	doi = {10.1214/13-AOS1200},
	number = {3},
	urldate = {2021-02-01},
	journal = {The Annals of Statistics},
	author = {Aronow, Peter M. and Green, Donald P. and Lee, Donald K. K.},
	month = jun,
	year = {2014},
	note = {arXiv: 1405.6555},
	pages = {850--871},
}

@article{li_asymptotic_2018,
	title = {Asymptotic theory of rerandomization in treatment–control experiments},
	volume = {115},
	url = {https://www.pnas.org/doi/10.1073/pnas.1808191115},
	doi = {10.1073/pnas.1808191115},
	number = {37},
	urldate = {2026-08-31},
	journal = {Proceedings of the National Academy of Sciences},
	publisher = {Proceedings of the National Academy of Sciences},
	author = {Li, Xinran and Ding, Peng and Rubin, Donald B.},
	month = sep,
	year = {2018},
	pages = {9157--9162},
}

@misc{johari_experimental_2021,
	title = {Experimental {Design} in {Two}-{Sided} {Platforms}: {An} {Analysis} of {Bias}},
	shorttitle = {Experimental {Design} in {Two}-{Sided} {Platforms}},
	url = {http://arxiv.org/abs/2002.05670},
	doi = {10.48550/arXiv.2002.05670},
	urldate = {2026-08-31},
	publisher = {arXiv},
	author = {Johari, Ramesh and Li, Hannah and Liskovich, Inessa and Weintraub, Gabriel},
	month = sep,
	year = {2021},
	note = {arXiv:2002.05670 [stat.ME]},
}

@article{abdulkadiroglu_research_2017,
	title = {Research {Design} {Meets} {Market} {Design}: {Using} {Centralized} {Assignment} for {Impact} {Evaluation}},
	volume = {85},
	issn = {0012-9682},
	shorttitle = {Research {Design} {Meets} {Market} {Design}},
	url = {https://www.econometricsociety.org/doi/10.3982/ECTA13925},
	doi = {10.3982/ECTA13925},
	language = {en},
	number = {5},
	urldate = {2026-08-31},
	journal = {Econometrica},
	author = {Abdulkadiroğlu, Atila and Angrist, Joshua D. and Narita, Yusuke and Pathak, Parag A.},
	year = {2017},
	pages = {1373--1432},
}

@article{chen_nonparametric_2026,
	title = {Nonparametric treatment effect identification in school choice},
	volume = {253},
	issn = {0304-4076},
	url = {https://www.sciencedirect.com/science/article/pii/S0304407625002258},
	doi = {10.1016/j.jeconom.2025.106172},
	urldate = {2026-08-31},
	journal = {Journal of Econometrics},
	author = {Chen, Jiafeng},
	month = jan,
	year = {2026},
	pages = {106172},
}

@misc{chang_design-based_2025,
	title = {Design-based {Estimation} {Theory} for {Complex} {Experiments}},
	url = {http://arxiv.org/abs/2311.06891},
	doi = {10.48550/arXiv.2311.06891},
	urldate = {2026-08-16},
	publisher = {arXiv},
	author = {Chang, Haoge},
	month = may,
	year = {2025},
	note = {arXiv:2311.06891 [econ.EM]},
}

@article{xu_design-based_2026,
	title = {A design-based approach to spatial correlation},
	volume = {257},
	issn = {0304-4076},
	url = {https://www.sciencedirect.com/science/article/pii/S0304407626001314},
	doi = {10.1016/j.jeconom.2026.106310},
	urldate = {2026-08-16},
	journal = {Journal of Econometrics},
	author = {Xu, Ruonan and Wooldridge, Jeffrey M.},
	month = sep,
	year = {2026},
	pages = {106310},
}

@article{zigler_bipartite_2021,
	title = {Bipartite {Causal} {Inference} with {Interference}},
	volume = {36},
	issn = {0883-4237, 2168-8745},
	url = {https://projecteuclid.org/journals/statistical-science/volume-36/issue-1/Bipartite-Causal-Inference-with-Interference/10.1214/19-STS749.full},
	doi = {10.1214/19-STS749},
	language = {en},
	number = {1},
	urldate = {2026-08-07},
	journal = {Statistical Science},
	publisher = {Institute of Mathematical Statistics},
	author = {Zigler, Corwin M. and Papadogeorgou, Georgia},
	month = feb,
	year = {2021},
	pages = {109--123},
}

@book{imbens_causal_2015,
	address = {Cambridge},
	title = {Causal {Inference} for {Statistics}, {Social}, and {Biomedical} {Sciences}: {An} {Introduction}},
	isbn = {978-0-521-88588-1},
	shorttitle = {Causal {Inference} for {Statistics}, {Social}, and {Biomedical} {Sciences}},
	url = {https://www.cambridge.org/core/books/causal-inference-for-statistics-social-and-biomedical-sciences/71126BE90C58F1A431FE9B2DD07938AB},
	doi = {10.1017/CBO9781139025751},
	urldate = {2020-07-23},
	publisher = {Cambridge University Press},
	author = {Imbens, Guido W. and Rubin, Donald B.},
	year = {2015},
}

@article{kernighan_efficient_1970,
	title = {An efficient heuristic procedure for partitioning graphs},
	volume = {49},
	issn = {0005-8580},
	url = {https://ieeexplore.ieee.org/document/6771089},
	doi = {10.1002/j.1538-7305.1970.tb01770.x},
	number = {2},
	urldate = {2026-08-06},
	journal = {The Bell System Technical Journal},
	author = {Kernighan, B. W. and Lin, S.},
	month = feb,
	year = {1970},
	pages = {291--307},
}

@incollection{buluc_recent_2016,
	address = {Cham},
	title = {Recent {Advances} in {Graph} {Partitioning}},
	isbn = {978-3-319-49487-6},
	url = {https://doi.org/10.1007/978-3-319-49487-6_4},
	doi = {10.1007/978-3-319-49487-6_4},
	language = {en},
	urldate = {2026-08-06},
	booktitle = {Algorithm {Engineering}: {Selected} {Results} and {Surveys}},
	publisher = {Springer International Publishing},
	author = {Buluç, Aydın and Meyerhenke, Henning and Safro, Ilya and Sanders, Peter and Schulz, Christian},
	editor = {Kliemann, Lasse and Sanders, Peter},
	year = {2016},
	pages = {117--158},
}

@article{samii_equivalencies_2012,
	title = {On equivalencies between design-based and regression-based variance estimators for randomized experiments},
	volume = {82},
	copyright = {https://www.elsevier.com/tdm/userlicense/1.0/},
	issn = {01677152},
	url = {https://linkinghub.elsevier.com/retrieve/pii/S0167715211003452},
	doi = {10.1016/j.spl.2011.10.024},
	language = {en},
	number = {2},
	urldate = {2026-08-06},
	journal = {Statistics \& Probability Letters},
	author = {Samii, Cyrus and Aronow, Peter M.},
	month = feb,
	year = {2012},
	pages = {365--370},
}

@misc{zhong_unconditional_2024,
	title = {Unconditional {Randomization} {Tests} for {Interference}},
	url = {http://arxiv.org/abs/2409.09243},
	doi = {10.48550/arXiv.2409.09243},
	urldate = {2026-08-06},
	publisher = {arXiv},
	author = {Zhong, Liang},
	month = sep,
	year = {2024},
	note = {arXiv:2409.09243 [econ.EM]
	version: 1},
}

@article{leung_treatment_2020,
	title = {Treatment and {Spillover} {Effects} {Under} {Network} {Interference}},
	volume = {102},
	issn = {0034-6535, 1530-9142},
	url = {https://direct.mit.edu/rest/article/102/2/368-380/96745},
	doi = {10.1162/rest_a_00818},
	language = {en},
	number = {2},
	urldate = {2026-08-06},
	journal = {The Review of Economics and Statistics},
	author = {Leung, Michael P.},
	month = may,
	year = {2020},
	pages = {368--380},
}

@article{baldi_asymptotic_1989,
	title = {Asymptotic normality of some {Graph}-{Related} statistics},
	volume = {26},
	issn = {0021-9002, 1475-6072},
	url = {https://www.cambridge.org/core/journals/journal-of-applied-probability/article/abs/asymptotic-normality-of-some-graphrelated-statistics/FC7075BE0A6E27414E6390A01A03917B?utm_source=chatgpt.com},
	doi = {10.2307/3214327},
	language = {en},
	number = {1},
	urldate = {2026-08-06},
	journal = {Journal of Applied Probability},
	author = {Baldi, Pierre and Rinott, Yosef},
	month = mar,
	year = {1989},
	pages = {171--175},
}

@misc{bhadra_causal_2025,
	title = {Causal {Inference} {Under} {Network} {Interference}},
	url = {http://arxiv.org/abs/2508.06808},
	doi = {10.48550/arXiv.2508.06808},
	urldate = {2026-08-06},
	publisher = {arXiv},
	author = {Bhadra, Subhankar and Schweinberger, Michael},
	month = aug,
	year = {2025},
	note = {arXiv:2508.06808 [stat.ME]},
}

@article{savje_average_2021,
	title = {Average treatment effects in the presence of unknown interference},
	volume = {49},
	issn = {0090-5364, 2168-8966},
	url = {https://projecteuclid.org/journals/annals-of-statistics/volume-49/issue-2/Average-treatment-effects-in-the-presence-of-unknown-interference/10.1214/20-AOS1973.full},
	doi = {10.1214/20-AOS1973},
	language = {en},
	number = {2},
	urldate = {2026-08-06},
	journal = {The Annals of Statistics},
	publisher = {Institute of Mathematical Statistics},
	author = {Sävje, Fredrik and Aronow, Peter M. and Hudgens, Michael G.},
	month = apr,
	year = {2021},
	pages = {673--701},
}

@article{conley_gmm_1999,
	title = {{GMM} estimation with cross sectional dependence},
	volume = {92},
	url = {https://www.sciencedirect.com/science/article/pii/S0304407698000840},
	number = {1},
	urldate = {2026-08-06},
	journal = {Journal of econometrics},
	publisher = {Elsevier},
	author = {Conley, Timothy G.},
	year = {1999},
	pages = {1--45},
}

@article{leung_causal_2022,
	title = {Causal {Inference} {Under} {Approximate} {Neighborhood} {Interference}},
	volume = {90},
	url = {https://onlinelibrary.wiley.com/doi/10.3982/ECTA17841},
	doi = {10.3982/ECTA17841},
	language = {en},
	number = {1},
	urldate = {2026-08-06},
	journal = {Econometrica},
	author = {Leung, Michael P},
	year = {2022},
}

@article{wang_design-based_2025,
	title = {Design-based inference for spatial experiments under unknown interference},
	volume = {19},
	issn = {1932-6157, 1941-7330},
	url = {https://projecteuclid.org/journals/annals-of-applied-statistics/volume-19/issue-1/Design-based-inference-for-spatial-experiments-under-unknown-interference/10.1214/24-AOAS1986.full},
	doi = {10.1214/24-AOAS1986},
	language = {en},
	number = {1},
	urldate = {2026-08-06},
	journal = {The Annals of Applied Statistics},
	publisher = {Institute of Mathematical Statistics},
	author = {Wang, Ye and Samii, Cyrus and Chang, Haoge and Aronow, P. M.},
	month = mar,
	year = {2025},
	pages = {744--768},
}

@article{basse_randomization_2019,
	title = {Randomization tests of causal effects under interference},
	volume = {106},
	issn = {0006-3444},
	url = {https://www.jstor.org/stable/48546381},
	number = {2},
	urldate = {2026-08-06},
	journal = {Biometrika},
	publisher = {[Oxford University Press, Biometrika Trust]},
	author = {Basse, G. W. and Feller, A. and Toulis, P.},
	year = {2019},
	pages = {487--494},
}

@article{athey_exact_2018,
	title = {Exact p-{Values} for {Network} {Interference}},
	volume = {113},
	issn = {0162-1459},
	url = {https://doi.org/10.1080/01621459.2016.1241178},
	doi = {10.1080/01621459.2016.1241178},
	number = {521},
	urldate = {2026-08-06},
	journal = {Journal of the American Statistical Association},
	publisher = {Taylor \& Francis},
	author = {Athey, Susan and Eckles, Dean and Imbens, Guido W.},
	month = jan,
	year = {2018},
	note = {\_eprint: https://doi.org/10.1080/01621459.2016.1241178},
	pages = {230--240},
}

@article{paluck_changing_2016,
	title = {Changing climates of conflict: {A} social network experiment in 56 schools},
	volume = {113},
	shorttitle = {Changing climates of conflict},
	url = {https://www.pnas.org/doi/10.1073/pnas.1514483113},
	doi = {10.1073/pnas.1514483113},
	number = {3},
	urldate = {2026-07-27},
	journal = {Proceedings of the National Academy of Sciences},
	publisher = {Proceedings of the National Academy of Sciences},
	author = {Paluck, Elizabeth Levy and Shepherd, Hana and Aronow, Peter M.},
	month = jan,
	year = {2016},
	pages = {566--571},
}

@article{neyman_application_1923,
	title = {On the {Application} of {Probability} {Theory} to {Agricultural} {Experiments}. {Essay} on {Principles}. {Section} 9.},
	volume = {5},
	issn = {0883-4237},
	url = {https://www.jstor.org/stable/2245382},
	number = {4},
	urldate = {2020-04-21},
	journal = {Statistical Science},
	publisher = {Institute of Mathematical Statistics},
	author = {Neyman, Jerzy},
	year = {1923},
	pages = {465--472},
}

@article{horvitz_generalization_1952,
	title = {A {Generalization} of {Sampling} {Without} {Replacement} {From} a {Finite} {Universe}},
	volume = {47},
	issn = {0162-1459},
	url = {https://www.jstor.org/stable/2280784},
	doi = {10.2307/2280784},
	number = {260},
	urldate = {2026-07-14},
	journal = {Journal of the American Statistical Association},
	publisher = {[American Statistical Association, Taylor \& Francis, Ltd.]},
	author = {Horvitz, D. G. and Thompson, D. J.},
	year = {1952},
	pages = {663--685},
}

@article{borusyak_nonrandom_2023,
	title = {Nonrandom {Exposure} to {Exogenous} {Shocks}},
	volume = {91},
	copyright = {© 2023 The Econometric Society},
	issn = {1468-0262},
	url = {https://onlinelibrary.wiley.com/doi/abs/10.3982/ECTA19367},
	doi = {10.3982/ECTA19367},
	language = {en},
	number = {6},
	urldate = {2024-04-19},
	journal = {Econometrica},
	author = {Borusyak, Kirill and Hull, Peter},
	year = {2023},
	note = {\_eprint: https://onlinelibrary.wiley.com/doi/pdf/10.3982/ECTA19367},
	pages = {2155--2185},
}

@article{Deryugina2017,
  author       = {Tatyana Deryugina},
  title        = {The Fiscal Cost of Hurricanes: Disaster Aid Versus Social Insurance},
  journal      = {American Economic Journal: Economic Policy},
  year         = {2017},
  volume       = {9},
  number       = {3},
  pages        = {168--198},
  doi          = {10.1257/pol.20140296},
  url          = {https://www.aeaweb.org/articles?id=10.1257/pol.20140296}
}

@article{harshaw_optimized_2026,
	title = {Optimized {Variance} {Estimation} under {Interference} and {Complex} {Experimental} {Designs}},
	volume = {0},
	issn = {0162-1459},
	url = {https://doi.org/10.1080/01621459.2026.2627027},
	doi = {10.1080/01621459.2026.2627027},
	number = {0},
	urldate = {2026-06-22},
	journal = {Journal of the American Statistical Association},
	publisher = {Taylor \& Francis},
	author = {Harshaw, Christopher and Middleton, Joel and Sävje, Fredrik},
	month = may,
	year = {2026},
	note = {\_eprint: https://doi.org/10.1080/01621459.2026.2627027},
	pages = {1--14},
}

@article{pollmann_causal_2026,
	title = {Causal {Inference} for {Spatial} {Treatments}},
	url = {http://arxiv.org/abs/2011.00373},
	doi = {10.48550/arXiv.2011.00373},
	journal = {Econometrica},
	author = {Pollmann, Michael},
	year = {2026},
}

@article{dell_development_2020,
	title = {The {Development} {Effects} of the {Extractive} {Colonial} {Economy}: {The} {Dutch} {Cultivation} {System} in {Java}},
	volume = {87},
	issn = {0034-6527},
	shorttitle = {The {Development} {Effects} of the {Extractive} {Colonial} {Economy}},
	url = {https://doi.org/10.1093/restud/rdz017},
	doi = {10.1093/restud/rdz017},
	number = {1},
	urldate = {2026-08-06},
	journal = {The Review of Economic Studies},
	author = {Dell, Melissa and Olken, Benjamin A},
	month = jan,
	year = {2020},
	pages = {164--203},
}

@article{aronow_estimating_2017,
	title = {Estimating {Average} {Causal} {Effects} {Under} {General} {Interference}, with {Application} to a {Social} {Network} {Experiment}},
	volume = {11},
	issn = {1932-6157},
	url = {https://www.jstor.org/stable/26362172},
	number = {4},
	urldate = {2025-05-29},
	journal = {The Annals of Applied Statistics},
	author = {Aronow, Peter M. and Samii, Cyrus},
	year = {2017},
	note = {Publisher: Institute of Mathematical Statistics},
	pages = {1912--1947},
}

@article{delevoye_consistency_2020,
	title = {Consistency of the {Horvitz}–{Thompson} estimator under general sampling and experimental designs},
	volume = {207},
	issn = {03783758},
	url = {https://linkinghub.elsevier.com/retrieve/pii/S0378375819301211},
	doi = {10.1016/j.jspi.2019.12.002},
	language = {en},
	urldate = {2025-05-29},
	journal = {Journal of Statistical Planning and Inference},
	author = {Delevoye, Angèle and Sävje, Fredrik},
	month = jul,
	year = {2020},
	pages = {190--197},
}

@misc{chandrasekhar_general_2024,
	title = {General {Covariance}-{Based} {Conditions} for {Central} {Limit} {Theorems} with {Dependent} {Triangular} {Arrays}},
	url = {http://arxiv.org/abs/2308.12506},
	doi = {10.48550/arXiv.2308.12506},
	urldate = {2026-08-06},
	publisher = {arXiv},
	author = {Chandrasekhar, Arun G. and Jackson, Matthew O. and McCormick, Tyler H. and Thiyageswaran, Vydhourie},
	month = oct,
	year = {2024},
	note = {arXiv:2308.12506 [stat.ME]},
}
\appendix

\section{Derivations for Individual- and Cluster-Randomized Experiments\label{sec:cre-derivations}}

We now derive the forms of the variance estimators under individual- and cluster-level complete randomization, as presented in \Cref{sec:comparison-in-experiments}. To avoid repetition, we derive the form of each estimator in a cluster-randomized experiment where each cluster has size $m = N/G \geq 1$; the individual-level CRE is then the special case where $m=1$. 

As in \Cref{example:cluster-experiment}, we suppose there are $G$ clusters, and denote by $g(i)$ the cluster of unit $i$. We assume there are at least two clusters in each treatment arm, $G_0,G_1 \geq 2$. We suppose each cluster has $m=N/G$ units. We index the potential outcomes by $k=(i,d)$, and let $Y_k = 2 Y_i(1)$ if $d=1$ and $-2Y_i(0)$ if $d=0$, so that $\tau = \frac{1}{K} \sum_k Y_k = \frac{1}{N} \sum_i Y_i(1)-Y_i(0)$. We have $W_k = 1\{D_{g(i)} =d\}$. For cluster $g$, we define the averages
\[\bar{Y}_g(d) = \frac{1}{m} \sum_{i: g(i) =g} Y_i(d), \qquad \bar{Y}_g =  \frac{1}{m} \sum_{i: g(i) =g} Y_i = \bar{Y}_g(D_g).\]
All of the estimators that we consider take the form in \eqref{eqn:hat-v-m}, where $m_{kk'}=v_{kk'} +t_{kk'}$ for $t_{kk'}$ the elements of the slack matrix $T(M) = M-V$. We now solve for $v_{kk'}$ and $t_{kk'}$ for the relevant estimators. Let $k=(i,d)$ and $k'=(i',d')$, and suppose $g(i)=g$ and $g(i')=h$. Then standard calculations for the assignment probabilities in a cluster-randomized CRE yield
\[\pi_{kk'} = \begin{cases*}
	\frac{G_d}{G} & $g=h, d=d'$ \\
	0 & $g=h,d \neq d'$ \\	
	\frac{G_d }{G} \frac{G_d - 1}{G-1} & $g\neq h, d=d'$ \\
	\frac{G_d}{G} \frac{G_{d'}}{G-1} & $g \neq h, d \neq d'$
\end{cases*}, \qquad \qquad
 v_{kk'} = \begin{cases*}
	\frac{G_{1-d}}{G_d} & $g=h, d=d'$ \\
	-1 & $g=h,d \neq d'$ \\	
	\frac{-G_{1-d} }{G_d (G-1)} & $g\neq h, d=d'$ \\
	\frac{1}{G-1} & $g \neq h, d \neq d'$
\end{cases*}.\]
Recall that each connected component collects all the potential outcomes for a particular cluster, so $C=G$ and $k,k'$ are in the same connected component if and only if $g=h$. Hence, for the CC estimator we have
\[t_{kk'}^{CC} = s_{kk'} = \begin{cases*}
	1 & $g=h$ \\
	-1/(G-1) & $g \neq h$ 
\end{cases*}.\]
The direct AS estimator has slack $T^{AS} = D + A$ where $D$ is the degree matrix of $A$. Note that any potential outcome $(i,d)$ is jointly unobservable with the opposite-treatment potential outcomes for units in the same cluster, and therefore $D_{kk}=m$ for all $k$. Likewise, $A_{kk'}=1$ if and only if $g=h, d \neq d'$. It follows that 
\[t_{kk'}^{AS} = \begin{cases*}
	m & $k=k'$ \\
	1 & $g =h, d \neq d'$ \\
	0 & o.w.
\end{cases*}.\]
Finally, consider the collapsed-AS estimator, applied to the cluster sums $Y^{tot}_g(d) = \sum_{k=(i,d), g(i)=g} Y_k = 2 (2d-1) m \bar{Y}_g(d)$. (Since all clusters are the same size, applying collapsed-AS to cluster-sums and cluster-averages is equivalent.) Note that $Y^{tot}_{g}(d)$ and $Y^{tot}_{g'}(d')$ are jointly unobservable if and only if $g=g', d=1-d'$. Hence, if we consider the ``collapsed'' joint unobservability graph applied to cluster-level potential outcomes $(g,d)$, every observation has degree 1, and the adjacency graph $A^G$ has $A^G_{(g,d),(g',d')}=1$ if and only if $g=g', d=1-d'$. The conservativeness of the collapsed-AS estimator can thus be written as $ \frac{1}{K^2}\mathbf{Y}_{tot}' (I_{2G} + A^G) \mathbf{Y}_{tot}$, where $\mathbf{Y}_{tot}$ collects the  $Y^{tot}_{g}(d)$. Note that $\mathbf{Y}_{tot} = \tilde{H}'Y$, where $\tilde{H}$ is the $2N \times 2G$ matrix indicating cluster-by-$d$ membership  (equivalently, connected component-by-$d$ membership), i.e. $\tilde{H}_{(i,d),(g,d')} = 1\{g(i) =g, d=d' \}$ . Pulling back to the individual level potential outcomes, we see that the conservativeness of the collapsed-AS estimator is $\frac{1}{K^2}Y' T^{\textrm{AS,collapsed}} Y$ where $T^{\textrm{AS,collapsed}} = \tilde{H} (I_{2G} + A^G) \tilde{H}'$, which has elements 
 \[t_{kk'}^{\textrm{AS,collapsed}} = \begin{cases*}
 	1 & $g=h$ \\
 	0 & o.w.
 \end{cases*}.\]
Note that $T^{\textrm{AS,collapsed}}$ can equivalently be written as $HH'$, where $H$ is the $2N \times G$ cluster (equivalently, component) membership matrix. 

We now solve for the variance estimators. Observe that
\[v_{kk'} + t^{CC}_{kk'} = \begin{cases*}
	\frac{G}{G_d} & $g=h,d=d'$ \\
	0 & $g=h, d \neq d'$ \\
	-\frac{G}{G_d (G-1)} & $g \neq h, d=d'$ \\
	0 & $g \neq h, d \neq d'$
\end{cases*}. \]
Hence, the CC estimator puts non-zero weight on observed pairs $Y_k Y_{k'}$ that share the same treatment status, so we can write $$\VhatCC  = \frac{1}{K^2} \sum_{k,k':k=(i,1),k'=(i',1)} \frac{W_k W_{k'}}{ \pi_{kk'} } (v_{kk'} + t^{CC}_{kk'} ) Y_k Y_{k'} + \frac{1}{K^2} \sum_{k,k':k=(i,0),k'=(i',0)} \frac{W_k W_{k'}}{ \pi_{kk'} } (v_{kk'} + t^{CC}_{kk'})  Y_k Y_{k'} =: \widehat{V}_{\mathrm{CC},1} + \widehat{V}_{\mathrm{CC},0}.$$
Focusing on the first term, and using the expressions derived above, we see that
\begin{align*}
	\widehat{V}_{\mathrm{CC},1}  = \frac{1}{N^2} \sum_{i,i' \in \mathcal{I}_1 : g(i) = g(i')} \bigg( \frac{G}{G_1} \bigg)^2 Y_i Y_{i'} -  \frac{1}{N^2} \sum_{i,i' \in \mathcal{I}_1: g(i) \neq g(i')} \bigg( \frac{G}{G_1} \bigg)^2 \frac{1}{G_1-1} Y_i Y_{i'}
\end{align*}
where $\mathcal{I}_1 := \{i \,:\, D_{g(i)}=1\}$ is the set of treated indices. Note that
\[ \frac{1}{N^2} \sum_{i,i' \in \mathcal{I}_1 : g(i) = g(i')}  Y_i Y_{i'}  = \frac{m^2}{N^2}  \sum_{g: D_g= 1}  \bar{Y}_g^2 = \bigg( \frac{G_1}{G} \bigg)^2 \frac{1}{G_1^2}  \sum_{g: D_g= 1}  \bar{Y}_g^2 .  \]
Likewise, 
\begin{align*}
\frac{1}{N^2}
\sum_{i,i' \in \mathcal{I}_1: g(i) \neq g(i')}
Y_iY_{i'}
&=
\frac{1}{N^2}
\sum_{i,i' \in \mathcal{I}_1}
Y_iY_{i'}
-
\frac{1}{N^2}
\sum_{i,i' \in \mathcal{I}_1:\, g(i)=g(i')}
Y_iY_{i'} \\
&=  \bigg(\frac{G_1}{G}\bigg)^2  \bar{Y}_1^2 - \bigg( \frac{G_1}{G} \bigg)^2 \frac{1}{G_1^2}  \sum_{g: D_g= 1}  \bar{Y}_g^2 ,
\end{align*}
where $\bar{Y}_1 = \frac{1}{G_1 m} \sum_{i:D_{g(i)}=1} Y_i = \frac{1}{G_1} \sum_{g:D_g=1} \bar{Y}_g$ is the sample average among treated units (equivalently, among treated clusters). Putting the previous displays together, we see that
\[\widehat{V}_{\mathrm{CC},1}  = \frac{1}{G_1} \bigg( \frac{1}{G_1 -1} \sum_{g:D_g =1} \bar{Y}_g^2 - \frac{G_1}{G_1-1} \bar{Y}_1^2 \bigg) = \frac{1}{G_1} \hat{s}^2_{cl,1}
.   \]
By analogous argument, we see that $\widehat{V}_{\mathrm{CC},0} = \frac{1}{G_0} \hat{s}^2_{cl,0}$, and hence $\VhatCC = \frac{1}{G_1} \hat{s}^2_{cl,1} + \frac{1}{G_0} \hat{s}^2_{cl,0} = \VhatNeymanCl$.

Next observe that $r_{kk'} := t^{\textrm{AS,collapsed}}_{kk'} - t^{CC}_{kk'} = 1/(G-1)$ if $g \neq h$ and is equal to zero otherwise. Hence, $\VhatAScollapsed  = \VhatNeymanCl  + \frac{1}{K^2} \sum_{k,k': \pi_{kk'}>0} \frac{W_k W_{k'}}{\pi_{kk'}} r_{kk'} Y_k Y_{k'}$. The latter summation can be broken into three sums, two involving same-treatment pairs and one involving opposite-treatment pairs. The same-treatment pairs take the form 
\[ \frac{1}{N^2} \sum_{i,i' \in \mathcal{I}_d, g(i) \neq g(i')} \frac{G}{G_d (G_d-1)} Y_{i} Y_{i'} = \frac{1}{G} \frac{1}{G_d (G_d-1)} \sum_{g,g': g \neq g', D_g = D_{g'}=d} \bar{Y}_g \bar{Y}_{g'} \]

Now, observe that 
$$\sum_{g,g': g \neq g', D_g = D_{g'}=d}  \bar{Y}_g \bar{Y}_{g'}  = \sum_{g,g': D_g = D_{g'}=d} \bar{Y}_{g} \bar{Y}_{g'} - \sum_{g: D_g = d} \bar{Y}_g^2 = G_d^2 \bar{Y}_d^2 - \sum_{g:D_g=d} \bar{Y}_g^2 = -(G_d-1) \hat{s}^2_{cl,d} + (G_d -1)G_d \bar{Y}_d^2$$
and hence the cross-cluster terms reduce to 
\[ \frac{1}{G} \left[ \bar{Y}_d^2 - \frac{ \hat{s}_{cl,d}^2 }{G_d}  \right]  .\]
The opposite-arm terms take the form
\[ \frac{2}{N^2} \sum_{i \in \mathcal{I}_0, i' \in \mathcal{I}_1, g(i) \neq g(i') } \frac{G}{G_1 G_0} -Y_{i} Y_{i'} = \frac{2}{G} \frac{1}{G_1 G_0} \sum_{g,g':D_g=1, D_{g'}=0}- \bar{Y}_g \bar{Y}_{g'}  =- \frac{2}{G}  \bar{Y}_1 \bar{Y}_0  . \]
Combining terms, we see that 
\[\VhatAScollapsed  = \VhatNeymanCl + \frac{1}{G} \left( \bar{Y}_1^2 + \bar{Y}_0^2 - 2 \bar{Y}_1 \bar{Y}_0 - \frac{\hat{s}_{cl,1}^2 }{G_1} - \frac{\hat{s}_{cl,0}^2 }{G_0}  \right)  = (1- 1/G) \VhatNeymanCl + \frac{1}{G} \hat\tau^2, \]
where we use the fact that $\hat\tau^2 = (\bar{Y}_1 - \bar{Y}_0 )^2$. 

Finally, note that $\VhatAS=\VhatAScollapsed$ when $m=1$, while when $m>1$, we have that
\[r^{AS}_{kk'}:= t^{AS}_{kk'}- t^{\textrm{AS,collapsed}}_{kk'} = \begin{cases}
	m-1 & k=k' \\
	-1 & g=h, d=d' , k \neq k'\\
	0 & \text{ o.w. }
\end{cases}.\]
We thus see that 
\[\VhatAS = \VhatAScollapsed +  \frac{1}{N^2} \sum_{i,i' \in \mathcal{I}_1,g(i)=g(i')} \frac{G}{G_1}  (m 1\{i=i'\} - 1) Y_i Y_{i'} + \frac{1}{N^2} \sum_{i,i' \in \mathcal{I}_0, g(i) = g(i')} \frac{G}{G_0}  (m 1\{i=i'\} - 1) Y_i Y_{i'}.\]
Focusing on the second term on the right-hand side, we see that it can be written as 
$$ \frac{1}{N^2}  \frac{G}{G_1}  \sum_{g:D_g= 1} \left( \sum_{i: g(i)=g} m Y_i^2 - \sum_{i,i':g(i)=g(i')=g} Y_{i} Y_{i'} \right)  =  \frac{1}{N^2}  \frac{G}{G_1}  \sum_{g:D_g= 1} \left(m^2 \frac{m-1}{m} s_g^2 \right)= \frac{1}{G^2} \frac{G}{G_1} \frac{m-1}{m} \sum_{g:D_g=1}s_g^2.$$
By an analogous argument for the third term, we obtain that when $m>1$,
\[\VhatAS = \VhatAScollapsed +  \frac{m-1}{m G^2} \left( \frac{G}{G_1} \sum_{g:D_g=1} s_g^2 + \frac{G}{G_0} \sum_{g:D_g=0} s_g^2 \right) = \VhatAScollapsed +\VhatWithin. \]

The following results formalize the trace-optimal and Frobenius-optimal solutions in the cluster-randomized experiment (again, nesting the individual-CRE as the special case $m=1$). 

\begin{prop} \label{prop:trace-optimal-cluster-exp}
	Consider a cluster-randomized experiment, as in \Cref{example:cluster-experiment}, with $G$ clusters each of size $m = N/G$. Suppose that $G_0,G_1 \geq 2$. Then the estimator $\VhatM$ of the form in \eqref{eqn:hat-v-m} is trace-optimal, i.e. $M$ is a solution to \eqref{eqn:g-optimization} with $g(\cdot) = \tr(\cdot)$, if and only if $M-V = H R H'$ for $R$ a $G \times G$ matrix with $R \succeq 0$ and $R_{gg}=1$ for all $g$. The CC estimator and collapsed-AS estimator both correspond to trace-optimal solutions. When $m=1$, we have a completely randomized experiment, and $\VhatAScollapsed= \VhatAS$, so $\VhatAS$ is also trace-optimal.
\end{prop}

\begin{prop} \label{prop:collapsed-as-frobenius-optimal}
	Consider a cluster-randomized experiment, as in \Cref{example:cluster-experiment}, with $G$ clusters each of size $m = N/G$. Suppose that $G_0=G_1 \geq 2$. Then the collapsed AS estimator is the unique Frobenius-norm minimizer, i.e. $M = V + HH'$ is the unique solution to \eqref{eqn:g-optimization} with $g(\cdot) = || \cdot + V||_F$. When $m=1$, we have a completely randomized experiment, and $\VhatAScollapsed= \VhatAS$, so $\VhatAS$ is the Frobenius-norm minimizer.
\end{prop}

\section{Proofs}
\paragraph{Proof to \Cref{prop:cc-valid}}
We prove a slightly more general version of \Cref{prop:cc-valid}, which allows for a vector-valued outcome. 
	\begin{prop} \label{prop:cc-valid-vector}
	Suppose $C > 1$. Let $\mathbf{Y}_k = (Y_{k1},\ldots, Y_{kQ})'$ be a $Q$-dimensional vector and let $\boldsymbol{\hat\tau} = \frac{1}{K} \sum_k \frac{W_k}{\pi_k} \mathbf{Y}_k$ be the corresponding vector of Horvitz-Thompson estimates. Let $\VhatCCbold= \frac{1}{K^2} \sum_{k, k': \pi_{kk'}>0} \frac{W_k W_{k'}}{\pi_{kk'}} b_{kk'} \mathbf{Y}_k \mathbf{Y}_{k'}'$ be the extension of $\VhatCC$ to vector-valued outcomes. Then
	
	\begin{equation}
		\E[\VhatCCbold] = \Var(\boldsymbol{\hat\tau}) + \frac{C}{K^2} \Var_c(\boldsymbol{\tau}_c) \succeq \Var(\boldsymbol{\hat\tau}) \label{eqn:conservativeness-vhatcc-vector}	
	\end{equation}
	where $$\Var_c(\boldsymbol{\tau}_c) = \frac{1}{C-1} \sum_c (\boldsymbol{\tau}_c - \bar{\boldsymbol{\tau}}) (\boldsymbol{\tau}_c - \bar{\boldsymbol{\tau}})' $$
	is the finite-population variance of the component sums $\boldsymbol{\tau}_c:=\sum_{k\in\mathcal{K}_c}\boldsymbol{Y}_k$ across components (with $\bar{\boldsymbol{\tau}}:= \frac{1}{C} \sum_c \boldsymbol{\tau}_c$).  	
\end{prop}

\begin{proof}

Note that if $\pi_{kk'}>0$, by standard Horvitz-Thompson logic, we have $\E[ \frac{W_k W_{k'}}{ \pi_{kk'} }]=1$, and hence
\[ \E[\VhatCCbold] = \frac{1}{K^2} \sum_{ \substack{ k,k' \\ \pi_{kk'} >0 }} b_{kk'} \mathbf{Y}_k \mathbf{Y}_{k'}' .\]

Observe that if $\pi_{kk'}=0$, then $v_{kk'} = \frac{ \pi_{kk'} - \pi_{k}\pi_{k'} }{ \pi_{k}\pi_{k'} } = -1$. By construction, if $\pi_{kk'} =0$, then $k,k' \in \mathcal{K}_c$ for some $c$, and hence $s_{kk'}= 1$. It follows that if $\pi_{kk'}=0$, then $b_{kk'}=v_{kk'} + s_{kk'} = 0$. Combining this fact with the previous display, we see that

\[ \E[\VhatCCbold] = \frac{1}{K^2} \sum_{  k,k' } b_{kk'} \mathbf{Y}_k \mathbf{Y}_{k'}' . \] 	
Since $b_{kk'} = v_{kk'} + s_{kk'}$ by definition, we see that 
$$\E[ \VhatCCbold ] = \underbracket{ \frac{1}{K^2} \sum_{k,k'} v_{kk'} \mathbf{Y}_k \mathbf{Y}_{k'}' }_{= \Var(\boldsymbol{\hat\tau})}+ \frac{1}{K^2} \sum_{k,k'} s_{kk'} \mathbf{Y}_k \mathbf{Y}_{k'}' .$$ 
It thus remains to show that $\frac{1}{K^2} \sum_{k,k'} s_{kk'} \mathbf{Y}_k \mathbf{Y}_{k'}' = \frac{C}{K^2} \Var_c(\boldsymbol{\tau}_c) .$ 
Recall that $s_{kk'}=1$ if $k,k' \in \mathcal{K}_c$ for some $c$ and $s_{kk'} = -1/(C-1)$ otherwise. We thus have that
\begin{align*}
\sum_{k,k'} s_{kk'} \mathbf{Y}_k \mathbf{Y}_{k'}' &= \sum_c \sum_{ \substack{k,k' \\ k,k' \in \mathcal{K}_c} } \mathbf{Y}_k \mathbf{Y}_{k'}' - \frac{1}{C-1} \sum_{c,c' : c \neq c'} \sum_{ \substack{k,k' \\ k \in \mathcal{K}_c, k' \in \mathcal{K}_{c'} }} \mathbf{Y}_k \mathbf{Y}_{k'}'  \\
&= \frac{C}{C-1}  \sum_c  \overbracket{ \sum_{ \substack{k,k' \\ k,k' \in \mathcal{K}_c} } \mathbf{Y}_k \mathbf{Y}_{k'}' }^{ \textcolor{blue}{=\boldsymbol{\tau}_c\boldsymbol{\tau}_c'}} - \frac{1}{C-1}  \overbracket{\sum_{ \substack{k,k'}} \mathbf{Y}_k \mathbf{Y}_{k'}' }^{\textcolor{blue}{= (\sum_c \boldsymbol{\tau}_c)(\sum_c \boldsymbol{\tau}_c)' = C^2 \bar{\boldsymbol{\tau}} \bar{\boldsymbol{\tau}}' }} \\
&= \frac{C}{C-1} \sum_c \boldsymbol{\tau}_c \boldsymbol{\tau}_c' - \frac{C^2}{C-1} \bar{\boldsymbol{\tau}} \bar{\boldsymbol{\tau}}' \\
& = C \Var_c(\boldsymbol{\tau}_c),
\end{align*} 
which completes the proof.
\end{proof}

\paragraph{Proof of \Cref{prop:super-pop-interp}}
\begin{proof}
	Let $\bar{\tau} = \frac{1}{C} \sum_c \tau_c = \frac{K}{C} \tau$. Observe that
	\begin{align*}
		\E_Y\big[(C-1) \Var_c(\tau_c)\big] &= \E_Y\bigg[ \sum_c \tau_c^2 - C \bar\tau^2 \bigg] \\
		& = \sum_c \bigg( \Var_Y(\tau_c) + \mu_c^2 \bigg) - C \bigg( \Var_Y(\bar\tau) + \E_Y[\bar\tau]^2 \bigg)  \\
		&=  \sum_c \bigg( \Var_Y(\tau_c) + \mu_c^2 \bigg) - C \bigg( \frac{1}{C^2} \sum_c \Var_Y(\tau_c)   + \bigg( \frac{1}{C}\sum_c \mu_c\bigg)^2 \bigg) \\
		&= \frac{C-1}{C} \sum_c \Var_Y(\tau_c) + \bigg( \underbracket{\sum_c \mu_c^2 - C  \bigg( \frac{1}{C}\sum_c \mu_c\bigg)^2 }_{=(C-1) \Var_c(\mu_c)} \bigg) 
	\end{align*}
	where the third line uses the fact that $\Cov_Y(\tau_c,\tau_{c'})=0$ for $c \neq c'$ to obtain that $\Var_Y(\bar\tau)= C^{-2} \sum_c \Var_Y(\tau_c)$. We thus see that
	\[\E_Y[\Var_c(\tau_c)] = \frac{1}{C} \sum_c \Var_Y( \tau_c )+  \Var_c(\mu_c) = C \Var_Y(\bar\tau) +  \Var_c(\mu_c).\]
	Combining the previous display with  \eqref{eqn:var-tauhat-sampling}, \eqref{eqn:e-vcc-sampling} and \eqref{eqn:e-hat-tau-given-y}, we then see that
	$$\E_{Y,W}[\VhatCC] = \Var_{Y,W}(\hat\tau) + \frac{C}{K^2} \Var_c(\mu_c),$$
	as desired. 
	
\end{proof}

\paragraph{Proof to \Cref{prop:vhatcc-admissible}}
\begin{proof}
	Let $S$ denote the matrix with $kk'$th entry $s_{kk'}$, so that $B=V+S$. For any given $M$, let $T(M) = M-V$, so that $S = T(B)$.  Note that since $V_{kk'}=-1$ for any $k,k'$ such that $\pi_{kk'}=0$, the constraint that $M_{kk'}=0$ if $\pi_{kk'}=0$ is equivalent to the constraint that $(T(M))_{kk'}=1$ if $\pi_{kk'}=0$. It therefore suffices to show that, for some positive definite $\Sigma$, $S$ solves
	\begin{equation}
	\min_T \ \langle \Sigma,T\rangle
	\quad\text{subject to}\quad
	T\succeq0,\qquad T_{kk'}=1\text{ for all }k,k'\text{ s.t. }\pi_{kk'}=0. \label{eqn:sigma-optimization}
\end{equation}
		Let $P_{S^\perp}$ be the orthogonal projection onto the null space of $S$. For $\alpha>0$, define
		\[
		\Sigma_{\alpha}:=P_{S^\perp}+\alpha A.
		\]
		
		We first show that $\Sigma_{\alpha}$ is positive definite for sufficiently small $\alpha>0$. We claim that the matrix $A$ is positive definite on $\operatorname{range}(S)$. To see this, note that any $x\in\operatorname{range}(S)$ is constant within connected components, since the rows of $S$ are identical within each component. Thus for any $x \in \operatorname{range}(S) \setminus \{0\}$, there exist $a_1,\ldots,a_C$ not all equal to zero such that $x_k=a_{c(k)}$ for $c(k)$ the value such that $k\in\mathcal{K}_{c(k)}$. Since each connected component has at least two items, each component contains at least one pair $k,k'$ such that $\pi_{kk'}=0$, or equivalently, $A_{kk'}=1$. Therefore,
	\[
	x'Ax
	=
	\sum_{\substack{k,k'\\  \pi_{kk'}=0}}x_kx_{k'}
	=
	\sum_{c=1}^C m_c a_c^2
	>
	0,
	\]
		where $m_c>0$ is the number of pairs $(k,k')$ in $\mathcal{K}_c$ with $A_{kk'}=1$.
		
		To see that $\Sigma_\alpha$ is positive definite for small $\alpha$, write
		$\mathbb{R}^K=\mathcal{R}\oplus\mathcal{N}$ with
		$\mathcal{R}=\operatorname{range}(S)$ and $\mathcal{N}=\operatorname{null}(S)$. Relative
		to this decomposition, write
		\[
		A=
		\begin{pmatrix}
		A_{\mathcal{R}\mathcal{R}} & A_{\mathcal{R}\mathcal{N}}\\
		A_{\mathcal{N}\mathcal{R}} & A_{\mathcal{N}\mathcal{N}}
		\end{pmatrix}.
		\]
		The preceding display shows that $A_{\mathcal{R}\mathcal{R}}\succ0$; let
		$\mu=\lambda_{\min}(A_{\mathcal{R}\mathcal{R}})>0$. Since $P_{S^\perp}$ is the
		orthogonal projection onto $\mathcal{N}$,
		\[
		\Sigma_\alpha
		=
		\begin{pmatrix}
		\alpha A_{\mathcal{R}\mathcal{R}} & \alpha A_{\mathcal{R}\mathcal{N}}\\
		\alpha A_{\mathcal{N}\mathcal{R}} & I_{\mathcal{N}}+\alpha A_{\mathcal{N}\mathcal{N}}
		\end{pmatrix}.
		\]
		For sufficiently small $\alpha>0$, the lower-right block is positive
		definite and has an inverse with operator norm bounded above by $2$. The Schur complement
		of this lower-right block is
		\[
		\alpha A_{\mathcal{R}\mathcal{R}}
		-
		\alpha^2 A_{\mathcal{R}\mathcal{N}}
		(I_{\mathcal{N}}+\alpha A_{\mathcal{N}\mathcal{N}})^{-1}
		A_{\mathcal{N}\mathcal{R}}.
		\]
		Its smallest eigenvalue is at least
		\[
		\alpha\mu
		-
		2\alpha^2\|A_{\mathcal{R}\mathcal{N}}\|^2,
		\]
		which is strictly positive for all sufficiently small $\alpha>0$. Hence
		$\Sigma_\alpha\succ0$. Set $\Sigma=\Sigma_{\alpha}$ for such an $\alpha$.

	Finally, we will show that $S$ is optimal for the problem \eqref{eqn:sigma-optimization}. It is clear that $S$ is feasible, since by definition $S_{kk'}=1$ for all pairs $k,k'$ with $\pi_{kk'}=0$, and we argued above that
	\[
	x'Sx
	=
	C\,\Var_c\left(\sum_{k\in\mathcal{K}_c}x_k\right)
	\geq0
	\]
	for all $x$. Now take any feasible $T$. Since $A$ is the adjacency matrix of the joint unobservability graph,
	\[
	\langle A,T\rangle
	=
	\sum_{\substack{k,k'\\  \pi_{kk'}=0}}T_{kk'}
	=
	\sum_{\substack{k,k'\\  \pi_{kk'}=0}}1
	=
	\langle A,S\rangle,
	\]
	where the second equality uses the fact that feasibility requires $T_{kk'}=1$ if $\pi_{kk'}=0$. Moreover, $\langle P_{S^\perp},T\rangle\geq0$ because $P_{S^\perp}\succeq0$ and $T\succeq0$, while $\langle P_{S^\perp},S\rangle=0$ because $P_{S^\perp}$ projects onto the null space of $S$, and thus $P_{S^\perp} S =0$. Therefore,
	\[
		\langle \Sigma,T\rangle
		=
		\langle P_{S^\perp},T\rangle+\alpha\langle A,T\rangle
		\geq
		\langle P_{S^\perp},S\rangle+\alpha\langle A,S\rangle
		=
		\langle \Sigma,S\rangle.
	\]
	Thus $S$ is optimal for the problem \eqref{eqn:sigma-optimization}, and therefore $B=V+S$ solves \eqref{eqn:g-optimization} with $g(M-V)=\langle\Sigma,M-V\rangle$.
		
	Finally, since $\Sigma$ is positive definite, $g$ is strictly monotone in the Loewner order: if $D\succeq0$ and $D\neq0$, then $\langle\Sigma,D\rangle>0$. Thus, as argued in \citet{harshaw_optimized_2026}, if there were a feasible $Q\neq B$ with $Q\preceq B$, then $Q-V\preceq B-V$ and strict monotonicity would imply $g(Q-V)<g(B-V)$, contradicting optimality of $B$. Hence $\VhatCC$ is admissible.
\end{proof}

We now provide proofs to \Cref{prop:trace-optimal-cluster-exp,prop:collapsed-as-frobenius-optimal}, since these build on arguments in the proof to \Cref{prop:vhatcc-admissible}. 
\paragraph{Proof to \Cref{prop:trace-optimal-cluster-exp}}
\begin{proof}

Let $T(M)=M-V$. As argued in the proof to \Cref{prop:vhatcc-admissible}, $M$ solves \eqref{eqn:g-optimization} with $g(\cdot) = \tr(\cdot)$ if and only if $T(M)$ solves
	\begin{equation}
	\min_T \ \tr(T)
	\quad\text{subject to}\quad
	T\succeq0,\qquad T_{kk'}=1\text{ for all }k,k'\text{ s.t. }\pi_{kk'}=0. \label{eqn:tr-optimization-T}
\end{equation}
Recall that in a cluster-randomized experiment, each $k$ indexes an $(i,d)$ combination, and since $G_0,G_1 \geq 2$, $\pi_{kk'}=0$ if and only if $\{k,k'\}= \{(i,0),(i',1)\}$ for $i,i'$ in the same cluster ($g(i)=g(i')$). Let $T$ be feasible in \eqref{eqn:tr-optimization-T}. Fix $k=(i,0),k'=(i',1)$ for $i,i'$ in the same cluster, so that $\pi_{kk'}=0$. Consider the principal submatrix $T_{sub}^{kk'}$ corresponding to the indices $k,k'$. By \Cref{lem:trace-principal-submatrix}, $T_{kk} +T_{k'k'} \geq 2$, with equality if and only if $T_{kk}=T_{k'k'}=1$. We thus see that $\tr(T) = \sum_i (T_{(i,0),(i,0)} + T_{(i,1),(i,1)}) \geq K = 2N$, and equality is obtained if and only if $T_{kk}=1$ for all $k$. 

We now show that if $T = HRH'$ for $R \succeq 0$ with $R_{gg}=1$, then $T$ is feasible and $\tr(T)=K$, and hence $T$ is optimal. Since $T = HRH'$ for $R \succeq 0$ it is immediate that $T \succeq 0$. If $\pi_{kk'}=0$, then $\{k,k'\}= \{(i,0),(i',1)\}$ with $g(i)=g(i')=g$ for some $g$, and hence $T_{kk'}= R_{gg}=1$. Thus, $T$ is feasible. Finally, if $k=(i,d)$ with $g(i)=g$, then $T_{kk}=R_{gg}=1$, so $\tr(T)=K$ and hence $T$ is optimal for \eqref{eqn:tr-optimization-T}.

To see that the CC estimator is optimal observe that its slack is $T(B) = S = H R_{CC} H'$ for $R_{CC} = \frac{C}{C-1}  I_C - \frac{1}{C-1} \mathbf{1}_C \mathbf{1}_{C}' $. Recalling that $C=G$ for a cluster-randomized experiment, we see that $R_{CC}$ is $G \times G$, and it is straightforward to verify that $R_{CC} \succeq 0$ with $R_{CC,gg} = 1$ for all $g$. Likewise, the collapsed-AS estimator has slack matrix $H H'$, from which it is immediate that it takes the desired form with $R =I_C$.

We now show the converse direction. Suppose $T$ is optimal. Then it is feasible and has $\tr(T)=K$, in which case $T^{kk'}_{sub}$ is a matrix of 1s for every pair $k =(i,0),k'=(i',1)$ where $g(i) = g(i')$. We claim that the columns of $T$ depend only on the relevant cluster, that is if $k=(i,d)$ and $k'=(i',d')$ for $g(i)=g(i')$, then $T (e_{k}- e_{k'}) = 0$, where $e_k$ is the $k$th basis vector. Consider first the case where $k=(i,d)$ and $k'=(i',1-d)$, in which case $T^{kk'}_{sub}$ is a matrix of 1s. We then see that
$$(e_{(i,d)} - e_{(i',1-d)} )'T(e_{(i,d)} - e_{(i',1-d)} ) = T_{(i,d),(i,d)} + T_{(i',1-d),(i',1-d)} - 2 T_{(i,d), (i',1-d)} = 1+1-2 = 0. $$
Since $T \succeq 0$, it follows that $T(e_{(i,d)} - e_{(i',1-d)})=0$. By analogous argument, $T(e_{(i',1-d)} - e_{(i',d)}) = 0$. Adding these two equations, we see that $T (e_{(i,d)} - e_{(i',d)})=0$, as well. Combined with the result in the previous equation display, we see that $T (e_{(i,d)} - e_{(i',d')}) = 0$ for any $i,i'$ with $g(i)= g(i')$, as we wished to show. Since $T$ is symmetric, it follows that the rows of $T$ also only depend on cluster identity. Hence, we can write $T = Q R Q' $ where $R$ is $G \times G$ and $Q$ is $K \times G$ with $Q_{kg} = 1$ if and only if $k=(i,d)$ with $g(i) =g$. Recall that in a cluster-randomized experiment with $G_0,G_1 \geq 2$, the connected components $\mathcal{K}_c$ collect all the potential outcomes for a particular $g$, and hence $Q=H$. Note that $R_{gg}=1$ for all $g$, since $T_{kk'}=1$ for any $k,k'$ corresponding to $i,i'$ in the same cluster. Note that $\mathrm{rank}(H)=G$, since each row of $H$ is a basis vector and every connected component is non-empty, so there is at least one row of $H$ matching each basis vector. Since $T = H R H' \succeq 0$ and $H$ has rank $G$, it follows that $R \succeq 0$. We have thus shown that if $T$ is optimal, then $T = H R H'$ for $R \succeq 0$ with $R_{gg}=1$.

\end{proof}

\paragraph{Proof to \Cref{prop:collapsed-as-frobenius-optimal}}
\begin{proof}
	As in the proof to \Cref{prop:vhatcc-admissible}, $M$ is optimal in  \eqref{eqn:g-optimization} with $g(\cdot) = || \cdot + V||_F$ if and only if $T(M) = M-V$ is optimal in  
	\[ \min_T ||T +V ||_F \text{ s.t. } T \succeq 0, \qquad T_{kk'}=1 \text{ whenever } \pi_{kk'}=0 .\]
	It thus suffices to show that $T = HH'$ is the unique optimizer to the optimization in the previous display.
	
	We proceed in two steps. First, we show that the optimal $T$ must be in a restricted class of matrices satisfying certain symmetry conditions. Second, we show that $T=HH'$ is optimal among this class. 
	
	\textit{Step 1: Reduction to a symmetric class of solutions.} Optimizing $||T+V||_F$ is equivalent to optimizing the squared Frobenius norm, $||T+V||_F^2$, which is strictly convex. Since the constraint set is convex and the objective is strictly convex, it follows that there is a unique solution. Let $k= (i,d)$ and $k'=(i',d')$. Note that the matrix $V$ has entry $V_{kk'} = \Cov\big( \tfrac{ 1\{D_i = d \} }{ P(D_i = d) }, \tfrac{ 1\{D_{i'} = d' \} }{ P(D_{i'} = d') }   \big) $, which depends on $(i,d),(i',d')$ only through: (a) whether $i,i'$ are in the same cluster, (b) whether $d=d'$, (c) whether $(i,d)=(i',d')$. It follows that $V$ is invariant to permutations of the rows and columns that swap the indices of pairs $(i,d),(i',d)$ for $i,i'$ in the same cluster, swap the cluster labels, or swap the treatment labels. Let $P$ be any permutation matrix that swaps indices in these ways and $\mathcal{G}$ the set of such permutations. It follows that $V = P V P'$ for any $P \in \mathcal{G}$, and hence for any $T$, $||PTP'+V||_F^2 = ||P (T+V) P'||_F^2 = ||T+V||_F^2$, since the Frobenius norm is invariant to permutations of row and column indices. Hence, the objective is invariant to permutations of $T$ by $P \in \mathcal{G}$. Likewise, the constraints that $T \succeq 0$ and $T_{kk'}=1$ when $\pi_{kk'}=0$ are invariant to such permutations (the latter because $\pi_{kk'}=0$ if and only if $V_{kk'} = -1$, and we've shown that $V$ is permutation-invariant). It follows that the unique optimizer $T$ must have $PT P'= T$ for all $P \in \mathcal{G}$, since if not there would exist $\tilde{T} = P T P' \neq T$ for some $P \in \mathcal{G}$, and $\tilde{T}$ would also be feasible and optimal, contradicting the existence of a unique optimal solution. We thus see that $T$ is permutation-invariant. Order the potential outcomes so that all the values of $Y(1)$ for a given cluster come first, followed by all the values of $Y(0)$ for that cluster. Let $T_{gh}$ correspond to the submatrix of $T$ with rows corresponding to potential outcomes $i$ with $g(i)=g$ and columns corresponding to potential outcomes with $g(i) = h$. Let $E= \mathbf{1}_m \mathbf{1}_m'$ be the $m \times m$ matrix of 1s. The symmetries just derived, along with the feasibility constraint, imply that 
	\[ T_{gg} = \twobytwomat{ \tilde{a} E + \delta I  }{ E }{ E }{ \tilde{a} E + \delta I }, \qquad T_{gh} = \twobytwomat{x E}{y E}{y E}{x E} \quad (g \neq h) \]
	for constants $\tilde{a}, \delta,x,y$. Note that if $m=1$ (which corresponds to a completely-randomized experiment), then we can normalize $\delta =0$, since $E=I=1$. Note, further, that if we define $a = \tilde{a} + \tfrac{1}{m} \delta$ and $Q = I - \frac{1}{m} E$, then we can equivalently write
		\[ T_{gg} = \twobytwomat{ a E + \delta Q  }{ E }{ E }{ a E + \delta Q } .\]
		The latter formulation will be more convenient since it implies that $E Q =0$. 
		
		\textit{Step 2: optimization over the restricted class.} Standard calculations for the pairwise treatment probabilities in a cluster-randomized experiment with treatment probability $1/2$ yield, for $r = 1/(G-1)$, 
		\[ V_{gg} = \twobytwomat{ E  }{ - E }{  -E }{ E} , \qquad V_{gh} = \twobytwomat{ -rE  }{ r E }{ rE }{-r E}   \quad (g \neq h).\]
		Hence, for $T$ in the restricted class derived in Step 1, 
		\[ V_{gg} + T_{gg} = \twobytwomat{ (a+1) E + \delta Q  }{ 0 }{ 0 }{ (a+1) E + \delta Q }, \qquad V_{gh} + T_{gh} = \twobytwomat{(x-r) E}{(y+r) E}{(y+r) E}{(x-r) E}  \quad (g \neq h).\]
		Recalling that $E$ and $Q$ are orthogonal in the Frobenius inner product, we see that  
		\begin{align*}
			||V+T||_F^2 &= 2 G (a+1)^2 ||E||_F^2 + 2G \delta^2 ||Q||_F^2 + 2 G(G-1) ((x-r)^2 + (y+r)^2) ||E||_F^2 \\
									&= 2 m^2 \big(G (a+1)^2 + G(G-1) ((x-r)^2 + (y+r)^2) \big) + 2G \delta^2 ||Q||_F^2.
		\end{align*}
	Next, we claim that the optimal $T$ must have $\delta = 0$. It is clear that, holding other parameters constant, the objective is minimized at $\delta =0$. Now, consider the constraint that $T \succeq 0$, and note that $T \succeq 0$ if and only if $v' T v \geq 0$ for any vector $v$. Let $v_{gd}$ be the subvector of $v$ corresponding to cluster $g$ and potential outcome $Y(d)$. We can write any $v_{gd} = \alpha_{gd} \mathbf{1}_m + \beta_{gd} v^{\perp}_{gd}$ where $\mathbf{1}_m' v^{\perp}_{gd} =0$. Note that $Q \mathbf{1}_m = 0$ and $E v^{\perp}_{gd} = 0$. Let $v_1,v_2$ respectively be the vectors that stack the $\alpha_{gd} \mathbf{1}_m$ and $\beta_{gd} v^{\perp}_{gd}$ terms. For fixed $a,x,y$, let $T_{\delta=0}$ be the matrix that sets $\delta=0$ and let $D_{Q}$ be the block diagonal matrix with $Q$ on the diagonals, so that $T = T_{\delta=0} + \delta D_{Q}$. Then for any $v$ we have that
	\[ v'T v = v_1 ' T_{\delta=0} v_1 +  v_2' (\delta D_{Q}) v_2.   \]
	It follows that $T \succeq 0$ if and only if $T_{\delta =0 } \succeq 0$ and $\delta  D_{Q} \succeq 0$. But since $Q \succeq 0$, $\delta D_{Q} \succeq  0$ if and only if $\delta \geq 0$. We thus see that any $(a,x,y)$ that are feasible with $\delta >0$ are also feasible when $\delta = 0$. Since the objective is minimized in $\delta$ at $\delta =0 $, holding other parameters constant, it follows that $\delta = 0$ must be optimal. 
	
	Next, observe that for $v_1$ of the form described above and $T$ with $\delta =0$, we have 
	
		\[ v_1' T  v_1 =  \sum_g \alpha_{g}'  \twobytwomat{a E  }{ E }{  E }{a E} \alpha_g + \sum_g \sum_{g': g' \neq g}  \alpha_g \twobytwomat{x E}{y E}{y E}{x E} \alpha_{g'} \]
		where $\alpha_g' = (\alpha_{g1} \mathbf{1}_m, \alpha_{g0} \mathbf{1}_m).$ Since $\mathbf{1}_m' E \mathbf{1}_m = m^2$, it follows that  
		
		\[ v_1' T  v_1  / m^2 =  \sum_g \twovec{\alpha_{g1}}{\alpha_{g0}} ' \twobytwomat{a   }{1}{ 1}{a}  \twovec{\alpha_{g1}}{\alpha_{g0}} + \sum_g \sum_{g':g' \neq g}  \twovec{\alpha_{g1}}{\alpha_{g0}} ' \twobytwomat{x   }{y}{ y}{x}  \twovec{\alpha_{g'1}}{\alpha_{g'0}}\]
		
		Let $e_+ = \tfrac{1}{\sqrt{2}} \mathbf{1}_2$ and $e_{-}= \tfrac{1}{\sqrt{2}} (1,-1)'$. Note that we have the decompositions
		\[
		\twobytwomat{a}{1}{1}{a}=(a+1) e_{+} e_{+}'+ (a-1)e_{-}e_{-}',
		\qquad
		\twobytwomat{x}{y}{y}{x}=(x+y) e_{+} e_{+}'+ (x-y)e_{-}e_{-}'.
		\]
		
	We can likewise decompose $(\alpha_{g,1},\alpha_{g,0}) = \alpha_{g,+} e_+ + \alpha_{g,-} e_{-} $ for $\alpha_{g,+} = (\alpha_{g,1} + \alpha_{g,0} ) /\sqrt{2}$ and  $\alpha_{g,-} = (\alpha_{g,1} - \alpha_{g,0} ) /\sqrt{2}$. Letting $\alpha_+$ and $\alpha_{-}$ be the corresponding vectors, we see that
	\[ v_1' T  v_1  / m^2 = \alpha_+' \big[ ((a+1) - (x+y) ) I_G + (x+y) \mathbf{1}_G \mathbf{1}_G'  \big]  \alpha_+  +  \alpha_-' \big[ ((a-1) - (x-y) ) I_G + (x-y) \mathbf{1}_G \mathbf{1}_G'  \big]  \alpha_-. \] 
	It follows that $v_1' T v_1 \geq 0$ for all choices of $\alpha_+,\alpha_{-}$ if and only if the two matrices in the quadratic forms in the previous display are PSD. Note that a matrix of the form $c_1 I_G + c_2 \mathbf{1}_G \mathbf{1}_G'$ has eigenvalue of $c_1 + G c_2$ corresponding to eigenvector $\mathbf{1}_G$ and eigenvalue $c_1$ with multiplicity $G-1$, corresponding to the remaining vectors in any orthogonal basis relative to $\mathbf{1}_G$. Hence, $c_1 I_G + c_2 \mathbf{1}_G \mathbf{1}_G'$ is PSD if and only if $c_1 \geq 0$ and $c_1 \geq -G c_2$. Letting $u=(x+y)$ and $v=(x-y)$ for notational convenience, we see that $T$ is PSD iff 
	\begin{align*}
		& (a+1) - u \geq 0 & (a+1)  \geq -(G-1)  u \\
		& (a-1) - v \geq 0 & (a-1)   \geq  -(G-1) v. \numberthis \label{eqn:psd-feasibility}
	\end{align*}
Using the fact that $\delta=0$ and some algebra, we can express the objective as
$$||T+V||_F^2 = 2m^2G ((a+1)^2 + \frac{1}{2} (G-1) (u^2 + (v-2r)^2 )   ),$$
which is additively separable in functions of $a$, $u,$ and $v$. Note that the two constraints on the second line in display \eqref{eqn:psd-feasibility} are satisfied if and only if
\begin{equation}  
 - \frac{a-1}{G-1} \leq  v \leq (a-1) , \label{eqn:constraint-on-v} 
\end{equation}
which can be satisfied only if $a \geq 1$. Observe that the constraints involving $u$ involve only $a$ and not $v$, and for any $a \geq 1$, $u =0$ is feasible. Since the objective is minimized wrt $u$ at $u=0$, it follows that $u=0$ must be optimal. Next, consider the optimization wrt $v$ for a fixed value of $a$. Since only the bottom two constraints in \eqref{eqn:psd-feasibility} involve $v$, for a fixed value of $a$, $v$ is feasible if and only if it satisfies \eqref{eqn:constraint-on-v}. Since $r> 0$, it is clear that the constrained optimum wrt $v$ is $v = \min\{a-1,2r\}$. We are then left with the one-dimensional profiled objective with respect to $a$,
$$ 2m^2 G \big( (a+1)^2 + \frac{1}{2} (G-1) \max\{0, 2r - a +1 \}^2  \big)  $$
subject to the constraint $a \ge 1$. For $a \in [1,1+2r] $, the objective is proportional to 
$$(a+1)^2 + \frac{1}{2} (G-1) (2r - a + 1)^2,$$
which has derivative
$$ 2(a+1) - (G-1) (2r - a +1)  = (G+1)(a-1) + 2 > 0 $$
where the second equality uses the fact that $r = 1/(G-1)$.  Likewise for $a \geq 1+2r$, the objective is proportional to $(a+1)^2$. Hence the objective is increasing in $a$ for $a \geq 1$, and thus the minimum is achieved at $a=1$, which implies that $v = \min\{a-1,2r\}=0$. Note that if $u = (x+y)=0$ and $v= (x-y)=0$, then $x=y=0$. It follows that the optimum is obtained at $T = HH'$, which is precisely the collapsed-AS slack matrix. 
\end{proof}

\paragraph{Proof to \Cref{prop:as-admissibility}}
\begin{proof}
Suppose $|\mathcal{K}_c| \leq 2$ for all $c$. If $|\mathcal{K}_c|=1$ for all $c$, then the AS slack matrix is the zero matrix, and so it is trivially admissible. For the remainder of the proof, we assume $|\mathcal{K}_c| =2$ for at least some $c$. We will show that the AS slack matrix is trace optimal, i.e. is a solution to \eqref{eqn:tr-optimization-T}, and therefore the AS estimator is admissible. By \Cref{lem:trace-principal-submatrix}, for any feasible $T$, the principal submatrix $T_{sub}^{kk'}$ corresponding to any pair $k,k'$ with $\pi_{kk'}=0$ must have trace of at least 2, so that $T_{kk}+T_{k'k'} \geq 2$. It follows that \[\tr(T) = \sum_k T_{kk} \geq \sum_{c} \sum_{k \in \mathcal{K}_c} 1\{|\mathcal{K}_c| = 2\} T_{kk} \geq 2  \sum_{c} 1\{|\mathcal{K}_c| = 2\}\]
where the first inequality uses the fact that $T \succeq 0$ and hence all its diagonal elements are weakly positive. However, the AS slack is $T_{AS} = D + A$, where $D$ is the degree matrix with diagonal element $D_{kk} = \sum_{k'} 1\{\pi_{kk'} = 0\}$. Note that $\tr(A) =0$ by definition, so that $\tr(T_{AS})= \tr(D)$. Since $|\mathcal{K}_c | \leq 2$, $D_{kk}=1$ if $k \in \mathcal{K}_c$ with $|\mathcal{K}_c|=2$ and $D_{kk}=0$ otherwise. Hence, $$\tr(D) = \sum_k 1\{k \in \mathcal{K}_c \text{ s.t. } |\mathcal{K}_c| = 2\} =  2 \sum_c 1\{ |\mathcal{K}_c| = 2\},$$
from which we see that $T_{AS}$ is trace-optimal, and hence AS is admissible.

Conversely, suppose that there exists at least one $c$ such that $|\mathcal{K}_c |\geq 3$. First, suppose that the subgraph for the component $c$ is non-bipartite. Let $T_{AS,c}$ denote the principal submatrix of $T_{AS}$ corresponding to component $c$. Since $T_{AS,c} = D_c + A_c$ is the signless Laplacian of the subgraph of component $c$, which is non-bipartite, it is positive definite. Hence for $\epsilon$ sufficiently small $T_{\epsilon,c} = T_{AS,c} - \epsilon I$ is positive definite. Observe that if the elements are ordered by connected component, then $T_{AS}$ is block diagonal, and recall that a block-diagonal matrix is PSD if and only if all its blocks are PSD. Letting $T_{\epsilon}$ be the matrix that replaces $T_{AS,c}$ with $T_{\epsilon,c}$ and preserves $T_{AS}$ otherwise, we see that $T_{\epsilon} \succeq 0$.  Moreover, since $T_{\epsilon}$ differs from $T_{AS}$ only on the diagonal, and $\pi_{kk}>0$ for all $k$, it follows that we have $T_{\epsilon,kk'}= T_{AS,kk'} = 1$ for all $k,k'$ with $\pi_{kk'}=0$, and hence $T_\epsilon$ is feasible and dominates $T_{AS}$. Next, suppose that the subgraph for component $c$ is bipartite. Partition $\mathcal{K}_c$ into sets $L$ and $R$ so that all pairs $k,k'\in\mathcal{K}_c$ with $\pi_{kk'}=0$ have one entry in $L$ and one entry in $R$. Because $|\mathcal{K}_c|>2$, at least one of $L$ and $R$ has two distinct elements. Suppose, without loss of generality, that $L$ has at least two elements, so there exist $k_1,k_2 \in L$, $k_1 \neq k_2$. Then by construction $\pi_{k_1k_2}>0$.  Let $u = e_{k_1}-e_{k_2}$, for $e_k$ the $k$th basis vector. Let $U = uu'$ and note that $U_{k k'} = u_{k} u_{k'} \neq 0$ if and only if $\{k,k'\}\subseteq \{k_1,k_2\}$, and hence $U_{kk'}=0$ for all $k,k'$ such that $\pi_{kk'}=0$. Hence, for any $\epsilon>0$,  $T_\epsilon = T_{AS} - \epsilon U$ satisfies the constraint that $T_{\epsilon,kk'}=1$ if $\pi_{kk'} =0$.  Next, observe that for any vector $x$, $x' T_{AS} x = \frac{1}{2} \sum_{k,k': \pi_{kk'} =0} (x_k +x_{k'})^2$. It follows that $x \in \mathrm{null}(T_{AS})$ only if $x_{k} + x_{k'}=0$ for all $k,k'$ s.t. $\pi_{kk'}=0$. Because the subgraph for component $c$ is connected and bipartite, any $x \in \mathrm{null}(T_{AS})$ must have $x_\ell=\alpha, x_{r} = -\alpha$ for all $\ell \in L, r \in R$ and some constant $\alpha$: we can construct a sequence $\tilde{k}_1,\ldots,\tilde{k}_{m}$ including all the elements of $\mathcal{K}_c$ (possibly multiple times) such that $\pi_{\tilde{k}_t,\tilde{k}_{t+1}}=0$ for all $t$ and the elements alternate between $L$ and $R$; the constraints $x_{\tilde{k}_t} + x_{\tilde{k}_{t+1}}=0$ can then hold simultaneously only if $x_{\tilde{k}_t} = -x_{\tilde{k}_{t+1}}$ for all $t$. However, since $k_1,k_2$ are both in $L$, we see that $u'x =0$ for any $x \in \mathrm{null}(T_{AS})$ and hence $u \in \mathrm{range}(T_{AS})$. Since $T_{AS} \succeq 0$ and $u$ is in the range of $T_{AS}$, it follows from \Cref{lem:psd-transformation-in-range} that there exists $\epsilon$ sufficiently small such that $T_\epsilon \succeq 0$. Thus, $T_\epsilon$ is feasible and $T_{AS} - T_\epsilon = \epsilon uu' \succeq 0$, so $T_\epsilon$ dominates $T_{AS}$.

\end{proof}

\paragraph{Proof of \Cref{prop:sharp-null-exact}}
\begin{proof}
	\emph{(Proof of part 1.)} We begin by showing that if $T$ has the form in \eqref{eqn:h-lambda-h}, then $\VhatM$ is upwards-biased and sharp-null-exact. Observe that if $\pi_{kk'}=0$, then $k,k' \in \mathcal{K}_c$ for some $c$, and therefore $T_{kk'} = (H \Lambda H')_{kk'} = \Lambda_{cc} = 1$. Further, since $\Lambda \succeq 0$, we have that $\E[\VhatM] - \Var(\hat\tau) = \frac{1}{K^2}(H' Y)' \Lambda (H'Y) \geq 0$, and hence $\VhatM$ is upwards-biased. Next, let $r_i$ be the $K \times 1$ vector with elements $\omega_{i,q}$ in position $(i,q)$ and zeros otherwise. Since all potential outcomes corresponding to the same unit $i$ are in the same connected component, and $\sum_{q} \omega_{i,q}=0$, it follows that $H' r_i = 0$ for all $i$. Under the sharp null, however, $Y = \sum_i Y_i(0) r_i$, which is in the span of the $r_i$, and hence $Y' T Y = 0$, so that $\VhatM$ is sharp-null-exact. Further, observe from the definition of $H$ that $H'Y = \boldsymbol{\tau}$, and hence $\E[\VhatM] - \Var(\hat\tau) = \frac{1}{K^2}(H' Y)' \Lambda (H'Y)= \frac{1}{K^2}\boldsymbol{\tau}' \Lambda \boldsymbol{\tau}$. 
	
	Conversely, suppose that $T=T(M)$ corresponds to the slack of an upwards-biased, sharp-null-exact estimator $\VhatM$. Then $T \succeq 0$ and $T_{kk'}=1$ for all $k,k'$ with $\pi_{kk'}=0$. Further, since $\VhatM$ is sharp-null-exact, we have that  $ (\sum_i Y_i(0) r_i)' T \left( \sum_i Y_i(0) r_i \right) =0$ for any choice of $Y_i(0)$, from which it follows that $r_i ' T r_i =0$ for all $i$. Since $T \succeq 0$, however, this implies that $T r_i = 0$ for all $i$. Since $\pi_{(i,q),(i,q')}=0$ for any $q \neq q'$, we have that $T_{(i,q),(i,q')}=1$ for all $q \neq q'$. Hence $(T r_i)_{(i,q)} = T_{(i,q),(i,q)} \omega_{i,q} + \sum_{q': q' \neq q} \omega_{i,q'} = 0$. Since $\sum_{\tilde{q}} \omega_{i,\tilde{q}} =0$ and $\omega_{i,q} \neq 0$, it follows that $T_{(i,q),(i,q)}=1$. We thus see that $T$ has all subblocks corresponding to pairs of potential outcomes for the same unit $i$ equal to 1. Now take the basis vectors $e_{(i,q)}, e_{(i,q')}$ for $q \neq q'$. Note that $(e_{(i,q)} - e_{(i,q')})' T (e_{(i,q)} - e_{(i,q')}) = T_{(i,q),(i,q)} +  T_{(i,q'),(i,q')} - 2 T_{(i,q),(i,q')}= 1+1-2=0$. Since $T \succeq 0$, we see that $ T (e_{(i,q)} - e_{(i,q')}) =0$. It follows that all columns involving potential outcomes for the same unit are identical. By symmetry, the same is true of the rows of $T$. Hence $T = G_N C_N G_N'$ where $G_N$ is the $K \times N$ matrix indicating unit membership ($G_{N,(i,q),j}= 1\{i=j\}$) and $C_N$ is a symmetric matrix. We claim that $C_N \succeq 0$. Indeed, note that $G_N$ has full column rank, and so for any $x \in \reals^N$, we can choose $y \in \reals^K$ such that $G_N'y = x$ and hence $0 \leq y' T y  = x' C_N x$, since $T \succeq 0$. Moreover, for $q \neq q'$, $T_{(i,q),(i,q')} = C_{N,ii} $, and hence $C_{N,ii} =1$ for all $i$.  Now consider any $(i,q),(j,q')$ such that $\pi_{(i,q),(j,q')}=0$. Since $T_{(i,q),(j,q')}=1$, we have that $C_{N,ij}=C_{N,ji}=1$. It follows that $(e_i - e_j)' C_N (e_i - e_j)= 1+1-2=0$, and since $C_N \succeq 0$, we see that $C_N (e_i - e_j)=0$, so that all columns of $C_N$ corresponding to $i$ and $j$ are equal. By symmetry of $C_N$, all rows corresponding to $i$ and $j$ are equal as well. Iterating this argument, we see that all rows and columns of $C_N$ must be equal for indices in the same connected components, since for any $k=(i,q)$ and $k'=(j,q')$ in the same $\mathcal{K}_c$, there exists a path of finite length $r$, $(i,q)=(i_1,q_1),\ldots,(i_r,q_r)=(j,q')$ such that $\pi_{(i_t,q_t), (i_{t+1},q_{t+1})}=0$, and hence the rows and columns corresponding to $(i_t,q_t), (i_{t+1},q_{t+1})$ are the same. It follows that there exists a symmetric $C \times C$ matrix $\Lambda$ such that $C_N = J \Lambda J'$, where $J$ is the $N \times C$ matrix that maps units to components (and this is well-defined since all potential outcomes for the same unit are in the same component). Since $C_N \succeq 0$ and $J$ has full column rank, we see that $\Lambda \succeq 0$, and further, for any $i$ associated with component $c$, we have $C_{N,ii} =\Lambda_{c,c} =1$.  Additionally, since all potential outcomes for each unit are in the same component, we see that $G J = H$. It follows that $T = G C_N G' = GJ \Lambda (GJ)' = H \Lambda H'$, as we wished to show. 
	
	\textit{(Proof of part 2.)} Suppose $\VhatM$ is exact under component homogeneity. Then $\VhatM$ must be exact under the sharp null, since under the sharp null $\tau_c=0$ for all $c$. Hence, $T(M)$ takes the form in \eqref{eqn:h-lambda-h}. However, if $\VhatM$ is exact under component homogeneity, then $\boldsymbol{\tau}'\Lambda \boldsymbol{\tau}=0$ whenever $\boldsymbol{\tau} = \alpha \mathbf{1}_C$ for any $\alpha$, which together with the fact that $\Lambda \succeq 0$ implies that $\Lambda \mathbf{1}_C=0$. (We note that $\boldsymbol{\tau} = \alpha \mathbf{1}_C$ is achievable for any $\alpha$ since $\omega_{i,q}\neq 0$ for all $i,q$, and thus one can choose potential outcomes to achieve any $\boldsymbol{\tau}$.) Recall that $||\Lambda||_{op}$ is equal to the largest eigenvalue of $\Lambda$. Since $\Lambda \mathbf{1}_C = 0$, $\Lambda$ has rank at most $C-1$, and hence at most $C-1$ strictly positive eigenvalues. Since $\Lambda_{cc}=1$ for all $c$, $\tr(\Lambda) =C$, and thus the eigenvalues of $\Lambda$ sum to $C$. It follows that the maximal eigenvalue must be at least $C/(C-1)$, with equality only if all non-zero eigenvalues equal $C/(C-1)$. However, $\Lambda_{CC} = \frac{C}{C-1}I_C - \frac{1}{C-1} \mathbf{1}_C \mathbf{1}_C'$ has eigenvalue of 0 in the direction $\mathbf{1}_C$ and has eigenvalues equal to $C/(C-1)$ in all orthogonal directions, and thus achieves the lower bound. To show uniqueness, suppose that $\Lambda$ is such that $\Lambda \mathbf{1}_C=0$ and the remaining eigenvalues are equal to $C/(C-1)$. Since $\Lambda$ is rank $C-1$, its remaining eigenvectors span the orthogonal complement to $\mathbf{1}_C$, and hence for any $x \in \reals^C$, $\Lambda x=  C/(C-1) \cdot P_{ \mathbf{1}_C^\perp } x$, where $P_{ \mathbf{1}_C^\perp }$ is the projection into the orthogonal complement to $\mathbf{1}_C$. Since this holds for all $x$, we see that $\Lambda =  C/(C-1) \cdot P_{ \mathbf{1}_C^\perp }  = \Lambda_{CC}$. 
\end{proof}

\begin{lem} \label{lem:trace-principal-submatrix}
	Let $T \succeq 0$ be such that $T_{kk'}=1$ for all $k,k'$ such that $\pi_{kk'}=0$.	Consider a pair $k,k'$ such that $\pi_{kk'}=0$.  Let $T_{sub}^{kk'}$ be the principal submatrix  corresponding to the indices $k,k'$. Then $\tr(T_{sub}^{kk'}) \geq 2$ with equality if and only if $T_{kk}=T_{k'k'}=1$. 
\end{lem}
\begin{proof}
	By assumption $T_{kk'}=T_{k'k} = 1$, so we have
	$$ T_{sub}^{kk'} = \twobytwomat{ T_{kk} }{1 }{1}{T_{k'k'}} .$$
	The principal submatrices of PSD matrices are PSD, so $T_{sub}^{kk'}  \succeq 0$, and hence $T_{kk} T_{k'k'} \geq 1$. By the AM-GM inequality, it follows that $T_{kk} +T_{k'k'} \geq 2$, with equality if and only if $T_{kk}=T_{k'k'}=1$.
\end{proof}

\begin{lem} \label{lem:psd-transformation-in-range}
Suppose $T \succeq 0$ and $u \in \mathrm{range}(T)$. Then there exists $\epsilon >0$ such that $T_{\epsilon}= T - \epsilon u u' \succeq 0$. 
\end{lem}
\begin{proof}
If $u=0$ the result holds trivially, so assume $||u||>0$. Let $v_1,\ldots,v_m$ be the eigenvectors of $T$, normalized to have Euclidean norm of 1, and $\lambda_1,\ldots,\lambda_m$ the corresponding eigenvalues. Take any non-zero $x \in \reals^m$ and write $x = c_{x,1} v_1 + \ldots c_{x,m} v_m$ and $u= c_{u,1} v_1 + \ldots c_{u,m} v_m$. Note that since $u \in \mathrm{range}(T)$, $c_{u,m}=0$ for all $m$ such that $\lambda_m = 0$. Then \begin{align*}
	x' T_\epsilon x &= \sum_k c_{x,k}^2 \lambda_k - \epsilon \left(\sum_k c_{x,k} c_{u,k} \right)^2 \\ 
	&= \sum_{k: \lambda_k > 0} c_{x,k}^2 \lambda_k - \epsilon \left(\sum_{k: \lambda_k > 0} c_{x,k} c_{u,k} \right)^2 \\&  \geq  \sum_{k: \lambda_k > 0} c_{x,k}^2 \lambda_k - \epsilon \left(\sum_{k: \lambda_k > 0} c_{x,k}^2 \right) \left( \sum_{k: \lambda_k >0} c_{u,k}^2 \right)  \text{ (Cauchy-Schwarz) }\\
	&=  \sum_{k: \lambda_k > 0} c_{x,k}^2 \cdot \left( \lambda_k - \epsilon ||c_u||^2 \right).
\end{align*} 

It follows that $T_\epsilon \succeq 0$ for $\epsilon <  \min\{  \lambda_k / ||c_u||^2  \,:\, \lambda_k > 0 \}$. 
\end{proof}

\paragraph{Proof of \Cref{lem:tauhat-consistency}}
\begin{proof}
	We have that $\Var(\hat\tau) = \frac{1}{K^2} Y' V Y$. By definition of the operator norm, $||V Y|| \leq ||V||_{op} ||Y||$, and hence by Cauchy-Schwarz, $\frac{1}{K^2} Y' V Y \leq \frac{1}{K^2} ||V||_{op} ||Y||^2 = s_K \frac{1}{K} \sum_k Y_k^2 = s_K O(1)= O(s_K)$. Since $\E[\hat\tau] = \tau$, Chebyshev's inequality implies $\hat\tau - \tau = O_p( \sqrt{\Var(\hat\tau)} ) = O_p( \sqrt{s_K} )$.
\end{proof}

\paragraph{Proof of \Cref{prop:vhat-consistency}}
We prove \Cref{prop:vhat-consistency-vector}, which generalizes \Cref{prop:vhat-consistency} to vector-valued outcomes (and relaxes the assumptions needed for the first conclusion).
\begin{prop} \label{prop:vhat-consistency-vector}
	Let $\mathbf{Y}_k \in \reals^Q$, $\boldsymbol{\hat\tau}, \VhatCCbold$ be as in \Cref{prop:cc-valid-vector}, with the dimension $Q$ fixed in $K$. Suppose that $C>1$, $s_K \to 0$. Consider the following conditions:	
	
	\begin{enumerate}[(i)]
		\item 
		(Limiting distribution of $\boldsymbol{\hat\tau}$) $s_K^{-1/2} (\boldsymbol{\hat\tau} - \boldsymbol{\tau}) \xrightarrow{d} \mathbf{L}$ for a non-degenerate distribution $\mathbf{L}$ with $\Var(\mathbf{L}) = \lim_{K \to \infty} s_K^{-1} \Var(\boldsymbol{\hat\tau}) =: \mathbf{V}^*\succ 0$.
		
		\item (Limit of CC slack)
		$s_K^{-1} (\E[\VhatCCbold]- \Var(\boldsymbol{\hat\tau})) = s_K^{-1}  \frac{C}{K^2} \Var_c(\boldsymbol{\tau}_c) \to \boldsymbol{\delta} \succeq 0$.
		
		\item (Bounded fourth moment of $Y$) $\frac{1}{K} \sum_k ||\mathbf{Y}_k||^4 = O(1)$.
		
		\item (Limited pair-dependence in treatment indicators) 
		\begin{equation*} \bigg(\frac{1}{K^4} \sum_{ \substack{(k,k'), (\ell,\ell') \\ \pi_{kk'}>0, \pi_{\ell \ell'}>0}}   \left| \Cov\bigg(\frac{W_k W_{k'}}{\pi_{kk'}} b_{kk'}, \frac{W_{\ell}W_{\ell'} }{\pi_{\ell \ell'}} b_{\ell\ell'} \bigg)\right|^{4/3}\bigg)^{3/4} = o(s_{K}^2) .
		\end{equation*}

	\end{enumerate}
	If conditions (iii) and (iv) hold, then $\VhatCCbold =\E[\VhatCCbold] + o_p(s_K)$. Under (i)-(iv), it is further the case that $s_K^{-1} \VhatCCbold \xrightarrow{p} \mathbf{V}^* + \boldsymbol{\delta} \succeq \mathbf{V}^*$.
\end{prop}

\begin{proof}

	Note that $\mathrm{vec}(\VhatCCbold)$ is a vector-valued Horvitz-Thompson estimator for $\frac{1}{K^2} \sum_{k,k':\pi_{kk'}>0} b_{kk'} \mathrm{vec}(\mathbf{Y}_k \mathbf{Y}_{k'}')$, and hence 
	\[\Var(\mathrm{vec}(\VhatCCbold)) = \frac{1}{K^4} \sum_{ \substack{(k,k'), (\ell,\ell') \\ \pi_{kk'}>0, \pi_{\ell \ell'}>0}}   \Cov\bigg(\frac{W_k W_{k'}}{\pi_{kk'}}, \frac{W_{\ell}W_{\ell'} }{\pi_{\ell \ell'}} \bigg) b_{kk'} b_{\ell \ell'} \mathrm{vec}(\mathbf{Y}_{k} \mathbf{Y}_{k'}') \mathrm{vec}(\mathbf{Y}_{\ell} \mathbf{Y}_{\ell'})'. \]
	It follows that 
	
	\begin{align*}  \E[ || \VhatCCbold - \E[\VhatCCbold] ||_F^2 ] &= \E[ || \mathrm{vec}(\VhatCCbold) - \E[ \mathrm{vec}(\VhatCCbold)] ||_2^{2} ]\\  &=   \tr( \Var(\mathrm{vec}(\VhatCCbold)) )  \\ &= \frac{1}{K^4} \sum_{ \substack{(k,k'), (\ell,\ell') \\ \pi_{kk'}>0, \pi_{\ell \ell'}>0}}   \Cov\bigg(\frac{W_k W_{k'}}{\pi_{kk'}}, \frac{W_{\ell}W_{\ell'} }{\pi_{\ell \ell'}} \bigg) b_{kk'} b_{\ell \ell'} \langle \mathbf{Y}_{k} \mathbf{Y}_{k'}', \mathbf{Y}_{\ell} \mathbf{Y}_{\ell'}' \rangle_F,
	\end{align*} 
	where we use the fact that for vectors $v_1,v_2$, $\tr(v_1 v_2') = v_1'v_2$ and hence $\tr(\mathrm{vec}(\mathbf{Y}_{k} \mathbf{Y}_{k'}') \mathrm{vec}(\mathbf{Y}_{\ell} \mathbf{Y}_{\ell'})') = \tr(\mathrm{vec}(\mathbf{Y}_{k} \mathbf{Y}_{k'}')' \mathrm{vec}(\mathbf{Y}_{\ell} \mathbf{Y}_{\ell'})) = \langle \mathbf{Y}_{k} \mathbf{Y}_{k'}', \mathbf{Y}_{\ell} \mathbf{Y}_{\ell'}' \rangle_F$. Applying Holder's inequality to the final expression in the previous display along with (iv), we obtain that 
	{\small
		\begin{align*}  \E[ || \VhatCCbold - \E[\VhatCCbold] ||_F^2 ] & \leq \bigg(\frac{1}{K^4} \sum_{ \substack{(k,k'), (\ell,\ell') \\ \pi_{kk'}>0, \pi_{\ell \ell'}>0}}   \left| \Cov\bigg(\frac{W_k W_{k'}}{\pi_{kk'}} b_{kk'}, \frac{W_{\ell}W_{\ell'} }{\pi_{\ell \ell'}} b_{\ell\ell'} \bigg)\right|^{4/3}\bigg)^{3/4} 	\bigg(\frac{1}{K^4} \sum_{ \substack{(k,k'), (\ell,\ell') \\ \pi_{kk'}>0, \pi_{\ell \ell'}>0}} |\langle \mathbf{Y}_{k} \mathbf{Y}_{k'}', \mathbf{Y}_{\ell} \mathbf{Y}_{\ell'}' \rangle_F |^4 \bigg)^{1/4} \\
			&= o(s_K^2) \bigg(\frac{1}{K^4} \sum_{ \substack{(k,k'), (\ell,\ell') \\ \pi_{kk'}>0, \pi_{\ell \ell'}>0}} |\langle \mathbf{Y}_{k} \mathbf{Y}_{k'}', \mathbf{Y}_{\ell} \mathbf{Y}_{\ell'}' \rangle_F |^4 \bigg)^{1/4}.
	\end{align*}}
	However, by Cauchy-Schwarz, $|\langle \mathbf{Y}_{k} \mathbf{Y}_{k'}', \mathbf{Y}_{\ell} \mathbf{Y}_{\ell'}' \rangle_F| \leq ||\mathbf{Y}_{k} \mathbf{Y}_{k'}'||_F  ||\mathbf{Y}_{\ell} \mathbf{Y}_{\ell'}'||_F = ||\mathbf{Y}_k|| \cdot ||\mathbf{Y}_{k'}|| \cdot ||\mathbf{Y}_\ell|| \cdot ||\mathbf{Y}_{\ell'}|| $ and hence using (iii),
	\begin{align*}
		\bigg(\frac{1}{K^4} \sum_{ \substack{(k,k'), (\ell,\ell') \\ \pi_{kk'}>0, \pi_{\ell \ell'}>0}} |\langle \mathbf{Y}_{k} \mathbf{Y}_{k'}', \mathbf{Y}_{\ell} \mathbf{Y}_{\ell'}' \rangle_F |^4 \bigg)^{1/4}& \leq \bigg(\frac{1}{K^4} \sum_{ \substack{(k,k'), (\ell,\ell') \\ \pi_{kk'}>0, \pi_{\ell \ell'}>0}} ||\mathbf{Y}_k||^4 \cdot ||\mathbf{Y}_{k'}||^4 \cdot ||\mathbf{Y}_\ell||^4 \cdot ||\mathbf{Y}_{\ell'}||^4 \bigg)^{1/4} \\
		&\leq \bigg(\frac{1}{K^4} \sum_{ \substack{(k,k'),(\ell,\ell')}} ||\mathbf{Y}_k||^4 \cdot ||\mathbf{Y}_{k'}||^4 \cdot ||\mathbf{Y}_\ell||^4 \cdot ||\mathbf{Y}_{\ell'}||^4 \bigg)^{1/4} \\
		& = \frac{1}{K} \sum_{k}||\mathbf{Y}_k||^4   = O(1).
	\end{align*}
	Combining the previous two displays we see that $ \E[ || \VhatCCbold - \E[\VhatCCbold] ||_F^2 ] = o(s_K^2)$. From Chebyshev's inequality, it follows that $\VhatCCbold = \E[\VhatCCbold] + o_p(s_K)$. Under conditions (i) and (ii), we then see that 
	\[s_K^{-1} \VhatCCbold  = s_K^{-1} \E[\VhatCCbold]  + o_p(1) = s_K^{-1} (\Var(\boldsymbol{\hat\tau}) +  (\E[\VhatCCbold]- \Var(\boldsymbol{\hat\tau}))  ) + o_p(1) \xrightarrow{p} \mathbf{V}^* + \boldsymbol{\delta}, \]
	as desired. In the univariate case $Q=1$, we have $\frac{\VhatCC}{ \Var(\hat\tau) } = \frac{s_K^{-1}\VhatCC}{ s_K^{-1} \Var(\hat\tau) } \xrightarrow{p} \frac{V^* + \delta}{V^*} = 1 + \frac{\delta}{V^*}$, which yields the conclusion in \Cref{prop:vhat-consistency}.
\end{proof}

\paragraph{Proof of \Cref{lem:ols-consistency}}
\begin{proof}
	Let $Z_k := \pi_k X_k X_k'$ and $u_k := \pi_k X_k \epsilon_k$.  Note that $\hat{\Gamma}_K$ and $\hat{U}_K$ are respectively (matrix- and vector-valued-) Horvitz Thompson estimators for $\Gamma_K = \frac{1}{K} \sum_k Z_k$ and  $\frac{1}{K} \sum_k u_k$, and $\frac{1}{K} \sum_k u_k=0$ by the OLS orthogonality conditions. Condition (ii) implies that each element of $Z_k$ and $u_k$ has bounded fourth moment (and hence, bounded second moment), and thus applying \Cref{lem:tauhat-consistency} elementwise we obtain that $\hat{\Gamma}_K = \Gamma_K + O_p(\sqrt{s_K})$ and $\hat{U}_K = O_p(\sqrt{s_K})$. Since $\Gamma \succ 0$, the continuous mapping theorem implies $\hat{\Gamma}_K^{-1} \xrightarrow{p} \Gamma^{-1}$. From Condition (iv) and Slutsky's lemma, we then see that $s_K^{-1/2} a'(\hat\theta - \theta) \xrightarrow{d} a' \Gamma^{-1} L =: L^*$, which yields the first result of the lemma. 
	
	We next claim that $s_K^{-1} \hat\Omega_{\epsilon} \xrightarrow{p} \Omega + S_u^* \succeq \Omega$, where $S_u^*$ is the upper left block of $S^*_{(u,z)}$. Observe that $\hat\Omega_{\epsilon}$ corresponds to the CC variance estimator for the (vector-valued) HT estimator $\hat{U}_K$. Hence, by \Cref{prop:vhat-consistency-vector}, we have that $\hat{\Omega}_{\epsilon} = \E[\hat{\Omega}_{\epsilon}] + o_p(s_K) = \Var(\hat{U}_K) + \frac{C}{K^2} \Var_c(\tau_c^u) + o_p(s_K)$.  Using conditions (iv) and (v), we thus see that $s_K^{-1} \hat{\Omega}_{\epsilon} \xrightarrow{p} \Omega + S_{u}^* \succeq \Omega$, where $S_u^* = \lim s_K^{-1} \frac{C}{K^2} \Var_c(\tau_c^u)$ is the upper-left block of $S^*_{(u,z)}$.
	
	We next claim that $||\hat\Omega_{\epsilon} - \hat\Omega_{\hat\epsilon}||_{op} = o_p(s_K)$. Let $z_k = \mathrm{vec}(Z_k)$, and note that $Z_k (\hat\theta - \theta) = ( (\hat\theta - \theta)' \otimes I ) z_k$. We thus see that
	\begin{align*} \hat{\Omega}_{\hat\epsilon} &= (I, -(\hat\theta - \theta) '\otimes I ) \bigg( \frac{1}{K^2} \sum_{ \substack{ k,k' \\ \pi_{kk'} >0 }} \frac{W_k W_{k'}}{ \pi_{kk'} } b_{kk'} \twovec{u_k}{z_k} \twovec{u_{k'}}{z_{k'}}' \bigg) (I, -(\hat\theta - \theta)' \otimes I )' \\
		&=: C_{(\hat\theta-\theta)} \widehat{V}_{(u,z)} C_{(\hat\theta-\theta)}'
	\end{align*}
	Note that the inner term $\widehat{V}_{(u,z)}$ is the CC variance estimator corresponding to the Horvitz-Thompson estimator $\hat\tau_{(u,z)} = \frac{1}{K} \sum_k \frac{W_k}{\pi_k} (u_k',z_k')'$. Condition (ii) ensures that each element of $(u_k',z_k')'$ has finite fourth moment and thus by \Cref{lem:tauhat-consistency} each diagonal element of $\Var(\hat\tau_{(u,z)})$ is $O(s_K)$; since $\Var(\hat\tau_{(u,z)}) \succeq 0$, it follows that $||\Var(\hat\tau_{(u,z)})||_{op} \leq \tr(\Var(\hat\tau_{(u,z)})) = O(s_K) $. Additionally, condition (v) implies that $||\E[\widehat{V}_{(u,z)}] - \Var(\hat\tau_{(u,z)})||_{op} = O(s_K)$. Applying \Cref{prop:vhat-consistency-vector}, we see that $\widehat{V}_{(u,z)} = \E[\widehat{V}_{(u,z)} ] + o_p(s_K)$ and thus 	
	\[||\widehat{V}_{(u,z)}||_{op} \leq || \Var(\hat\tau_{(u,z)}) ||_{op} + ||\E[\widehat{V}_{(u,z)}] - \Var(\hat\tau_{(u,z)})||_{op}  + o_p(s_K) = O_p(s_K) .\]
	Further, observe that $\hat\Omega_{\epsilon} = C_{0} \widehat{V}_{(u,z)} C_{0}'$. We thus see that
	\begin{align*}
		|| \hat{\Omega}_{\hat\epsilon} -  \hat{\Omega}_{\epsilon} ||_{op} &= ||  C_{(\hat\theta-\theta)} \widehat{V}_{(u,z)} C_{(\hat\theta-\theta)}' -  C_{0} \widehat{V}_{(u,z)} C_{0}' ||_{op} \\
		&= \frac{1}{2}|| ( C_{(\hat\theta-\theta)} + C_0) \widehat{V}_{(u,z)} ( C_{(\hat\theta-\theta)} - C_0)' +  ( C_{(\hat\theta-\theta)} - C_0) \widehat{V}_{(u,z)} ( C_{(\hat\theta-\theta)} + C_0)'  ||_{op} \\
		& \leq ||( C_{(\hat\theta-\theta)} + C_0) ||_{op} \cdot ||\widehat{V}_{(u,z)}||_{op} \cdot ||( C_{(\hat\theta-\theta)} - C_0) ||_{op} \\
		&= O_p(1) O_p(s_K) O_p(\sqrt{s_K}) = O_p(s_K^{3/2}) = o_p(s_K)
	\end{align*} 
	where the last line uses the fact that $\hat\theta - \theta = O_p(\sqrt{s_K})$ and hence $||( C_{(\hat\theta-\theta)} + C_0) ||_{op} = O_p(1)$ and $||( C_{(\hat\theta-\theta)} - C_0) ||_{op} = O_p(\sqrt{s_K})$. 
	
	It follows from the previous two claims that $s_{K}^{-1} \hat\Omega_{\hat\epsilon} \xrightarrow{p} \Omega + S_u^* \succeq \Omega$. From Slutsky's lemma, we see that $s_{K}^{-1} \widehat{V}_{\hat\epsilon} =a' \hat{\Gamma}^{-1} s_K^{-1} \hat{\Omega}_{\hat\epsilon} \hat{\Gamma}^{-1} a \xrightarrow{p}a' \Gamma^{-1} (\Omega + S_u^*) \Gamma^{-1} a \geq a'\Gamma^{-1} \Omega \Gamma^{-1} a= \Var(L^*)$, as we wished to show.
	
\end{proof}

\paragraph{Proof of \Cref{cor:cequals1-trivial-sharp-null-exact}}
\begin{proof}
	\Cref{prop:sharp-null-exact} shows that $\VhatM$ is weakly upwards biased and sharp-null exact if and only if $T(M)$ satisfies \eqref{eqn:h-lambda-h}. However, when $C=1$, $\Lambda$ is equal to 1, and $H$ is the $K \times 1$ matrix of 1s. Thus, $T$ is equal to the $K \times K$ matrix of 1s. We thus see that $m_{kk'} = v_{kk'}+1 = \frac{\pi_{kk'}}{ \pi_{k} \pi_{k'} }$. It follows that
	\[\VhatM = \frac{1}{K^2} \sum_{ \substack{ k,k' \\ \pi_{kk'} >0 }} \frac{W_k W_{k'}}{\pi_{kk'}} m_{kk'} Y_k Y_{k'} = \frac{1}{K^2} \sum_{ \substack{ k,k' \\ \pi_{kk'} >0 }} \frac{W_k W_{k'}}{ \pi_{k} \pi_{k'} }  Y_k Y_{k'} =  \frac{1}{K^2} \sum_{ \substack{ k,k'  }} \frac{W_k W_{k'}}{ \pi_{k} \pi_{k'} }  Y_k Y_{k'} = \left( \frac{1}{K} \sum_k \frac{W_k}{\pi_k} Y_k \right)^2 = \hat\tau^2,\]
	where the first equality uses the definition of $\VhatM$, the second uses the form of $m_{kk'}$ derived above, the third uses the fact that $W_{k}W_{k'}=0$ with probability 1 if $\pi_{kk'}=0$, and the final equalities use algebra and the definition of $\hat\tau$. 
\end{proof}

\paragraph{Proof of \Cref{lem:expectation-vwithin}}

\begin{proof}
	The proof is very similar to that of \Cref{prop:cc-valid}. Let $b^{\tilde{\mathcal{K}}}_{kk'}= v_{kk'} + s_{kk'}^{\tilde{\mathcal{K}}}$. Note that if $\pi_{kk'}>0$, by standard Horvitz-Thompson logic, we have $\E[ \frac{W_k W_{k'}}{ \pi_{kk'} }]=1$, and hence
	\[ \E[\VhatKwithin] = \frac{1}{K^2} \sum_{ \substack{ k,k' \\ \pi_{kk'} >0 }} b_{kk'}^{\tilde{\mathcal{K}}} Y_k Y_{k'} .\]
	
	Observe that if $\pi_{kk'}=0$, then $v_{kk'} = \frac{ \pi_{kk'} - \pi_{k}\pi_{k'} }{ \pi_{k}\pi_{k'} } = -1$. By construction, if $\pi_{kk'} =0$ and $k,k' \in \tilde{\mathcal{K}}_{\tilde{c}}$ for some $\tilde{c}$, then $s_{kk'}^{\tilde{\mathcal{K}}}= 1$ and so $b_{kk'}^{\tilde{\mathcal{K}}}= 0$. Meanwhile, if $\pi_{kk'} =0$ and $k,k'$ are in distinct $\tilde{\mathcal{K}}_{\tilde{c}}, \tilde{\mathcal{K}}_{\tilde{c}'}$, then  $s_{kk'}^{\tilde{\mathcal{K}}}= -1/(\tilde{C}-1)$ and hence $b^{\tilde{\mathcal{K}}}_{kk'} = -1-1/(\tilde{C}-1)$. It follows that
		\[ \E[\VhatKwithin] = \frac{1}{K^2} \sum_{  k,k' } b_{kk'}^{\tilde{\mathcal{K}}} Y_k Y_{k'} + \bigg(1 + \frac{1}{\tilde{C}-1} \bigg) \overbracket{ \frac{1}{K^2} \sum_{\substack{ k,k' \\ \pi_{kk'} =0  \\ k \in \tilde{\mathcal{K}}_{\tilde{c}}, k' \in \tilde{\mathcal{K}}_{\tilde{c}'} \neq \tilde{\mathcal{K}}_{\tilde{c}} }}   Y_k Y_{k'} }^{=-V_{\text{across}}} \] 	
	Since $b_{kk'}^{\tilde{\mathcal{K}}} = v_{kk'} + s_{kk'}^{\tilde{\mathcal{K}}}$ by definition, we see that 
	\begin{align*}
		\E[ \VhatKwithin ] &= \underbracket{ \frac{1}{K^2} \sum_{k,k'} v_{kk'} Y_k Y_{k'} }_{= \Var(\hat\tau)}+ \frac{1}{K^2} \sum_{k,k'} s^{\tilde{\mathcal{K}}}_{kk'} Y_k Y_{k'} -\bigg(1 + \frac{1}{\tilde{C}-1} \bigg) V_{\text{across}} \\
		&= V_{\text{within}} - \frac{1}{\tilde{C}-1}V_{\text{across}} + \frac{1}{K^2} \sum_{k,k'} s^{\tilde{\mathcal{K}}}_{kk'} Y_k Y_{k'}.
	\end{align*} 
	It thus remains to show that $\frac{1}{K^2} \sum_{k,k'} s_{kk'}^{\tilde{\mathcal{K}}} Y_k Y_{k'} = \frac{\tilde{C}}{K^2} \Var_{\tilde{c}}(\tau_{\tilde{c}}).$ 
	Recall that $s_{kk'}^{\tilde{\mathcal{K}}}=1$ if $k,k' \in \tilde{\mathcal{K}}_{\tilde{c}}$ for some $c$ and $s_{kk'}^{\tilde{\mathcal{K}}} = -1/(\tilde{C}-1)$ otherwise. We thus have that
	\begin{align*}
		\sum_{k,k'} s_{kk'}^{\tilde{\mathcal{K}}} Y_k Y_{k'} &= \sum_{\tilde{c}} \sum_{ \substack{k,k' \\ k,k' \in \tilde{\mathcal{K}}_{\tilde{c}}} } Y_k Y_{k'} - \frac{1}{\tilde{C}-1} \sum_{\tilde{c},\tilde{c}' : \tilde{c} \neq \tilde{c}'} \sum_{ \substack{k,k' \\ k \in \tilde{\mathcal{K}}_{\tilde{c}}, k' \in \tilde{\mathcal{K}}_{\tilde{c}'} }} Y_k Y_{k'}  \\
		&= \frac{\tilde{C}}{\tilde{C}-1}  \sum_{\tilde{c}}  \overbracket{ \sum_{ \substack{k,k' \\ k,k' \in \tilde{\mathcal{K}}_{\tilde{c}}} } Y_k Y_{k'} }^{ \textcolor{blue}{=\tau_{\tilde{c}}^2}} - \frac{1}{\tilde{C}-1}  \overbracket{\sum_{ \substack{k,k'}} Y_k Y_{k'} }^{\textcolor{blue}{= (\sum_{\tilde{c}} \tau_{\tilde{c}})^2 = \tilde{C}^2 \bar{\tau}^2  }} \\
		&= \frac{\tilde{C}}{\tilde{C}-1} \sum_{\tilde{c}} \tau_{\tilde{c}}^2  - \frac{\tilde{C}^2}{\tilde{C}-1} \bar{\tau}^2 \\
		& = \tilde{C} \Var_{\tilde{c}}(\tau_{\tilde{c}}),
	\end{align*} 
	which completes the proof.
\end{proof}

\section{\texorpdfstring{Details on \citet{pollmann_causal_2026}}{Details on Pollmann (2026)}}
\label{sec:pollmann-details}

We now provide additional details on the setting in \citet{pollmann_causal_2026} described in \Cref{subsec:sims-pollmann}. 

\paragraph{Pollmann's estimator.} Pollmann considers the H\'ajek-style estimator $\hat\tau(d)$ given in \eqref{eqn:tauhat-d-pollmann}. Under appropriate regularity conditions, he shows that $\hat\tau(d)$ is consistent for the estimand $\tau_{\mathrm{marginal}}(d)$ given in \eqref{eqn:tau-marginal-def}.

\paragraph{First-order approximation.} Pollmann likewise shows that under appropriate regularity conditions, $\hat\tau$ is asymptotically equivalent to an estimator $\tilde\tau$, where
\[\tilde\tau(d) =
\tau_{\mathrm{marginal}}(d)
+
\frac{1}{|\mathcal{S}|}
\sum_{s\in\mathcal{S}}
\sum_{i}
\frac{w_i(s,d)}{\bar n(d)}
\left[
T_s\left\{Y_i-\mu_t(d)\right\}
-
(1-T_s)\frac{\pi_s}{1-\pi_s}
\left\{Y_i-\mu_c(d)\right\}
\right],
\]
where
\[
\bar n(d)
=
\frac{1}{|\mathcal{S}|}
\sum_{s\in\mathcal{S}}
\pi_s
\sum_{i\in\mathcal{I}}
w_i(s,d),
\] 
is the expectation of the weights, and $\mu_t(d),\mu_c(d)$ are the values obtained if we replace $Y_i(\mathbf{s}_i) -Y_i( \mathbf{s}_i \setminus \{s\} )$ in the definition of $\tau_{\mathrm{marginal}}(d)$ with $Y_i(\mathbf{s}_i)$ and $Y_i( \mathbf{s}_i \setminus \{s\} )$, respectively. The estimator $\tilde\tau$ can be viewed as a Horvitz-Thompson estimator over pairs indexed by $k=(i,s,\mathbf{s}_i)$. In particular, we have
\[\tilde\tau(d) - \tau_{\mathrm{marginal}}(d) = \frac{1}{K} \sum_k \frac{W_k}{ \pi_k } \psi_k,  \]
where $W_k = 1\{\mathbf{S}_i = \mathbf{s}_i\}$, $\pi_k = P(W_k =1)$, and for $a_{s,\mathbf{s}_i} = 1\{ s \in \mathbf{s}_i  \}$,
\begin{equation} \psi_k = \frac{ K \pi_k }{ |\mathcal{S}| } \frac{w_i(s,d)}{\bar n(d)}
\bigg[
a_{s, \mathbf{s}_i} \left\{Y_i( \mathbf{s}_i )-\mu_t(d)\right\}
-
(1-a_{s, \mathbf{s}_i})\frac{\pi_s}{1-\pi_s}
\left\{Y_i(\mathbf{s}_i)-\mu_c(d)\right\} \bigg],  \label{eqn:defn-psi-k}
\end{equation}
and the summation runs over the $K$ index values for which $\pi_k >0$ and $w_i(s,d)>0$. 

\paragraph{Pollmann's variance bound.} Because $\tilde\tau$ is a Horvitz-Thompson estimator, its variance is given by the expression in \eqref{eqn:var-tauhat} with $Y_k = \psi_k$. Pollmann obtains an upper bound on the terms $-\psi_k \psi_{k'}$ for pairs with $\pi_{kk'}=0$ using a weighted version of Young's inequality, similar to \citet{aronow_estimating_2017}. Specifically, define $\dot{Y}_k$ equal to $Y_i( \mathbf{s}_i )-\mu_t(d)$ if $a_{s,\mathbf{s}_i}=1$ and equal to $Y_i( \mathbf{s}_i )-\mu_c(d)$ if $a_{s,\mathbf{s}_i}=0.$  Let $\sigma_k=(2a_{s,\mathbf{s}_i} - 1) \in \{-1,1\}$. Then we can write $\psi_k = \sigma_k c_k \dot{Y}_k$, where $c_k > 0$ collects the absolute value of the appropriate constants in \eqref{eqn:defn-psi-k}. Recall the weighted version of Young's inequality: for any $\eta \neq 0$, $-2 ab = \eta^2 a^2 + 1/\eta^2 \cdot b^2 - (\eta a + 1/\eta \cdot b)^2 \leq  \eta^2 a^2 + 1/\eta^2 \cdot b^2$. Applying this inequality to $\psi_k = \sigma_k c_k \dot{Y}_k, \psi_{k'} = \sigma_{k'} c_{k'} \dot{Y}_{k'}$ with weight $\eta = \sqrt{c_{k'}/c_k}$, we obtain that
\[ - \psi_k \psi_{k'} = \frac{1}{2} c_k c_{k'} ( \dot{Y}_k^2 + \dot{Y}_{k'}^2 - (\sigma_k \dot{Y}_k + \sigma_{k'} \dot{Y}_{k'})^2 ) \leq \frac{1}{2} c_k c_{k'} ( \dot{Y}_k^2 + \dot{Y}_{k'}^2).\]
Pollmann considers the upper bound for $\Var(\tilde\tau)$ that replaces the $-\psi_k \psi_{k'}$ with $\pi_{kk'}=0$ with the upper bound given in the previous display. The resulting upper bound has conservativeness
{\footnotesize\[\frac{1}{2 K^2} \sum_{ \substack{ k,k' \\ \pi_{kk'} =0 } }  c_k c_{k'} (\sigma_k \dot{Y}_k + \sigma_{k'} \dot{Y}_{k'})^2 =  \frac{1}{2 K^2} \sum_{ \substack{ k,k' \\ \pi_{kk'} =0  \\ \sigma_k = \sigma_{k'}=1 }}  c_k c_{k'} (\dot{Y}_k + \dot{Y}_{k'})^2  +  \frac{1}{2 K^2} \sum_{ \substack{ k,k' \\ \pi_{kk'} =0  \\ \sigma_k = \sigma_{k'}=-1 }}  c_k c_{k'} (\dot{Y}_k + \dot{Y}_{k'})^2 + \frac{1}{2 K^2} \sum_{ \substack{ k,k' \\ \pi_{kk'} =0  \\ \sigma_k \neq \sigma_{k'} }}  c_k c_{k'} (\dot{Y}_k - \dot{Y}_{k'})^2 ,\]} where the last three terms correspond, respectively, to the terms defined as $\tilde{V}_{tt},\tilde{V}_{cc}, \tilde{V}_{ct}$ in Pollmann. Pollmann considers a plug-in analog to the Horvitz-Thompson estimator for the upper bound, which uses $\hat\psi_k$ instead of $\psi_k$, where $\hat\psi_k$ is the analog to $\psi_k$ that replaces $\mu_t(d),\mu_c(d)$ with estimators $\hat\mu_t(d),\hat\mu_c(d)$ obtained from the first and second terms in the definition of $\hat\tau(d)$.

\begin{table}[htbp]
	\centering
	\caption{Simulation results based on \citet{pollmann_causal_2026}: Estimating the variance of $\tilde\tau$ using true scores}
	\label{tab:pollmann-sims-all-bins-true}
	\begin{tabular}{lrrrr}
\toprule
Estimator & Bias & Precision & Coverage & Width \\
\midrule
\multicolumn{5}{l}{\textit{Panel A: Distance bin: $(0,0.025]$ miles}} \\
Pollmann & 5.0802 & 2.286 & 1.000 & 1.000 \\
CC & -0.0045 & 0.664 & 0.954 & 0.392 \\
\addlinespace
\multicolumn{5}{l}{\textit{Panel B: Distance bin: $(0.025,0.050]$ miles}} \\
Pollmann & 4.5928 & 2.151 & 1.000 & 1.000 \\
CC & -0.0087 & 0.829 & 0.947 & 0.402 \\
\addlinespace
\multicolumn{5}{l}{\textit{Panel C: Distance bin: $(0.050,0.075]$ miles}} \\
Pollmann & 8.2045 & 1.163 & 1.000 & 1.000 \\
CC & 0.0024 & 0.361 & 0.952 & 0.326 \\
\addlinespace
\multicolumn{5}{l}{\textit{Panel D: Distance bin: $(0.075,0.100]$ miles}} \\
Pollmann & 7.0222 & 2.372 & 1.000 & 1.000 \\
CC & 0.0024 & 1.124 & 0.949 & 0.330 \\
\bottomrule
\end{tabular}

	\begin{minipage}{0.85\textwidth}
		\vspace{0.5em}
		\footnotesize
		\textit{Notes:} This table shows results analogous to Panel A of \Cref{tab:pollmann-sims-1bin} for all four distance bins. 
	\end{minipage}
\end{table}

\begin{table}[htbp]
	\centering
	\caption{Simulation results based on \citet{pollmann_causal_2026}: Estimating the variance of $\hat\tau$ using plug-in and leave-out scores}
	\label{tab:pollmann-sims-all-bins-feasible}
	\begin{tabular}{lrrrr}
\toprule
Estimator & Bias & Precision & Coverage & Width \\
\midrule
\multicolumn{5}{l}{\textit{Panel A: Distance bin: $(0,0.025]$ miles}} \\
Pollmann (plug-in $\hat\mu$) & 4.8320 & 2.017 & 1.000 & 0.972 \\
Pollmann (leave-out $\hat\mu_{-c}$) & 5.5326 & 2.594 & 1.000 & 1.025 \\
CC (plug-in $\hat\mu$) & -0.1104 & 0.546 & 0.914 & 0.369 \\
CC (leave-out $\hat\mu_{-c}$) & 0.1140 & 0.765 & 0.938 & 0.409 \\
\addlinespace
\multicolumn{5}{l}{\textit{Panel B: Distance bin: $(0.025,0.050]$ miles}} \\
Pollmann (plug-in $\hat\mu$) & 4.6362 & 1.975 & 1.000 & 0.983 \\
Pollmann (leave-out $\hat\mu_{-c}$) & 5.1327 & 2.469 & 1.000 & 1.021 \\
CC (plug-in $\hat\mu$) & -0.0712 & 0.704 & 0.917 & 0.383 \\
CC (leave-out $\hat\mu_{-c}$) & 0.1392 & 0.976 & 0.937 & 0.419 \\
\addlinespace
\multicolumn{5}{l}{\textit{Panel C: Distance bin: $(0.050,0.075]$ miles}} \\
Pollmann (plug-in $\hat\mu$) & 7.9954 & 1.110 & 1.000 & 0.992 \\
Pollmann (leave-out $\hat\mu_{-c}$) & 8.6949 & 1.283 & 1.000 & 1.030 \\
CC (plug-in $\hat\mu$) & -0.0697 & 0.324 & 0.924 & 0.315 \\
CC (leave-out $\hat\mu_{-c}$) & 0.1228 & 0.451 & 0.943 & 0.345 \\
\addlinespace
\multicolumn{5}{l}{\textit{Panel D: Distance bin: $(0.075,0.100]$ miles}} \\
Pollmann (plug-in $\hat\mu$) & 7.4809 & 2.127 & 1.000 & 0.988 \\
Pollmann (leave-out $\hat\mu_{-c}$) & 8.1547 & 2.822 & 1.000 & 1.023 \\
CC (plug-in $\hat\mu$) & -0.0625 & 0.845 & 0.918 & 0.312 \\
CC (leave-out $\hat\mu_{-c}$) & 0.2241 & 1.383 & 0.939 & 0.348 \\
\bottomrule
\end{tabular}

	\begin{minipage}{0.85\textwidth}
		\vspace{0.5em}
		\footnotesize
			\textit{Notes:} This table shows results analogous to Panel B of \Cref{tab:pollmann-sims-1bin} for all four distance bins. 
	\end{minipage}
\end{table}

\section{Comparison to spatial HAC}\label{sec:hac-comparison-details}

This section provides a comparison of the CC estimator and spatial HAC estimators a la \citet{conley_gmm_1999}. The original formulation of spatial HAC is model-based, based on assumptions about the spatial decay of the dependence in scores across observations. A challenge to applying spatial HAC in practice is determining an appropriate choice of bandwidth and kernel, which govern the spatial decay. Recent work by \citet{wang_design-based_2025} has made progress by providing conditions under which particular choices of bandwidth and kernel can be justified from the design-based perspective. Specifically, they consider a setting with a spatial treatment similar to \Cref{example:spatial}, and assume that the site-level treatment assignments $T_{s}$ are independent across $s$, which implies that outcomes for locations are independent if they are sufficiently far apart. They then provide conditions under which spatial HAC with a particular choice of bandwidth (large enough to ensure independence for units outside of the bandwidth) and the rectangle kernel yields valid inference for a H\'ajek-style estimator from the design-based perspective. One of the key assumptions is a ``homophily'' assumption that restricts treatment effects to be locally positively correlated. 

In this section, we derive the expectation of the HAC variance estimator in settings where outcomes are independent for units that are sufficiently far apart, as occurs under Bernoulli assignment in \citet{wang_design-based_2025}. We find that the HAC estimator is valid if the kernel is positive-semidefinite (PSD), but otherwise requires restrictions on the potential outcomes to obtain validity, such as the homophily assumption in \citet{wang_design-based_2025}. We further show that when the kernel is PSD, the HAC estimator takes a similar form to the CC estimator: it reduces to a clustered variance estimator using clusters that are a coarsening of the components used by the CC estimator, with conservativeness that depends on treatment effect heterogeneity across these coarsened components. HAC thus takes a similar form to CC in the special case in which it is valid without restrictions on treatment effect heterogeneity. We further show that HAC-PD, which projects the HAC kernel into the positive semidefinite cone, is valid without restrictions on treatment effect heterogeneity in this setting. However, in contrast to the CC estimator, we show that HAC-PD is potentially conservative even under the sharp null. 

Overall, it appears to us that the CC estimator and HAC estimators are complementary. The CC estimator accommodates complex assignment schemes that do not satisfy the local dependence condition, and does not require any restrictions on potential outcomes for its validity. On the other hand, it does require that there be multiple connected components to the joint unobservability graph $(C>1)$, which is not required by HAC. The CC estimator also does not exploit homophily in settings where this is plausible. The two approaches thus have non-nested potential use-cases.

\subsection{Set-up}

Suppose we have units $i=1,\ldots,N$ and consider an estimator of the form
\[ \hat\tau = \tau_0 + \frac{1}{N} \sum_i \psi_i(Z) ,\]
where $\tau_0$ is a fixed constant, $\psi_i(\cdot)$ is a non-stochastic unit-level score function, and $Z$ is a random variable governing treatment assignments.  For example, in \citet{pollmann_causal_2026}'s spatial setting, he shows that a H\'ajek-style estimator is first-order equivalent to $\hat\tau$ above with $Z= (T_1,\ldots,T_S)$ representing which sites are assigned to treatment and 
\[
\psi_i(Z)
=
\frac{N}{|\mathcal S|}
\sum_{s\in\mathcal S}
\frac{w_i(s,d)}{\bar n(d)}
\left[
T_s\{Y_i(\mathbf S_i)-\mu_t(d)\}
-
(1-T_s)\frac{\pi_s}{1-\pi_s}
\{Y_i(\mathbf S_i)-\mu_c(d)\}
\right],
\]
with all terms defined in \Cref{sec:pollmann-details}. For simplicity, we will focus on variance estimation in which the scores $\psi_i(\cdot)$ are known, and consider whether variance estimators that use the true scores $\psi(\cdot)$ are upwards-biased in finite-sample. In practice, one will often use a plug-in estimator based on estimated scores, which can often be justified asymptotically, as sketched in \Cref{sec:non-ht-estimators}. To focus on the main ideas, we abstract from the complication of using plug-in scores and study the validity of the HAC estimator using the oracle scores. 

We will assume that units sufficiently far apart are known to have independent scores.
\begin{asm}[Local dependence] \label{asm:local-dependence}
	There is a known distance metric $d(i,j)$ and threshold $\bar{d}$ such that $d(i,j) > \bar{d}$ implies that $\psi_i(Z) \indep \psi_j(Z)$.
\end{asm}
\noindent \citet{wang_design-based_2025} observed that if sites are assigned to treatment independently, then \Cref{asm:local-dependence} will hold for $\bar{d}$ sufficiently large. In Pollmann's setting (\Cref{example:spatial}), if we assume that the site-level assignments $T_s$ are independent across $s$, and that $Y_i(\cdot)$ depends only on the treatment assignments for sites within distance $d_0$ of location $i$, then \Cref{asm:local-dependence} holds with $\bar{d} \geq 2 d_0$, since locations separated by distance $\bar{d}$ cannot be exposed to a common site. We note that independent assignment of $T_s$ is important for this argument: local dependence would fail in Pollmann's setting, for example, if the sites $s$ were assigned to treatment via a completely-randomized experiment rather than a Bernoulli experiment, since then all the site-level treatment indicators would be correlated.\footnote{Other assignment mechanisms may also induce complicated dependence between the site indicators $T_s$. For example, \citet{dell_development_2020} consider distributions of sites for factories that preserve various features, such as the distance between sites and the amount of sugar available to each factory.} 

We will then consider the spatial HAC variance estimator with bandwidth $\bar{d}$ and the rectangular kernel. Specifically, let $\KHAC$ be the $N \times N$ matrix with $K_{\mathrm{HAC}, ij}=1$ if $d(i,j) \leq \bar{d}$ and zero otherwise. We then define 
\[\widehat{V}_{\mathrm{HAC}} = \frac{1}{N^2} \psi' \KHAC \psi, \qquad \text{where} \qquad \psi = (\psi_1(Z),\ldots,\psi_N(Z))'. \]

So far we have defined $\VhatHAC$ in the usual way using unit-level scores. To make the connection to the CC estimator more explicit, we can cast $\VhatHAC$ as a variance estimator for a Horvitz-Thompson estimator over potential outcomes. Specifically, suppose that $\psi_i(Z)$ can take on a finite number, $Q_i$, of possible values $\psi_{i1},\ldots,\psi_{iQ_i}$ with positive probability. In Pollmann's setting, this holds with $Q_i$ equal to the number of possible values for $\mathbf{S}_i$, and $\psi_{iq}$ corresponding to the realization of $\psi_i(Z)$ for a given value of $\mathbf{S}_i$. In this case, we can write $\hat\tau - \tau_0$ as a Horvitz-Thompson estimator: letting $k=(i,q)$ index the $K= \sum_i Q_i$ possible score values, and letting $W_k = 1\{\mathbf{S}_i = \mathbf{s}_{iq}\}$ and $\pi_k = P(W_k=1)$, we have that
\[\hat\tau - \tau_0 = \frac{1}{K} \sum_{k} \frac{W_k}{\pi_k} Y_k, \qquad \text{ where } \qquad Y_k = \frac{K}{N} \pi_k \psi_k. \]
We can thus view $\VhatHAC$ as a variance estimator for the Horvitz-Thompson estimator given in the previous display and compare to the CC variance estimator applied to the same estimator.

\subsection{Expectation of HAC}

We now derive the expectation of $\VhatHAC$. Let $m=E[\psi]$ and $\Gamma=\Var(\psi)$ be the mean and variance of the score vector $\psi$. Likewise, let $H_U$ be the $K \times N$ matrix such that $H_{U,ki}=1$ if $k$ corresponds to a score for unit $i$ and zero otherwise. 

\begin{prop}
	\label{prop:hac-score-slack}
	Under \Cref{asm:local-dependence}, 
	\[
	\E[ \VhatHAC ]-\Var(\hat\tau)
	= \frac{1}{N^2}m'  \KHAC m = \frac{1}{K^2} Y' H_U \KHAC H_U' Y.
	\]
\end{prop}

\begin{proof}
	Let $u = \psi - m$. Since $\hat\tau-\tau_0 = \frac{1}{N} \mathbf{1}' \psi$, we see that $\Var(\hat\tau) = \frac{1}{N^2} \mathbf{1}' \Gamma \mathbf{1} = \frac{1}{N^2} \tr(\Gamma \mathbf{1}\mathbf{1}')$. Writing $\psi = m +u$, we have that $$ N^2 E[ \VhatHAC ] = E[u'\KHAC u] + 2 E[u]' \KHAC m + m' \KHAC m.$$
	By construction $E[u]=0$, so the middle term drops out. Additionally, we can write $u'\KHAC u = \tr(uu'\KHAC)$, and hence $E[u'\KHAC u] = \tr(E[uu'] \KHAC) = \tr(\Gamma \KHAC)$. Hence, 
	$E[ \VhatHAC ] = \frac{1}{N^2} \tr(\Gamma \KHAC)  + \frac{1}{N^2} m' \KHAC m$. We thus see that
	\begin{equation} E[ \VhatHAC ] - \Var(\hat\tau) = \frac{1}{N^2} \left( \tr(\Gamma \KHAC) - \tr(\Gamma \mathbf{1}\mathbf{1}') \right) + \frac{1}{N^2} m' \KHAC m. \label{eqn:evhac}
	\end{equation}
	Using the definition of the trace, the term in parentheses can be written as $\sum_{i,j} (\KHAC[,ij]-1)\Gamma_{ji}$. However, if $d(i,j) \leq \bar{d}$, then $\KHAC[,ij]=1$, while if $d(i,j)>\bar{d}$, then $\Gamma_{ji} =0$ by \Cref{asm:local-dependence}, and hence $\sum_{i,j} (\KHAC[,ij]-1)\Gamma_{ji} = 0$. The first equality then follows from the previous display. The second equality follows from noting that $ \tfrac{1}{N} m_i = \frac{1}{N} \sum_q \pi_{iq} \psi_{iq} = \frac{1}{K} \sum_{q} \tfrac{K}{N} \pi_{iq} \psi_{iq} = \frac{1}{K}(H_U'Y)_i$, and thus $\frac{1}{N} m = \frac{1}{K} H_U' Y$.  
\end{proof}

It is immediate from \Cref{prop:hac-score-slack} that $\VhatHAC$ is upwards-biased for $\Var(\hat\tau)$ if the kernel matrix $\KHAC$ is PSD. Note that $\KHAC$ is simply a function of the orientation of the locations in space and the threshold $\bar{d}$. As we will show below, the matrix $\KHAC$ is not guaranteed to be PSD for all orientations of locations in space. If it is not PSD, then \Cref{prop:hac-score-slack} implies that there are some choices of $Y$ for which $\VhatHAC$ is downward-biased. To restore the validity of $\VhatHAC$ in such cases, one therefore needs to impose some restrictions on $Y$ that guarantee that $ Y' H_U \KHAC H_U' Y \geq 0$. \citet{wang_design-based_2025} propose a homophily assumption that is sufficient for this purpose for the H\'ajek estimator they consider.

\subsection{When is HAC upwards-biased without restricting potential outcomes}
We showed above that $\VhatHAC$ is (weakly) upwards-biased without restrictions on the potential outcomes if and only if $\KHAC$ is PSD. We now characterize when $\KHAC$ is PSD. Let $\AHAC$ be the adjacency matrix corresponding to the graph among units where distinct units $i$ and $j$ are adjacent if and only if $d(i,j) \leq \bar{d}$. Observe that $\KHAC = \AHAC + I$. Our next result shows that $\KHAC$ is PSD if and only if the graph corresponding to $\AHAC$ is a union of cliques. 

\begin{lem}
	\label{lem:kh-psd-cliques}
	The matrix $\KHAC$ is positive semidefinite if and only if every connected
	component of the graph with adjacency matrix $\AHAC$ is a clique: that is, if and only if there exists a partition of the units $i=1,\ldots,N$, say $\tilde{\mathcal{J}}_{1},\ldots,\tilde{\mathcal{J}}_{\tilde{C}}$, such that $d(i,j) \leq \bar{d}$ if and only if $i,j \in \tilde{\mathcal{J}}_{\tilde{c}}$ for some $\tilde{c}$. 
\end{lem}

\begin{proof}
	If every connected component of the graph corresponding to $\AHAC$ is a clique, then, after ordering
	units by component, $\KHAC$ is block diagonal with each block equal to an
	all-ones matrix. Hence $\KHAC \succeq0$.
	
	Conversely, suppose $\KHAC\succeq0$. If some connected component is
	not a clique, choose two nonadjacent vertices in that component and take a
	shortest path between them. The first three vertices on this path, say
	$i,j,k$, must satisfy $A_{ij} = 1$ and $A_{jk}=1$, but $A_{ik} = 0$ (since otherwise one could obtain a shorter path by removing $j$). The principal
	submatrix of $\KHAC$ on these three vertices is
	\[
	\begin{pmatrix}
		1 & 1 & 0 \\
		1 & 1 & 1 \\
		0 & 1 & 1
	\end{pmatrix},
	\]
	which has quadratic form equal to $-1$ at $(1,-1,1)'$. This contradicts
	positive semidefiniteness. Therefore every connected component must be a
	clique.
\end{proof}

Intuitively, \Cref{lem:kh-psd-cliques} shows that $\KHAC$ is positive semi-definite if and only if we can group the units into components (denoted $\tilde{\mathcal{J}}_{\tilde{c}}$) such that all units within a component are within $\bar{d}$ distance from each other, and all cross-component pairs are at distance greater than $\bar{d}$. It is instructive to relate these components to the connected components of the joint unobservability graph. Note that under \Cref{asm:local-dependence}, $\psi_i(Z) \indep \psi_j(Z)$ if $d(i,j) > \bar{d}$. It follows that if $d(i,j) > \bar{d}$, then all scores $\psi_{iq}$ and $\psi_{jq'}$ are jointly observable. Hence, $i$ and $j$ can have a pair of jointly unobservable scores only if $d(i,j) \leq \bar{d}$. It follows that the scores for units $i$ and $j$ can be in the same connected component of the joint unobservability graph only if $i$ and $j$ are in the same connected component of the graph corresponding to $\AHAC$. Put otherwise, if we let $\tilde{\mathcal{K}}_{\tilde{c}}$ collect all the indices $k=(i,q)$ with $i \in \tilde{\mathcal{J}}_{\tilde{c}}$, then the partition $\tilde{\mathcal{K}}_1,\ldots,\tilde{\mathcal{K}}_{\tilde{C}}$ is a coarsening of the connected components $\mathcal{K}_1,\ldots,\mathcal{K}_C$ of the joint unobservability graph. Applied to the setting in Pollmann: two locations can have jointly unobservable outcomes only if the distance between them is no more than $\bar{d} = 2d_0$. But it need not be the case that all locations within distance $2 d_0$ of each other have a common site, and therefore such locations need not necessarily have jointly unobservable outcomes. The partitions $\tilde{\mathcal{K}}_1,\ldots,\tilde{\mathcal{K}}_{\tilde{C}}$  collect groups of locations within $\bar{d}$ distance of each other, which may combine multiple connected components of the joint unobservability graph. The two partitions coincide if every pair of locations within distance $\bar{d} = 2 d_0$ of each other share a common site. 

The characterization in \Cref{lem:kh-psd-cliques} implies that $\VhatHAC$ takes a simple form in settings where it is weakly upwards-biased without restrictions on the potential outcomes: it corresponds to a clustered variance estimator using the $\tilde{\mathcal{K}}_{\tilde{c}}$ as clusters. Its conservativeness is then related to the variance of treatment effects across the $\tilde{\mathcal{K}}_{\tilde{c}}$. 

\begin{cor} \label{cor:conservativeness-hac-when-valid}
	Suppose that \Cref{asm:local-dependence} holds and $\KHAC \succeq 0$. Then 
	\[ \VhatHAC = \frac{1}{K^2} \sum_{\tilde c}  \bigg( \sum_{ k \in \tilde{\mathcal{K}}_{\tilde{c}} } W_k \underbracket{\tfrac{K}{N} \psi_k}_{= \frac{Y_k}{\pi_k}} \bigg)^2   \]
	and, defining $\tau_{\tilde{c}} =  \sum_{ k \in \tilde{\mathcal{K}}_{\tilde{c}} } Y_k$ and $\bar{\tau} = \frac{1}{\tilde{C}} \sum_{\tilde{c}} \tau_{\tilde{c}}$,
	\[\E[ \VhatHAC ]  - \Var(\hat\tau) = \frac{\tilde{C} - 1}{K^2} \Var_{\tilde c}( \tau_{\tilde{c}} ) +  \frac{\tilde{C}}{K^2}  \cdot \bar\tau^2 \geq 0. \]
\end{cor}

\begin{proof}
	By \Cref{lem:kh-psd-cliques}, if $\KHAC \succeq 0$, then $\AHAC$ is the adjacency matrix corresponding to a graph where all connected components are cliques. Ordering the units by their connected components, we then see that $\AHAC + I$ is block-diagonal with each block equal to the all-ones matrix. The form for $\VhatHAC$ then follows immediately from the definition of $\VhatHAC = 1/N^2 \psi' \KHAC \psi$. For the second result, \Cref{prop:hac-score-slack} gives that $\E[ \VhatHAC ]-\Var(\hat\tau) = \frac{1}{K^2} Y' H_U \KHAC H_U' Y$. Using the all-ones block-diagonal structure of $\KHAC$, we see that
	\[\E[ \VhatHAC ]-\Var(\hat\tau) = \frac{1}{K^2} \sum_{\tilde{c}} \tau_{\tilde{c}}^2 = \frac{\tilde{C} - 1}{K^2} \Var_{\tilde c}( \tau_{\tilde{c}} ) +  \frac{\tilde{C}}{K^2} \cdot \bar\tau^2, \]
	as desired.
	
\end{proof}

\Cref{cor:conservativeness-hac-when-valid} shows that when $\VhatHAC$ is valid without restrictions on the potential outcomes, its conservativeness depends on two terms: the first is proportional to the heterogeneity across the coarsened connected components, $\Var_{\tilde c}( \tau_{\tilde{c}})$, while the second is proportional to $\bar\tau^2$. In some cases the structure of the scores will imply that $\sum_k Y_k = 0$, as with the example from Pollmann given above, in which case $\bar\tau = 0$. In that case, the conservativeness of $\VhatHAC$ is similar to that for the CC estimator, except that it uses heterogeneity across the coarsened partition $\tilde{\mathcal{K}}_c$ rather than $\mathcal{K}_c$.\footnote{If the two partitions happen to coincide, then we have that $\E[ \VhatHAC ]  - \Var(\hat\tau) = \frac{C - 1}{K^2} \Var_{c}( \tau_c )$, and so the slack of the two estimators differ only by the degrees-of-freedom correction $\frac{C-1}{C}$.} We note that the CC estimator is guaranteed to be weakly upwards-biased for $\Var(\hat{\tau})$ without restrictions on $Y$, while this is only true for $\VhatHAC$ when the conditions of \Cref{lem:kh-psd-cliques} are satisfied.

\subsection{Expectation of HAC-PD}
Our results above showed that $\VhatHAC$ is weakly upwards-biased without restrictions on the potential outcomes only if $\KHAC \succeq 0$. When $\KHAC \not\succeq 0$, one can alternatively consider a variance estimator based on the positive-semidefinite projection of $\KHAC$. Let $\KHACPD$ be the matrix obtained by setting all negative eigenvalues of $\KHAC$ equal to zero (the Frobenius-projection into the PSD cone), and let $\VhatHACPD$ be the corresponding variance estimator that replaces $\KHAC$ with $\KHACPD$. 

We now show that $\VhatHACPD$ is weakly upwards-biased under \Cref{asm:local-dependence} without restrictions on the potential outcomes, although unlike the CC estimator, it may be conservative even under the sharp null. \citet{wang_design-based_2025} consider the estimator $\VhatHACPD$ and note that it is always weakly larger than $\VhatHAC$, but they do not formally show that it is valid without the homophily assumption that they impose for $\VhatHAC$. The following result thus appears to provide a new design-based justification for $\VhatHACPD$, without restrictions on the potential outcomes. 

\begin{prop}
	\label{prop:hac-pd-score-slack}
	Under \Cref{asm:local-dependence}, 
	\begin{align*}
		\E[ \VhatHACPD ]-\Var(\hat\tau)
		& = \frac{1}{N^2}m'  \KHACPD m + \frac{1}{N^2} \tr(\Gamma (\KHACPD - \KHAC)  )\\
		& = \frac{1}{K^2} Y' H_U \KHACPD H_U' Y + \frac{1}{N^2} \tr(\Gamma (\KHACPD - \KHAC)  ) \geq 0.
	\end{align*}
	
\end{prop}

\Cref{prop:hac-pd-score-slack} shows that under \Cref{asm:local-dependence}, $\VhatHACPD$ is weakly upwards biased for $\Var(\hat\tau)$ without a restriction on the potential outcomes. Its conservativeness depends on two terms: the first, $\frac{1}{N^2}m'  \KHACPD m$ is analogous to the slack of $\VhatHAC$, except replacing the (possibly indefinite) kernel matrix $\KHAC$ with its PSD projection $\KHACPD$. The second term, $ \tfrac{1}{N^2}\tr(\Gamma (\KHACPD - \KHAC)  )$, was not present in the decomposition for $\VhatHAC$ and intuitively relates to the conservativeness induced by using the PSD projection $\KHACPD$ instead of $\KHAC$. 

If $\KHAC$ is not PSD, then this latter conservativeness term can generically be non-zero even under the sharp null, and thus $\VhatHACPD$ will generally not be sharp-null exact. To see why, let $\lambda_1,\ldots,\lambda_N$ be the eigenvalues of $\KHAC$ and $v_1,\ldots,v_N$ the corresponding eigenvectors. Suppose, without loss of generality, that the first $r^* \geq 1$ eigenvalues are strictly negative. Then we can write $\KHACPD - \KHAC = \sum_{r=1}^{r^*} |\lambda_r| v_r v_r'$. It follows that \[ \tr\bigg( \Var(\psi) (\KHACPD - \KHAC)  \bigg) =  \tr\bigg( \Var(\psi) \sum_{r=1}^{r^*} |\lambda_r| v_r v_r'  \bigg) = \tr\bigg( \sum_{r=1}^{r^*} |\lambda_r| v_r' \Var(\psi) v_r  \bigg) = \sum_{r=1}^{r^*} |\lambda_r| \Var(v_r' \psi). \]
We thus see that $\tr( \Var(\psi) (\KHACPD - \KHAC)  ) > 0$ if $\Var(v_r' \psi) >0$ for any eigenvector $v_r$ corresponding to a negative eigenvalue of $\KHAC$. In the example from Pollmann described above, it is straightforward to construct potential outcomes satisfying the sharp null such that $\Var(v_r' \psi) >0$ for some $r$ corresponding to a negative eigenvalue, so long as the weights $w_i(s,d)$ satisfy the mild regularity conditions that every location receives nonzero weight from some site (for all $i$, $w_i(s,d)>0$ for some $s$) and every site has some far-away location that receives zero weight from that site (for all $s$, $w_i(s,d)=0$ for some $i$).

\paragraph{Proof of \Cref{prop:hac-pd-score-slack}}
\begin{proof}
    By an analogous derivation as for \eqref{eqn:evhac}, we have that
\begin{align*}
   E[ \VhatHACPD ] - Var(\hat\tau) &= \frac{1}{N^2} \left( \mathrm{tr}(\Gamma \KHACPD) - \mathrm{tr}(\Gamma \mathbf{1}\mathbf{1}') \right) + \frac{1}{N^2} m' \KHACPD m.
\end{align*}
Using the linearity of trace, we then see that
\begin{align*}
 E[ \VhatHACPD ] - Var(\hat\tau)  &= \frac{1}{N^2}(\mathrm{tr}(\Gamma (\KHACPD - \KHAC))) +  \frac{1}{N^2} \big( \underbracket{\mathrm{tr}(\Gamma \KHAC) - \mathrm{tr}(\Gamma \mathbf{1}\mathbf{1}')}_{=0, \text{ as shown above}} \big) + \frac{1}{N^2} m' \KHACPD m \\
 &=\frac{1}{N^2}(\mathrm{tr}(\Gamma (\KHACPD - \KHAC))) +  \frac{1}{N^2} m' \KHACPD m.
\end{align*} 
It remains only to show that this expression is weakly positive. Note that $\Gamma$ and $\KHACPD - \KHAC$ are both PSD by construction, and so $\mathrm{tr}(\Gamma (\KHACPD - \KHAC)) \geq 0$. Likewise, $\KHACPD$ is PSD, and so $m' \KHACPD m \geq 0$, which completes the proof.
\end{proof}

\subsection{Comparison in simulations}

We now revisit the simulation designs calibrated to \citet{pollmann_causal_2026} considered in \Cref{sec:sims-and-apps}. For simplicity, we focus on the CC and HAC estimators applied to the oracle scores, as in \Cref{tab:pollmann-sims-all-bins-true}, and focus on the $(0.075,.1]$ bin. Since the simulation design uses Bernoulli assignment of the site-level treatments $T_s$, the local dependence condition in \Cref{asm:local-dependence} holds using $\bar{d} = 2d_0$, and we thus consider the HAC estimator with rectangular kernel and bandwidth $2 d_0$. From \Cref{prop:hac-score-slack}, the bias of the $\VhatHAC$ estimator is proportional to $m' \KHAC m$. In this setting, $\KHAC$ is not PSD, and so whether $\VhatHAC$ is upwards or downwards-biased depends on the potential outcomes (which determine $m$). It is straightforward to verify that under the sharp null, $m=0$, and thus we expect $\VhatHAC$ to be unbiased. We then also consider two forms of treatment effect heterogeneity: first, we consider a simple form of heterogeneity in which each grocery store within 0.1 miles of a restaurant increases their outcome by 10 percent of the baseline potential outcome (additively across stores). This induces heterogeneity in the level treatment effect since different restaurants have different baseline potential outcomes. Second, we consider a specification of treatment effect heterogeneity designed
to violate the homophily assumption of \citet{wang_design-based_2025}.
Specifically, we start with the potential outcomes used in our sharp-null
simulation and then perturb them in the direction of the eigenvector
corresponding to the most negative eigenvalue of the HAC bias matrix. Let
$Q_{PO}$ be the matrix mapping the vector of potential outcomes
$\mathbf{Y}(\cdot)$ to the transformed scores $Y$, so that $Y = Q_{PO}\mathbf{Y}(\cdot)$. The bias of HAC is then proportional to $\mathbf{Y}(\cdot)'Q_{PO}' H_U \KHAC H_U' Q_{PO} \mathbf{Y}(\cdot)$. Let $\mathbf{y}_0$ denote the vector of potential outcomes under the sharp null
and let $v_1$ denote the eigenvector corresponding to the minimal eigenvalue of $Q_{PO}' H_U \KHAC H_U' Q_{PO}$. We consider potential outcomes given by $\mathbf{y}_0 + \epsilon v_1$, where we choose $\epsilon$ such that the variance of treatment effects constitutes 10\% of the variance of the outcome under the sharp null.\footnote{Specifically, we define the variance of treatment effects as the variance of $Y_i(\mathbf{s}_i) - Y_i(\emptyset)$ across all combinations $(i,\mathbf{s}_i)$ with $\mathbf{s}_i \neq \emptyset$.} This can be viewed as an ``adversarial'' form of treatment effect heterogeneity, since it perturbs the baseline potential outcomes in a direction that induces downward bias in $\VhatHAC$, subject to the constraint that the total amount of treatment effect heterogeneity not be too large. The results can be viewed as illustrating how badly HAC can perform under unrestricted treatment effect heterogeneity of bounded magnitude. 

\Cref{tab:pollmann-sims-hac} shows the results. Under the sharp null (Panel A), both the CC estimator and the HAC estimator are close to unbiased and have similar length, as expected. The estimators also have similar performance under the simple form of treatment effect heterogeneity. However, under the adversarial form of heterogeneity shown in Panel C, the HAC estimator is severely downwards-biased, with relative bias of about -50\% and coverage of a nominal 95\% interval of only 63\%, while the CC estimator is close to unbiased, with coverage close to the nominal level (94\%).\footnote{We truncate the HAC SE at zero when the HAC variance estimate is negative.} The HAC-PD estimator $\VhatHACPD$, which uses the PSD projection of the HAC kernel, is upwards-biased in all three simulation designs, leading to confidence intervals with moderately longer length than CC in all specifications.

\begin{table}[htbp]
	\centering
	\caption{Additional Simulation results based on \citet{pollmann_causal_2026}: Comparing CC, HAC, and Pollmann estimators}
	\label{tab:pollmann-sims-hac}
	\begin{tabular}{lrrrr}
\toprule
Estimator & Bias & Precision & Coverage & Width \\
\midrule
\multicolumn{5}{l}{\textit{Panel A: Sharp Null}} \\
Pollmann & 7.0222 & 2.372 & 1.000 & 1.000 \\
CC & 0.0024 & 1.124 & 0.949 & 0.330 \\
HAC & 0.0019 & 1.121 & 0.953 & 0.330 \\
HAC-PD & 0.1062 & 1.177 & 0.967 & 0.349 \\
\addlinespace
\multicolumn{5}{l}{\textit{Panel B: Simple Treatment Effect Het}} \\
Pollmann & 8.4234 & 2.415 & 1.000 & 1.000 \\
CC & 0.0090 & 0.997 & 0.956 & 0.308 \\
HAC & 0.0093 & 0.991 & 0.959 & 0.309 \\
HAC-PD & 0.1244 & 1.040 & 0.973 & 0.328 \\
\addlinespace
\multicolumn{5}{l}{\textit{Panel C: Adversarial Treatment Effect Het}} \\
Pollmann & 7.7843 & 3.171 & 1.000 & 1.000 \\
CC & 0.0582 & 1.319 & 0.942 & 0.310 \\
HAC & -0.5037 & 0.835 & 0.634 & 0.184 \\
HAC-PD & 0.3764 & 1.865 & 0.969 & 0.349 \\
\bottomrule
\end{tabular}

\end{table}

\end{document}